\documentclass[a4paper,11pt]{article}

\usepackage[a4paper]{geometry}
\usepackage{color}
\usepackage{graphicx}
\usepackage{amsmath}
\usepackage{amsfonts}
\usepackage{amsthm}
\usepackage{amscd}
\usepackage[mathcal]{euscript}
\usepackage{epsfig}
\usepackage{amssymb}
\usepackage{slashed}
\usepackage{tabularx}
\usepackage{calligra}
\usepackage{enumerate}
\usepackage{hyperref}
\usepackage{float}
\usepackage{stackrel}
\usepackage[T1]{fontenc}
\usepackage{longtable}
\usepackage{amsthm}
\usepackage{mathrsfs}
\usepackage{verbatim} % for comments
\usepackage{fancyhdr}
\usepackage{tikz}
\usepackage{tikz-cd}
\usepackage{slashed}
\usetikzlibrary{matrix}
\usetikzlibrary{positioning}
\usepackage[all]{xy}
\usepackage{units}
\usepackage{textcomp}
\usepackage{turnstile}
\usepackage{vmargin}
\usepackage{anysize}
\usepackage{ stmaryrd }

\usepackage[font=small]{caption}

\usetikzlibrary{calc,fadings,decorations.pathreplacing}
\newcommand\pgfmathsinandcos[3]{%
  \pgfmathsetmacro#1{sin(#3)}%
  \pgfmathsetmacro#2{cos(#3)}%
}

\newcommand\LatitudePlane[3][current plane]{%
  \pgfmathsinandcos\sinEl\cosEl{#2} % elevation
  \pgfmathsinandcos\sint\cost{#3} % latitude
  \pgfmathsetmacro\yshift{\cosEl*\sint}
  \tikzset{#1/.style={cm={\cost,0,0,\cost*\sinEl,(0,\yshift)}}} %
}

\newcommand\DrawLatitudeCircle[2][1]{
  \LatitudePlane{\angEl}{#2}
  \tikzset{current plane/.prefix style={scale=#1}}
  \pgfmathsetmacro\sinVis{sin(#2)/cos(#2)*sin(\angEl)/cos(\angEl)}
  \pgfmathsetmacro\angVis{asin(min(1,max(\sinVis,-1)))}
  \draw[current plane] (\angVis:1) arc (\angVis:-\angVis-180:1);
  \draw[current plane,dashed] (180-\angVis:1) arc (180-\angVis:\angVis:1);
}

\numberwithin{equation}{section}

\setmargins{2,5cm}       % left margin
{1,5cm}                        %  top margin
{16cm}                      %  text widht
{23,42cm}                    % text height
{10pt}                           % header height
{1cm}                           % space between text and header
{0pt}                             % footer height
{2cm}                           %  space between text and footer

\theoremstyle{plain}
\newtheorem{theorem}{Theorem}[section]
\newtheorem{proposition}[theorem]{Proposition}
\newtheorem{lemma}[theorem]{Lemma}

\theoremstyle{definition}
\newtheorem{definition}[theorem]{Definition}
\newtheorem{remark}[theorem]{Remark}

\newcommand{\C}{\mathcal{C}}

\renewcommand{\L}{\Lambda}
\renewcommand{\k}{\kappa}

\renewcommand{\t}{\tau}

\newcommand{\R}{\mathbb{R}}
\newcommand{\mr}{\mathring}
\newcommand{\tC}{\tilde{\mathcal{C}}}
\newcommand{\Cym}{\mathfrak{C}}
\DeclareMathOperator{\tho}{\text{\rm\th}}
\DeclareMathOperator{\edt}{\text{\rm\dh}}

\begin{document}

\title{Conformal geometry and half-integrable spacetimes}

\author{Bernardo Araneda\footnote{Email: \texttt{bernardo.araneda@aei.mpg.de}} \\
Max-Planck-Institut f\"ur Gravitationsphysik \\ 
(Albert Einstein Institut), Am M\"uhlenberg 1, \\
D-14476 Potsdam, Germany}

\date{October 12, 2021}

\maketitle

\begin{abstract}
Using a combination of techniques from conformal and complex geometry, we show the
potentialization of 4-dimensional closed Einstein-Weyl structures which are half-algebraically special 
and admit a ``half-integrable'' almost-complex structure.
That is, we reduce the Einstein-Weyl equations to a single, conformally invariant, non-linear scalar equation,
that we call the ``conformal HH equation'', and we reconstruct the conformal structure 
(curvature and metric) from a solution to this equation.
We show that the conformal metric is composed of: a conformally flat part, a conformally half-flat part related 
to certain ``constants'' of integration, and a potential part that encodes the full non-linear curvature,
and that coincides in form with the Hertz potential from perturbation theory.
We also study the potentialization of the
Dirac-Weyl, Maxwell (with and without sources), and Yang-Mills systems.
We show how to deal with the ordinary Einstein equations by using a simple trick.
Our results give a conformally invariant, coordinate-free, generalization 
of the hyper-heavenly construction of Pleba\'nski and collaborators.
\end{abstract}

\section{Introduction}

One of the main connections between the notions of integrability in complex geometry and in 
differential equations is given by Penrose's twistor theory \cite{Penrose1976, MW},
which relates solutions to the self-dual (SD) Einstein vacuum equations 
to integrable almost-complex structures on twistor space. 
While in this approach the complex structure is defined on the latter space, 
one can also understand the Ricci-flat, SD conditions in terms of (orthogonal) 
complex structures on the 4-dimensional (4d) manifold, via Pleba\'nski's `heavenly equations' \cite{Plebanski1975}.
These equations encode the SD Einstein equations into a scalar PDE for a {\em potential}, 
in terms of which the metric and curvature can be fully reconstructed;
this approach then leads to a ``potentialization''.

The Ricci-flat, SD equations are equivalent to the hyperk\"ahler condition (see e.g. \cite{Dunajski}): 
there is a 2-sphere of orthogonal complex structures with respect to which the metric is K\"ahler.
Pleba\'nski's work \cite{Plebanski1975} involves complex 4-spaces, so one should talk about
a {\em complex hyperk\"ahler} structure, i.e. the ``complex structures'' are complex-valued maps\footnote{A complex 
hyperk\"ahler structure on a complex manifold is defined by three endomorphisms of the tangent bundle that satisfy
the quaternion algebra and that are parallel w.r.t the Levi-Civita connection (see \cite{MW}).}. 

Pleba\'nski's first heavenly equation translates the Ricci-flat condition 
to a non-linear scalar PDE (a Monge-Amp\`ere equation) 
for the K\"ahler scalar associated to a choice of complex structure. 
This form of the equation involves symmetrically both holomorphic and anti-holomorphic coordinates, 
or in other words, both eigenspaces of the complex structure.
Alternatively, the Ricci-flat equation can also be encoded in Pleba\'nski's second heavenly equation, 
which is a scalar (non-linear) ``wave-like'' equation such that only one of the eigenspaces of the 
chosen complex structure is manifest.
This second point of view is more appropriate for the generalization to the 
`hyper-heavenly equation' of Pleba\'nski and Robinson \cite{PlebanskiRobinson}.

In addition to leading to major simplifications in the study of the SD Einstein equations,  
Pleba\'nski's scalar potentials have found many applications in different areas of current interest, such 
as in supersymmetry \cite{Gaiotto}, algebraic geometry \cite{Bridgeland, Dunajski2020}, 
and scattering amplitudes \cite{Adamo}.
The self-duality condition, however, is a strong restriction on the curvature.
From the complex geometry perspective, self-duality is encoded in the hyperk\"ahler structure,
but the heavenly equations involve the choice of only one explicit complex structure out of the 2-sphere of them.
One may then wonder whether the potentialization result can be extended to 
the case where there is only {\em one} complex, non-K\"ahler structure.
The removal of the K\"ahler condition also allows to consider the possibility of a conformally invariant generalization, 
given that the integrability of an orthogonal almost-complex structure is a conformally invariant property.
In this work we study these questions, assuming a half-algebraically special (instead of self-dual) Weyl tensor.

\bigskip
These questions, and the results we obtain,
are closely related to (and inspired by) the hyper-heavenly (HH) construction of Pleba\'nski and collaborators 
\cite{PlebanskiRobinson, PlebanskiFinley, Boyer}, who extended the heavenly formalism to
the reduction of the full vacuum Einstein equations 
for an algebraically special (complex) 4-space to a single non-linear scalar PDE (the HH equation), by exploiting 
the existence of a special set of four complex coordinates adapted to a foliation by ``null strings''.

Our main result is the generalization of such a potentialization  
to a conformally invariant setting or, more specifically, to (4d) closed Einstein-Weyl 
structures\footnote{We will not be concerned here with conformal gravity in the sense of the Bach-flat equations.}, 
as well as a coordinate-free formulation based on the invariant structures 
associated to complex geometry. 
See \cite{Calderbank} for a review of Einstein-Weyl geometry.
We derive a (coordinate-free) ``conformal HH equation'', and reconstruct the conformal structure 
(curvature and metric) from a solution to this equation.
We comment on relations to perturbation theory,
and also study the potentialization of the
Dirac-Weyl, Maxwell (with and without sources), and Yang-Mills systems.

Our method is based on an involutive, maximally isotropic distribution on the tangent bundle.
Recall that, in dimension $n$, this is a subbundle $\tilde{L}$ such that $g|_{\tilde{L}}=0$ (isotropic), 
it is $(n/2)$-dimensional (maximal), and it satisfies $[\tilde{L},\tilde{L}]\subset\tilde{L}$ (involutive).
This immediately suggests a connection with complex geometry and twistor theory.
More precisely, for an $n$-dimensional orientable {\em Riemannian} manifold, we have the isomorphisms
(for a given choice of orientation)
\begin{align}
 \left\{ \begin{matrix} \text{orthogonal almost} \\ \text{complex structures} \end{matrix} \right\} 
 \;\; \cong \;\;
 \left\{ \begin{matrix} \text{maximal isotropic} \\ \text{subspaces} \end{matrix} \right\} 
 \;\; \cong \;\;
 \left\{ \begin{matrix} \text{projective} \\ \text{pure spinors} \end{matrix} \right\}  \label{isomorphisms}
\end{align}
see \cite[Chapter IV, Propositions 9.7 and 9.8]{Lawson}.
This means that, at least in the Riemannian setting, 
we can understand complex structures, isotropic subspaces, and projective (pure) spinors 
as different aspects of essentially the same object: twistor space.
This space can, itself, be given a natural almost-complex structure, and the Atiyah-Hitchin-Singer 
theorem \cite{Atiyah} establishes that, in 4d, this is integrable (and thus the twistor space is a complex three-fold)
if and only if the Weyl tensor is SD.

On the other hand, suppose that {\em a specific} element (as opposed to the whole space) 
in one of the spaces in \eqref{isomorphisms} has some special property,
then this translates into some special conditions for the associated elements of the others.
In particular, an {\em involutive} isotropic subspace corresponds, on the one hand, to an {\em integrable} 
almost-complex structure, and on the other hand, to a (projective, pure) spinor field satisfying some differential equation.
These equivalences give a unified perspective on the structures that lead to the integration procedure in this work.

\bigskip
An important point, however, is that the isomorphisms \eqref{isomorphisms} are valid for {\em Riemann} signature,
but in this work we are interested in working with generic signature (or complex metrics).
Let us restrict the discussion to four dimensions. 
All spinors in 4d are pure, so purity is not an extra condition here.
It is clear that a projective spinor defines a maximal isotropic subspace and viceversa,
so the issue is the equivalence with complex structures.
First, recall that an orthogonal almost-complex structure is a $(1,1)$ tensor $J$ that squares 
to $-1$ and that preserves the metric. 
Ignoring for the moment reality conditions, one can show \cite{Araneda2021} that 
in {\em any} signature (Riemann, Lorentz or split), 
this map is equivalent to {\em two} projective spinors with the same chirality. 
This is a point in a complex sphere 
$\mathbb{C}S^{2}=(\mathbb{CP}^{1}\times\mathbb{CP}^{1})\backslash\mathbb{CP}^{1}$.
Reality of $J$ then translates into reality conditions for spinors.
In Riemann and split signature, spinor complex conjugation preserves chirality so 
$J$ is real as long as its defining spinors are `real'. 
This reduces $\mathbb{C}S^{2}$ to a real sphere $\mathbb{CP}^{1}=S^{2}$
in the Riemannian case, and to a hyperboloid $\mathbb{CP}^{1}\backslash\mathbb{RP}^{1}$
in split signature, so, in both cases, $J$ is equivalent to only one projective spinor.
In Lorentz signature, by contrast, spinor complex conjugation 
changes chirality, and this translates into the fact that 
$J$ is necessarily complex-valued\footnote{There are different ways to see that $J$ is complex-valued, 
see e.g. the work of Flaherty \cite{Flaherty, Flaherty2}. See also remark \ref{Remark:Terminology} below.}.
Therefore, for Lorentz signature, the isotropic subspace $\tilde{L}$ defines only ``half'' of $J$. 
In the present work, either $J$ is fully determined by the geometry (Riemann and split signature),
or we can choose the other half at will (Lorentz signature and complex metrics); 
see section \ref{Sec:FullComplexStructure}.

\smallskip
In any case, the involutivity of $\tilde{L}$ implies some notion of ``integrability'' for $J$. 
In Riemann and split signature this is just the usual notion, whereas in Lorentz signature (and complex metrics),
we have ``half-integrability'': since the eigenspaces of $J$ are not related, only ``half of $J$'' is integrable.
We also mention that two actually different notions of ``halves'' are involved in this work, 
but we postpone this discussion to section \ref{Sec:Halves}.

\smallskip
Now, even though the heavenly equations may involve the choice of only one explicit complex structure, 
the hyperk\"ahler condition is still implicitly used, 
since the existence of a parallel spin frame is assumed (see \cite[Eq. (1.34)]{Plebanski1975}),
and this implies SD curvature. 
That is, even though only one point of the ``sphere'' $\mathbb{C}S^{2}$ is made explicit, 
the special structure of the rest still plays a central role.
In our work we choose a half-integrable $J$ but there are no other special conditions, apart 
from half-algebraic speciality of the Weyl tensor.
However, it turns out that a choice of $J$ is sufficient to construct a general formalism that,
when restricted to special conditions, somehow mimics, in a conformally invariant manner, the 
features of hyperk\"ahler that are needed.
The price to pay is, loosely speaking, that the whole construction is tied to such a specific $J$.
In more concrete terms, $J$ produces a conformally invariant connection, 
which we denote by $\C_{a}=\C_{AA'}$, see section \ref{Sec:ConnectionC}.
Half-integrability of $J$, together with algebraic speciality of the Weyl tensor, 
imply that the projective spinor field $o^{A}$ associated to the involutive eigenbundle satisfies
\begin{align}
 \C_{AA'}o^{B} = 0,  \qquad [o^{A}\C_{AA'}, o^{B}\C_{BB'}] = 0, \label{SpecialConditions}
\end{align}
where the commutator acts on {\em any} (primed or unprimed) spinor field.
These two conditions are sufficient to construct the two key objects that underlie our procedure in practice:
$(i)$ a special de Rham complex, whose local exactness leads to the existence of potentials, 
and $(ii)$ parallel frames for all the relevant fibre bundles involved. 
Notice that in twistor theory, operators of the form $o^{A}\nabla_{AA'}$ also play a special role since 
they represent differentiation along twistor surfaces; in addition, the analogue of the 
second equation in \eqref{SpecialConditions} corresponds to SD curvature.
The difference is that, in that case, the spinor field $o^{A}$
varies over the whole sphere since there are no preferred complex structures, 
whereas here $o^{A}$ is fixed and this changes $\nabla_{AA'}$ to $\C_{AA'}$.
We illustrate this in Fig. \ref{Figure:J} for the Riemannian case.

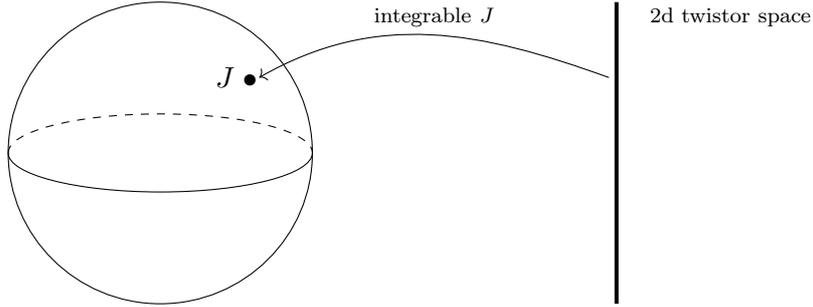
\begin{figure}
\centering
\begin{tikzpicture} 
\def\R{2} 
\def\angEl{15} 
\def\angAz{-100} 
\draw (0,0) circle (\R);
\DrawLatitudeCircle[\R]{0}
\node (J) at (1,1) {$J \; \bullet$};
\draw[<-] (1.3,1) to[out=30, in=160] (5.9,1);
\node (intJ) at (3.6,1.8) {\scriptsize integrable $J$};
\draw[line width=1.5pt] (6,-2) -- (6,2);
\node (TJ) at (7.5,1.8) {\scriptsize 2d twistor space};
\end{tikzpicture}
\caption{The 2-sphere $\mathbb{CP}^{1}=S^{2}$ of orthogonal almost-complex structures in Riemann signature. 
For a hyperk\"ahler geometry, all points in $S^{2}$ are equivalent in the sense that they all define integrable $J$'s.
In turn, each $J$ defines a 2d twistor space, which is a fiber of the 3d twistor space over $\mathbb{CP}^{1}$ 
(see \cite[Section 5.2]{Araneda2021}).
In the present paper, only one $J$ (i.e. one point in $S^{2}$, or more generically $\mathbb{C}S^2$) is privileged, 
and all geometric constructions, such as the connection $\C_{a}$, are adapted to this $J$. 
The only remaining twistor structure is a 2d twistor space associated to the (half-)integrability of $J$.}
\label{Figure:J}
\end{figure}

\bigskip
\noindent
{\bf Overview.}
In section \ref{Sec:Geometry} we give the necessary background for the rest of the work: 
Einstein-Weyl systems, almost-complex structures, the formalism of the complex-conformal connection $\C_{AA'}$, 
the integration machinery based on special de Rham complexes, and a notion of ``formal solutions'' and
``constants of integration'' that are needed in later sections.
In section \ref{Sec:LinearFields} we study linear fields: the Dirac-Weyl and Maxwell systems;
and in section \ref{Sec:YangMills} we generalize the construction to the non-linear Yang-Mills fields.
In section \ref{Sec:Einstein} we analyse closed Einstein-Weyl structures, deriving what we 
call the ``conformally invariant HH equation'', and showing also how to study the ordinary Einstein equations from this perspective.
In section \ref{Sec:Reconstruction} we reconstruct the conformal structure in terms of a solution to the conformal HH equation.
A summary and conclusions are given in section \ref{Sec:Conclusions}.
We include three appendices with additional identities and details of calculations.

\bigskip
\noindent
{\bf Notation and conventions.}
We work with 4d, orientable ``spacetimes'' which can be complex or real-analytic, with generic signature.
We follow the notation and conventions of Penrose and Rindler \cite{PR1, PR2}; 
in particular, we will use the abstract index formalism. 
Abstract indices are $a,b,c,...$, $A,B,C,...$, $A',B',C',...$ and $i,j,k,...$
for, respectively, spacetime, left and right spinors, and Yang-Mills objects.
We also follow the usual convention of associating tensor indices with pairs of spinor indices of 
opposite chirality, that is $a\equiv AA'$, $b\equiv BB'$, etc.
Where explicitly stated, concrete numerical indices will also be used,
and they are the boldface counterparts of their abstract versions, namely
${\bf a,b,}... =0,...,3$; ${\bf A,B,}... = 0,1$; ${\bf A',B'},... = 0',1'$; ${\bf i}, {\bf j}, ... =1,...,n$.
In addition, following also the notation in \cite{PR1}, we will often denote 
an arbitrary collection of {\em abstract spacetime/spinor} ({\em not} Yang-Mills) indices by $\mathscr{A}$, that is
(see \cite[pp. 87--91]{PR1})
\begin{equation}
 \varphi^{\mathscr{A}} \equiv \varphi^{DE...P'Q'...}_{KL...X'Y'...}. \label{NotationIndices}
\end{equation}

\smallskip
\noindent
{\bf Related work.} 
Our work is inspired in the hyper-heavenly (HH) construction of Pleba\'nski and collaborators 
\cite{PlebanskiRobinson, PlebanskiFinley, Boyer}. 
This construction was studied and developed further mainly in the 80s;
we mention in particular the works by Torres del Castillo \cite{TdC}, Hickman and McIntosh \cite{Hickman}, 
R\'ozga \cite{Rozga}, and Jeffryes \cite{Jeffryes}.
This last work analysed exhaustively several systems, in particular
the coupled Einstein-Yang-Mills equations, and we will find some parallelisms 
among Jeffryes' results and ours (we mention however that we do not treat the coupled Einstein-Yang-Mills system). 
Regarding the study of complex structures in relativity, we mention the works of Flaherty \cite{Flaherty, Flaherty2}.
The relation between complex structures and conformal connections was noticed by Bailey \cite{Bailey1}.
Finally, we would like to mention the works of Bailey \cite{Bailey2} and Penrose \cite{Penrose2DTa, Penrose2DTb} 
about a complex surface in projective twistor space (what we call a ``2d twistor space'').

\section{Complex and conformal geometry in 4d}\label{Sec:Geometry}

\subsection{Einstein-Weyl systems}\label{Sec:EinsteinWeyl}

Let $M$ be a manifold equipped with a metric $g_{ab}$. A conformal transformation is a map
\begin{equation}
 g_{ab} \to \Omega^{2}g_{ab}, \label{ConformalTransformation}
\end{equation}
where $\Omega$ is a positive-definite, smooth scalar field.
The conformal class of metrics associated to $g_{ab}$ is $[g_{ab}]=\{\Omega^{2}g_{ab} \;|\; \Omega\in C^{\infty}(M), \Omega>0\}$, 
and we call the pair $(M,[g_{ab}])$ a conformal structure.

A {\em Weyl connection} for the conformal structure is a torsion-free, linear connection 
$\nabla^{\rm w}_{a}$ such that for any representative 
$g_{ab}\in[g_{ab}]$, it holds $\nabla^{\rm w}_{a}g_{bc}=-2{\rm w}_{a}g_{bc}$ for some 1-form 
${\rm w}_{a}$, that we call the ``Weyl 1-form''.
It follows that for any other representative $\hat{g}_{ab}=\Omega^{2}g_{ab}$, one has 
$\nabla^{\rm w}_{a}\hat{g}_{ab}=-2\hat{{\rm w}}_{a}\hat{g}_{bc}$, where $\hat{{\rm w}}_{a}={\rm w}_{a}-\partial_{a}\log\Omega$.
A conformal structure equipped with a Weyl connection is called a Weyl manifold \cite{Calderbank}.
Given a particular $g_{ab}\in[g_{ab}]$, with Levi-Civita connection $\nabla_{a}$, 
the relation between $\nabla^{\rm w}_{a}$ and $\nabla_{a}$ when acting on an arbitrary tensor, say $T_{a}{}^{b}$, is
\begin{equation*}
 \nabla^{\rm w}_{a}T_{b}{}^{c} = \nabla_{a}T_{b}{}^{c} - K_{ab}{}^{d}T_{d}{}^{c} + K_{ad}{}^{c}T_{b}{}^{d},
\end{equation*}
where $K_{ab}{}^{c} = {\rm w}_{a}\delta_{b}^{c} + {\rm w}_{b}\delta_{a}^{c} - {\rm w}_{e}g^{ce}g_{ab}$.

The curvature tensor of $\nabla^{\rm w}_{a}$, which we denote by $R^{\rm w}_{abc}{}^{d}$, is introduced as usual by 
$[\nabla^{\rm w}_{a}, \nabla^{\rm w}_{b}]v^{d} = R^{\rm w}_{abc}{}^{d}v^{c}$ for any vector field $v^{a}$.
The {\em Einstein-Weyl equations} are the condition that the symmetric, trace-free part of the Ricci tensor of 
$\nabla^{\rm w}_{a}$ vanishes, which is equivalent to
\begin{equation}
 R^{\rm w}_{abc}{}^{b} + R^{\rm w}_{cba}{}^{b} = \lambda^{\rm w} g_{ab}, \label{EWeq}
\end{equation}
for some function $\lambda^{\rm w}$. Unlike the ordinary Einstein case, the function $\lambda^{\rm w}$ 
is generally {\em not} constant, unless one requires some other special conditions.

A {\em closed} Einstein-Weyl manifold is a Weyl manifold in which, besides \eqref{EWeq}, the 
Weyl 1-form is closed, ${\rm dw}=0$. It is, therefore, locally exact: there exists, locally, a scalar field $u$ such that
\begin{equation}
 {\rm w}_{a} = \partial_{a}\log u, \label{closedWeylform}
\end{equation}
where $u$ transforms under conformal rescaling \eqref{ConformalTransformation} as $u\to\Omega^{-1}u$.
In this work we will focus on closed Einstein-Weyl systems.
For this case, the function $\lambda^{\rm w}$ in \eqref{EWeq} {\em is} actually constant \cite{Calderbank}.

We give the spinor decomposition of $R^{\rm w}_{abc}{}^{d}$, for the case in which \eqref{closedWeylform} holds, 
in appendix \ref{App:WeylConnections}.

\subsection{Almost-complex structures}\label{Sec:AlmostComplexStructures}

For the purposes of this work,
an orthogonal almost-complex structure on a 4-manifold equipped with a metric $g_{ab}$
is a tensor field $J^{a}{}_{b}$ such that $J^{a}{}_{c}J^{c}{}_{b}=-\delta^{a}{}_{b}$ and $g_{cd}J^{c}{}_{a}J^{d}{}_{b}=g_{ab}$. 
Notice that these conditions are conformally invariant.
The usual definition also requires $J^{a}{}_{b}$ to be real-valued. 
Ignoring for the moment reality conditions, one can show \cite{Araneda2021} 
that any such map is equivalent to two independent projective spinors with the same chirality. 
More explicitly, choosing for concreteness negative chirality (i.e. unprimed spinors),
$J^{a}{}_{b}$ is equivalent to two projective spinors $[o^{A}]$, $[\iota^{A}]$ via
\begin{equation}\label{JGR}
 J^{a}{}_{b} = i(o^{A}\iota_{B}+\iota^{A}o_{B})\delta^{A'}{}_{B'},
\end{equation}
where we chose the normalization $o_{A}\iota^{A}=1$ (otherwise the factor $(o_{C}\iota^{C})^{-1}$ 
should be included in the right hand side of \eqref{JGR}, see \cite{Araneda2021}).
Recalling that $o^{A}$ and $\iota^{A}$ should really be thought of as {\em projective} spinors, 
if we want to preserve the normalization the possible rescalings are
\begin{equation}
 o^{A} \to \lambda o^{A}, \qquad \iota^{A}\to \lambda^{-1}\iota^{A}. \label{GHP}
\end{equation}
This is the basic `gauge freedom' associated to the `GHP formalism' (see e.g. \cite[Section 4.12]{PR1}); 
here we see that it arises from the inherent 
ambiguity in the representation \eqref{JGR} of an orthogonal almost-complex structure.

The $(+i)$- and $(-i)$-eigenspaces of $J^{a}{}_{b}$, here denoted $L$ and $\tilde{L},$ are respectively 
spanned by vectors of the form $\iota^{B}\nu^{B'}$ and $o^{B}\mu^{B'}$
for varying $\nu^{B'}$ and $\mu^{B'}$.

\medskip
Regarding reality conditions, we see that, since $J^{a}{}_{b}$ can always be expressed in the form \eqref{JGR} 
for some $[o^{A}], [\iota^{A}]$ (with $o_{A}\iota^{A}=1$),
these conditions are reflected in the possibilities of having `real spinors'.
For genuinely complex manifolds, there is no meaning of reality conditions for a 
tensor (or spinor) field, see \cite[Section 6.9]{PR2}.
For Riemannian signature, spinor complex conjugation $\dagger$ preserves chirality, 
so $J^{a}{}_{b}$ is real as long as $\iota_{A}=o^{\dagger}_{A}$.
Similarly, in split signature $J^{a}{}_{b}$ is real provided that the condition $\bar{\iota}_{A}=o_{A}$ 
holds\footnote{Even though the basic spinors in split signature are real, the ones appearing 
in \eqref{JGR} must be complexified versions (this follows from the fact that $o^{A},\iota^{A}$ are {\em principal} spinors 
for $J^{a}{}_{b}$), see \cite[Remark 3.9]{Araneda2021}.}. 
Finally, in Lorentz signature there is no invariant notion of real spinors, 
so $J^{a}{}_{b}$ cannot be real either.

For real $J^{a}{}_{b}$ (Riemann and split signature), 
the eigenspaces $L$ and $\tilde{L}$ are complex conjugated of each other: $\tilde{L}=\bar{L}$.
For complex $J^{a}{}_{b}$ (Lorentz signature and complex manifolds), $L$ and $\tilde{L}$ are independent.
For this reason, we use a tilde for objects associated to the eigenspace $\tilde{L}$ (see e.g. \eqref{ProjectionsC} below).
The fact that $L$ and $\tilde{L}$ are independent in Lorentz signature gives rise to the phenomenon 
of ``half-integrability'', which we will discuss in Section \ref{Sec:HalfIntegrability} below. 
This phenomenon has no analogue in Riemann signature; but it also appears in split signature, 
more precisely in {\em para}-complex geometry.

\smallskip
In this work we will not specify the metric signature, but only make some comments throughout regarding 
special features that are particular to each case.
One can also take the point of view that everything is {\em complexified} from the beginning, 
and then the different cases can be recovered by the imposition of appropriate reality structures, 
see \cite{Adamo17}, \cite{Woodhouse}.

\begin{remark}\label{Remark:Terminology}
After the completion of this work, we learned about some very interesting works in the literature 
\cite{Nurowski, TCh1, TCh2} dealing also with the Lorentzian analogue of a complex structure, 
in which the terminology of {\em (almost) null, Robinson}, and {\em optical structures} is used.
Since here we are leaving the signature unspecified, we will continue to use the term ``almost-complex structure'',
but the terminology of \cite{Nurowski, TCh1, TCh2} mentioned above is more appropriate for the Lorentzian case.
In addition, the formalism developed in \cite{TCh2} will likely be useful for analysing the possible generalization 
of some of the results in the present work to higher dimensions.
\end{remark}

\subsection{The complex-conformal connection $\mathcal{C}_{a}$}\label{Sec:ConnectionC}

There is an intimate link between Weyl connections and orthogonal almost-complex structures, 
as follows from Lee's construction \cite{HCLee}: 
there is a natural Weyl connection for conformal manifolds with a non-degenerate 2-form. 
We know that an orthogonal almost-complex structure $J^{a}{}_{b}$ gives automatically a non-degenerate 2-form,
namely, the fundamental form $\omega_{ab}=g_{bc}J^{c}{}_{a}$ associated to a conformal representative $g_{ab}$. 
The associated Weyl 1-form is called {\em Lee form}, and will be denoted by $f_{a}$.
In four dimensions, and in terms of the Levi-Civita connection $\nabla_{a}$ of $g_{ab}$, it is given by\footnote{In 
GHP notation, $f_{a}$ can be expressed as $f_{a}=\rho' \ell_{a}+\rho n_{a} - \tau'm_{a} - \tau\tilde{m}_{a}$, 
where $\ell_{a}=o_{A}\mu_{A'}$, $n_{a}=\iota_{A}\nu_{A'}$, $m_{a}=o_{A}\nu_{A'}$ and $\tilde{m}_{a}=\iota_{A}\mu_{A'}$, 
and $\mu_{A'}, \nu_{A'}$ is an arbitrary primed spin frame with $\mu_{A'}\nu^{A'}=1$.}
\begin{equation}
 f_{a} = -\tfrac{1}{2}J^{b}{}_{c}\nabla_{b}J^{c}{}_{a}. \label{LeeForm}
\end{equation}
We will denote the corresponding Weyl connection by $\nabla^{\rm f}_{a}$.
Recall that this is defined by the condition $\nabla^{\rm f}_{a}g_{bc}=-2f_{a}g_{bc}$.

\smallskip
In the following we give a brief review of the construction of the complex-conformal connection $\C_{a}$ 
that is central to this work; for further details we refer to \cite{Araneda18, Araneda20, Araneda2021}.

\smallskip
The gauge group for the frame bundle $F$ of the conformal structure is\footnote{For simplicity 
we consider real conformal rescalings.} $G={\rm Spin}(4)\times\mathbb{R}^{+}$,
where ${\rm Spin}(4)$ is the spin group appropriate to the choice of signature.
For concreteness let us focus on the complex case, 
${\rm Spin}(4,\mathbb{C}) = {\rm SL}(2,\mathbb{C})\times{\rm SL}(2,\mathbb{C})$, 
but the construction can also be done in the other cases as well, with minor modifications.
Spinor fields over $M$ (including tensors and scalars) transform under representations of $G$, 
and should be thought of as sections of vector bundles associated to $F$.
However, a choice of almost-complex structure $J^{a}{}_{b}$ gives a reduction of $G$ to the subgroup that preserves it.
Taking into account our choice of normalization for the spinors representing $J^{a}{}_{b}$, 
this subgroup is $G_{o}=\mathbb{C}^{\times}\times{\rm SL}(2,\mathbb{C})\times\mathbb{R}^{+}$.
Denote by $F_{o}$ the principal bundle over $M$ with structure group $G_{o}$. 
Spinor fields now transform under representations of $G_{o}$.
Such representations can be characterized by two real numbers, which we call {\em weights} and are defined as follows.
If, under the transformations \eqref{ConformalTransformation} and \eqref{GHP}, a spinor field 
$\varphi^{\mathscr{A}}$ (recall the notation \eqref{NotationIndices}) transforms as
\begin{equation}
 \varphi^{\mathscr{A}} \to \lambda^{p}\Omega^{w}\varphi^{\mathscr{A}}, \label{GaugeTransformation}
\end{equation}
where $\lambda\in\mathbb{C}^{\times}$, $\Omega\in\mathbb{R}^{+}$, then we say that the field
has `GHP weight' $p$ and `conformal weight' $w$.
Fields transforming this way are to be thought of as sections of vector bundles associated to $F_{o}$; 
these bundles will be denoted by $\mathbb{S}_{p,w}$. 
The notion of weights is simply a reminder of the `gauge freedom' associated to a choice of spinors 
and metric from the equivalence classes.

\smallskip
It will be convenient to think of $p$ and $w$ as functions that take fields and return real numbers, 
so that we will use the notation $p(\varphi^{\mathscr{A}})$ and $w(\varphi^{\mathscr{A}})$. 
For example, a metric in the conformal class has weights $p(g_{ab})=0$ and $w(g_{ab})=2$, 
and the almost-complex structure has weights $p(J^{a}{}_{b})=0$, $w(J^{a}{}_{b})=0$.
The spinors $o^{A}$ and $\iota^{A}$ have weights $p(o^{A})=1$, $w(o^{A})=w^{0}$ and $p(\iota^{A})=-1$, $w(\iota^{A})=w^{1}$,
where $w^{0}$ and $w^{1}$ can be chosen at will as long as they satisfy $w^{0}+w^{1}+1=0$. 
Here we will choose 
\begin{equation}
 w^{0}=0, \qquad w^{1}=-1. \label{ConfWeights}
\end{equation}

The Weyl connection $\nabla^{\rm f}_{a}$ gives a local connection 1-form in $F$.
In turn, the reduced bundle $F_{o}$ inherits a connection from the reduction.
This 1-form allows to define a covariant derivative on the associated bundles $\mathbb{S}_{p,w}$. 
We denote this covariant derivative by $\mathcal{C}_{a}$.
Its action on a generic spinor field $\varphi^{\mathscr{A}}$ with weights $p$ and $w$ is explicitly
\begin{equation}
 \C_{a}\varphi^{\mathscr{A}}  = \nabla^{\rm f}_{a}\varphi^{\mathscr{A}}+ (w f_{a} + p P_{a})\varphi^{\mathscr{A}},
 \label{ActionOfC}
\end{equation}
where $P_{a} := \iota_{B}\nabla^{\rm f}_{a}o^{B}$ (for the choice \eqref{ConfWeights}).
Under a transformation \eqref{GaugeTransformation}, it holds
\begin{equation}
 \mathcal{C}_{a}\varphi^{\mathscr{A}} \to \lambda^{p}\Omega^{w}\mathcal{C}_{a}\varphi^{\mathscr{A}}.
\end{equation}

Of central importance to this work are the projections of $\C_{a}$ to the eigenbundles of $J$. We define 
\begin{equation}
 \tC_{A'} := o^{A}\C_{AA'}, \qquad \C_{A'}:=\iota^{A}\C_{AA'}. \label{ProjectionsC}
\end{equation}
These operators correspond to covariant differentiation (of sections of $\mathbb{S}_{p,w}$)
along vectors in $\tilde{L}$ and $L$ respectively:
if $\tilde{P}, P$ are the projectors to $\tilde{L}, L$, then $\tC_{A'}$ and $\C_{A'}$ correspond 
to $\C_{\tilde{P}X}$ and $\C_{PX}$ for any vector field $X$.

In appendix \ref{App:identities} we give some details and identities for the curvature of $\mathcal{C}_{a}$ 
and the operators \eqref{ProjectionsC}, 
both for arbitrary spacetimes and for the special spacetimes considered in this work.

\subsection{Half-integrability}\label{Sec:HalfIntegrability}

In Riemann and split signature, an orthogonal almost-complex structure is said to be integrable 
if its $(-i)$-eigenbundle $\tilde{L}$ is involutive under the Lie bracket of vector fields. 
Since the $(+i)$-eigenbundle is the complex conjugate, it is then also involutive. 
In Lorentz signature, the eigenbundles are independent, so involutivity of one of them 
does not imply involutivity for the other. 
For this reason we talk about ``half-integrability'' in this work, and we say that $J$
is ``half-integrable'' if only one of its eigenbundles is involutive. 

\medskip
The following discussion applies to any signature.
Let $J$ be an orthogonal almost-complex structure, and let $o^{A}$, $\iota^{A}$ be the spinors representing it 
as in \eqref{JGR}. A calculation shows (see \cite{Araneda2021})
that the eigenbundle $\tilde{L}$ is involutive if and only if the spinor field $o^{A}$ satisfies the equation
\begin{equation}
 o^{A}o^{B}\nabla_{AA'}o_{B} = 0. \label{SFR}
\end{equation}
We call this equation the ``shear-free condition'' on $o^{A}$, since, in Lorentz signature, 
when \eqref{SFR} holds it implies that the real vector field $o^{A}\bar{o}^{A'}$ is tangent to a null geodesic congruence 
which is shear-free; see \cite[Chapter 7]{PR2}.
Using identity \eqref{Co}, we see that \eqref{SFR} is equivalent to the equation
\begin{equation}
 \mathcal{C}_{AA'}o^{B} = 0, \label{ParallelSpinor}
\end{equation}
thus, integrability is encoded in parallel spinors under $\mathcal{C}_{a}$.
We notice though that \eqref{ParallelSpinor} is actually a {\em non-linear} equation for $o^{B}$, 
since $\mathcal{C}_{AA'}$ also depends on $o^{B}$ via the associated $J$.

\medskip
Besides the integrability property \eqref{SFR}, in this work we also require $o^{A}$ to be a repeated 
principal spinor of the Weyl tensor, i.e. we require $o^{A}$ to satisfy
\begin{equation}
 \Psi_{ABCD}o^{B}o^{C}o^{D} = 0 \label{rPND}
\end{equation}
where $\Psi_{ABCD}$ is the anti-self-dual (ASD) Weyl curvature spinor of the conformal structure.
We then have the following:
\begin{proposition}[See \cite{Araneda20, Araneda2021}]
When \eqref{SFR} and \eqref{rPND} hold, the operator $\tC_{A'}$ defined in \eqref{ProjectionsC} satisfies
\begin{equation}
 [\tC_{A'},\tC_{B'}] = 0 \label{FlatConnection}
\end{equation}
when acting on any weighted spinor field $\varphi^{\mathscr{A}}$.
\end{proposition}

Alternatively, \eqref{FlatConnection} can be written as $\tC^{A'}\tC_{A'}=0$. 
In practice, this is the condition we will use to find potentials: if a spinor field satisfies 
$\tC^{A'}\varphi_{A'}=0$, then, roughly speaking, 
$\tC^{A'}\tC_{A'}=0$ implies that there is $\psi$ such that $\varphi_{A'}=\tC_{A'}\psi$.
In order to make this more precise, in section \ref{Sec:deRham} we will introduce 
a special exterior derivative and its associated de Rham complex.

\subsubsection{Two different notions of ``halves''}\label{Sec:Halves}

We make a small digression here to comment on two different meanings of 
``halves'' regarding complex geometry in different signatures.
This is originated in the different possibilities for spinor complex conjugations: 
in the Riemannian and split cases, this conjugation preserves chirality of the spin-spaces, and in the 
Lorentzian case it interchanges the chiralities.
In all cases, the space of orthogonal almost-complex structures has two components, 
one of them associated to positive chirality and the other one to negative chirality.

In Riemann and split signature, the requirement that an almost-complex structure be real-valued 
translates into the condition that the two spinors in \eqref{JGR} are complex conjugated from each other,
and this implies that there is no notion of half-integrability.
On the other hand, positive and negative chirality spinors are independent, which means, for example, 
that the SD and ASD Weyl spinors $\tilde\Psi_{A'B'C'D'}$ and $\Psi_{ABCD}$ are independent, 
so one can talk about one ``side'' being special while the other one remains general.
In particular, this makes the twistor construction in Riemann signature non-trivial: 
the twistor space is a complex 3-fold if and only if all orthogonal almost-complex structures in $M$
are integrable, and this implies $\Psi_{ABCD}\equiv 0$; but one still has a non-trivial geometry 
since $\tilde{\Psi}_{A'B'C'D'}$ remains general.

In Lorentz signature, we mentioned that an orthogonal almost-complex structure is necessarily complex-valued, 
and this implies that the eigenspaces are independent, so we have a notion of half-integrability.
On the other hand, positive and negative chirality spinors are related by complex conjugation,
and this means that the SD and ASD Weyl spinors are complex conjugated, $\tilde\Psi_{A'B'C'D'}=\bar{\Psi}_{A'B'C'D'}$.
As a consequence, unlike the rest of the cases, there are no two ``sides'' regarding the curvature.

\medskip
In summary, let us denote the eigenspaces of an almost-complex structure by $\tilde{L}$, $L$, 
the spin-spaces with opposite chirality by $\mathbb{S}$, $\mathbb{S}'$, and complex conjugation by a bar $\bar{}$ . 
In Riemann and split signature, we have $\tilde{L}=\bar{L}$ and $\bar{\mathbb{S}}\neq\mathbb{S}'$, 
and in Lorentz signature $\tilde{L}\neq\bar{L}$ and $\bar{\mathbb{S}}=\mathbb{S}'$.

\subsection{De Rham complexes}\label{Sec:deRham}

Given an involutive distribution $D$ on the tangent bundle, a simple way to associate a natural de Rham complex to $D$ 
is to observe that $D$ has the structure of a Lie algebroid. See e.g. \cite[Section 3.1]{Gualtieri}.
Then the distribution comes automatically with a natural exterior derivative for differential forms over $D$,
together with an associated de Rham complex. 

\smallskip
Taking as our involutive distribution the eigenbundle $\tilde{L}$,
the above is the basic logic that allows us to deduce the existence of potentials. 
However, we need a slight generalization of this: we need to consider 
differential forms over $\tilde{L}$ with values on the weighted spinor bundles $\mathbb{S}_{p,w}$
of section \ref{Sec:ConnectionC}. 
That is, denoting $\Lambda^{k}=\wedge^{k}\tilde{L}^{*}$, we need to consider the space 
$\Lambda^{k}\otimes\mathbb{S}_{p,w}$.
Notice that elements of the dual $\tilde{L}^{*}$ are of the form $\iota_{A}\kappa_{A'}$ for varying $\kappa_{A'}$.

\smallskip
The first problem that arises when replacing $\Lambda^{k}$ by differential forms with values on a vector bundle is 
that one needs a connection for the latter structure (i.e. a {\em covariant} exterior derivative).
In our case this is easily solved, since for $\mathbb{S}_{p,w}$ we have the 
complex-conformal connection $\mathcal{C}_{a}$ of section \ref{Sec:ConnectionC}.
The second problem is that, in order to get a differential complex, the resulting covariant exterior derivative 
should square to zero. This will follow from \eqref{FlatConnection}.

\smallskip
Since $\tilde{L}$ is 2-dimensional, we have $\Lambda^{k}=\emptyset$ for $k>2$, 
so we only need to consider 0-, 1- and 2-forms.
Mimicking the definition for the Lie algebroid differential (which in turn mimics the ordinary 
exterior derivative), we introduce
\begin{equation}
 {\rm d}^{\tC} : \Gamma(\Lambda^{k}\otimes\mathbb{S}_{p,w}) \to \Gamma(\Lambda^{k+1}\otimes\mathbb{S}_{p,w})
\end{equation}
by its action on $\mathbb{S}_{p,w}$-valued 0- and 1-forms (recall the notation \eqref{NotationIndices}):
\begin{subequations}
\begin{align}
 ({\rm d}^{\tC}h)^{\mathscr{A}}(X) ={}& \mathcal{C}_{X}h^{\mathscr{A}}, \label{CovExtDer0} \\
 ({\rm d}^{\tC}\omega)^{\mathscr{A}} (X,Y) 
 ={}& \mathcal{C}_{X}\omega^{\mathscr{A}}(Y) - \mathcal{C}_{Y}\omega^{\mathscr{A}}(X) 
   - \omega^{\mathscr{A}}([X,Y]_{\mathcal{C}}). \label{CovExtDer1}
\end{align}
\end{subequations}
where $h^{\mathscr{A}}\in\Lambda^{0}\otimes\mathbb{S}_{p,w}$, $\omega^{\mathscr{A}}\in\Lambda^{1}\otimes\mathbb{S}_{p,w}$, 
$X$ and $Y$ are sections of $\tilde{L}$, and we are using the bracket 
$[X,Y]_{\mathcal{C}} = \mathcal{C}_{X}Y-\mathcal{C}_{Y}X$. 
We also make the convention that $\mathscr{A}$ does not include the differential form indices from $\Lambda^{k}$.
To compute $({\rm d}^{\tC})^{2}$, we only need to consider 0-forms, since 
$({\rm d}^{\tC})^{2}\omega^{\mathscr{A}} \equiv 0$ 
for any 1-form $\omega^{\mathscr{A}}$ (on account of the 2-dimensionality of $\tilde{L}$).
From \eqref{CovExtDer0}-\eqref{CovExtDer1}, we find:
\begin{equation*}
 [{\rm d}^{\tC}{\rm d}^{\tC}h]^{\mathscr{A}}(X,Y) 
 = \mathcal{C}_{X}\mathcal{C}_{Y}h^{\mathscr{A}} - \mathcal{C}_{Y}\mathcal{C}_{X}h^{\mathscr{A}} 
 - \mathcal{C}_{[X,Y]_{\mathcal{C}}}h^{\mathscr{A}},
\end{equation*}
so we see that $({\rm d}^{\tC})^{2}$ is given by the curvature 2-form of $\mathcal{C}_{a}$ evaluated on elements of $\tilde{L}$.
More explicitly, since $X,Y\in\Gamma(\tilde{L})$, we have $X^{a}=o^{A}\tilde{x}^{A'}$, $Y^{a}=o^{A}\tilde{y}^{A'}$ 
for some $\tilde{x}^{A'}, \tilde{y}^{A'}$. 
Then a short calculation (analogous to the proof of lemma 4.8 in \cite{Araneda2021}) gives
\begin{equation*}
 [{\rm d}^{\tC}{\rm d}^{\tC}h]^{\mathscr{A}}(X,Y) = \tilde{x}_{B'}\tilde{y}^{B'}\tC_{A'}\tC^{A'}h^{\mathscr{A}}.
\end{equation*}
By virtue of \eqref{FlatConnection}, the right hand side vanishes, so we have indeed $({\rm d}^{\tC})^{2}=0$.
In other words:
\begin{lemma}[See Eq. (3.27) in \cite{Araneda20}, and Lemma 4.8 in \cite{Araneda2021}]
Suppose that $o^{A}$ is shear-free and a repeated principal spinor, eqs. \eqref{SFR} and \eqref{rPND}.
Then the following is a differential complex:
\begin{equation}\label{wderham}
 0 \to \Gamma(\Lambda^{0}\otimes\mathbb{S}_{p,w}) 
 \xrightarrow{{\rm d}^{\tC}} \Gamma(\Lambda^{1}\otimes\mathbb{S}_{p,w}) 
 \xrightarrow{{\rm d}^{\tC}} \Gamma(\Lambda^{2}\otimes\mathbb{S}_{p,w}) \to 0.
\end{equation}
\end{lemma}

The abstract index version of eqs. \eqref{CovExtDer0}-\eqref{CovExtDer1} is:
\begin{subequations}
\begin{align}
 ({\rm d}^{\tC}h)^{\mathscr{A}}_{a} ={}& (\tC_{A'}h^{\mathscr{A}})\iota_{A}, \label{CovExtDer0Abstract}  \\
 ({\rm d}^{\tC}\omega)_{ab}^{\mathscr{A}} ={}& (\tC_{C'}\tilde{\omega}^{C'\mathscr{A}})\iota_{A}\iota_{B}\epsilon_{A'B'},
 \label{CovExtDer1Abstract}
\end{align}
\end{subequations}
where in the second line we defined $\tilde{\omega}^{C'\mathscr{A}}  \equiv o^{C}\omega^{C'\mathscr{A} }_{C}$.
We have:

\begin{lemma}\label{Lemma:Potentials}
Suppose that \eqref{SFR} and \eqref{rPND} are satisfied. Assume that the manifold is complex, or real-analytic
with complexification $\mathbb{C}M$.
\begin{enumerate}
\item Let $\psi^{\mathscr{A}}_{A'}$ be a spinor field with weights $(p, w)$ such that
\begin{equation}
 \tC^{A'}\psi^{\mathscr{A}}_{A'}= 0. \label{dfield}
\end{equation}
Then there exists, locally, a spinor field $\Phi^{\mathscr{A}}$ in $\mathbb{C}M$, with weights $(p-1, w)$, such that
\begin{equation}
 \psi^{\mathscr{A}}_{A'} = \tC_{A'}\Phi^{\mathscr{A}}. \label{fieldd}
\end{equation}
\item For any $\theta^{\mathscr{A}}$ with weights $(p,w)$, there exists locally a spinor field 
$\xi^{\mathscr{A}}_{A'}$ in $\mathbb{C}M$, with weights $(p-1,w)$, such that 
\begin{equation}\label{integral}
 \theta^{\mathscr{A}} = \tC^{A'}\xi^{\mathscr{A}}_{A'}.
\end{equation}
\end{enumerate}
\end{lemma}

\begin{proof}
First of all, the complexification $\mathbb{C}M$ is needed in order to be able to interpret 
the differential complex \eqref{wderham} as a de Rham complex, which is then locally exact.

For the first item, define $\Psi^{\mathscr{A}}_{a}= \iota_{A}\psi^{\mathscr{A}}_{A'}$, so 
$\Psi^{\mathscr{A}}_{a}\in\Gamma(\Lambda^{1}\otimes\mathbb{S}_{p-1,w})$.
Using equation \eqref{CovExtDer1Abstract}, we see that \eqref{dfield} is simply $({\rm d}^{\tC}\Psi)^{\mathscr{A}}_{ab}=0$.
Therefore, there exists, locally, a field $\Phi^{\mathscr{A}}\in\Gamma(\Lambda^{0}\otimes\mathbb{S}_{p-1,w})$ such that 
$\Psi^{\mathscr{A}}_{a} = ({\rm d}^{\tC}\Phi)^{\mathscr{A}}_{a}$. Using \eqref{CovExtDer0Abstract}:
\begin{align*}
 \Psi^{\mathscr{A}}_{a} = \iota_{A}\psi^{\mathscr{A}}_{A'} 
 = ({\rm d}^{\tC}\Phi)^{\mathscr{A}}_{a} = (\tC_{A'}\Phi^{\mathscr{A}})\iota_{A},
\end{align*}
thus \eqref{fieldd} follows.

For the second item, suppose that $\theta^{\mathscr{A}}$ is given, with weights $(p,w)$.
Define $\Theta^{\mathscr{A}}_{ab}=-\iota_{A}\iota_{B}\epsilon_{A'B'}\theta^{\mathscr{A}}$,
then $\Theta^{\mathscr{A}}_{ab}\in\Gamma(\Lambda^{2}\otimes\mathbb{S}_{p-2,w})$. 
From \eqref{wderham} we see that $({\rm d}^{\tC}\Theta)^{\mathscr{A}}_{abc}=0$, so there exists locally a field 
$\Xi^{\mathscr{A}}_{a}=\iota_{A}\xi^{\mathscr{A}}_{A'}\in\Gamma(\Lambda^{1}\otimes\mathbb{S}_{p-2,w})$
such that $\Theta^{\mathscr{A}}_{ab}=({\rm d}^{\tC}\Xi)^{\mathscr{A}}_{ab}$. That is:
\begin{align*}
 \Theta^{\mathscr{A}}_{ab}=-\iota_{A}\iota_{B}\epsilon_{A'B'}\theta^{\mathscr{A}}=({\rm d}^{\tC}\Xi)^{\mathscr{A}}_{ab} 
 = (\tC_{C'}\xi^{C'\mathscr{A}})\iota_{A}\iota_{B}\epsilon_{A'B'},
\end{align*}
so \eqref{integral} follows.
\end{proof}

\begin{remark}\label{Remark:InternalGaugeFreedom}
Notice that the representation \eqref{fieldd} is not unique: we have the ``gauge freedom''
\begin{equation}
 \Phi^{\mathscr{A}} \to \Phi^{\mathscr{A}} + g^{\mathscr{A}}, \qquad \tC_{A'}g^{\mathscr{A}}=0. \label{InternalGaugeFreedom}
\end{equation}
Similarly, for \eqref{integral} there is the freedom $\xi^{\mathscr{A}}_{A'} \to \xi^{\mathscr{A}}_{A'}+\tC_{A'}\gamma^{\mathscr{A}}$.
\end{remark}

\subsubsection{The Lee form}\label{Sec:FullComplexStructure}

As a first application of the integration machinery, let us study the Lee form \eqref{LeeForm}. 

\begin{proposition}
Suppose that the conditions \eqref{SFR} and \eqref{rPND} are satisfied.
Then there exists a scalar field $\phi$ with weights $p(\phi)=0$, $w(\phi)=-1$, and a scalar field 
$\eta$ with $p(\eta)=-2$, $w(\eta)=-1$, such that the Lee form \eqref{LeeForm} can be written as
\begin{equation}
 f_{AA'} = \nabla_{AA'}\log\phi - o_{A}\tC_{A'}\eta. \label{LeeForm2}
\end{equation}
\end{proposition}

\begin{proof}
The result \eqref{FlatConnection} implies that, in particular, we can consider solutions to the 
equation $\tC_{A'}\phi=0$ where $\phi$ is a scalar field with arbitrary weights. 
Choosing $p(\phi)=0$, $w(\phi)=-1$, we see from \eqref{ActionOfC} that $o^{A}\nabla_{AA'}\phi-o^{A}f_{AA'}\phi=0$, 
or equivalently $o^{A}\nabla_{AA'}\log\phi=o^{A}f_{AA'}$.
In addition, for scalar fields with $p=0$ the operators $\tC_{A'}$ and $\C_{A'}$ commute, 
so $0=\C_{A'}\tC^{A'}\phi=\tC^{A'}\C_{A'}\phi$. Applying lemma \ref{Lemma:Potentials}, 
there exists a scalar field $\eta'$, with $p(\eta')=-2$, $w(\eta')=-2$, 
such that $\C_{A'}\phi=\tC_{A'}\eta'$.
Multiplying by $\phi^{-1}$, we get $\phi^{-1}\C_{A'}\phi = \tC_{A'}\eta$, where $\eta=\phi^{-1}\eta'$
has weights $p(\eta)=-2$, $w(\eta)=-1$. 
Using $\phi^{-1}\C_{A'}\phi=\iota^{A}\nabla_{AA'}\log\phi-\iota^{A}f_{AA'}$, the result \eqref{LeeForm2} follows.
\end{proof}

Now, as mentioned in the introduction, the central object in which our method 
is based is a maximally isotropic subspace $\tilde{L}$.
In Riemann (and split) signature, the isomorphisms \eqref{isomorphisms} imply that this is equivalent 
to both a projective spinor field and an almost-complex structure.
In Lorentz signature, we still have the equivalence between $\tilde{L}$ and a projective 
spinor field, but this determines only ``half'' of an almost-complex structure $J$;
in other words, the spinor $\iota^{A}$ in \eqref{JGR} remains arbitrary.
It turns out that we can choose it so as to eliminate $\eta$ in \eqref{LeeForm2}:

\begin{proposition}\label{Prop:LeeForm}
Suppose that \eqref{SFR} and \eqref{rPND} are satisfied. Then the spinor field $\iota^{A}$ 
can be chosen such that the Lee form of the associated almost-complex structure \eqref{JGR} is
\begin{equation}
 f_{a} = \partial_{a}\log\phi. \label{fgradient}
\end{equation}
\end{proposition}

\begin{proof}
Under a change $\iota^{A}\to\iota^{A}+b o^{A}$, where $p(b)=-2$, $w(b)=-1$, we have
$J^{a}{}_{b} \to J^{a}{}_{b} + 2ibo^{A}o_{B}\delta^{A'}{}_{B'}$.
Then a short calculation shows that the Lee form changes as $f_{a} \to f_{a} -o_{A}\tC_{A'}b$.
Choosing then $b\equiv-\eta$, the new Lee form has the structure \eqref{fgradient}.
\end{proof}

\begin{remark}\label{Remark:LeeForm}
The condition \eqref{fgradient} {\em is not} assumed in sections \ref{Sec:LinearFields} and \ref{Sec:YangMills},
but it {\em is} assumed in sections \ref{Sec:Einstein} and \ref{Sec:Reconstruction}.
\end{remark}

Proposition \ref{Prop:LeeForm} shows that 
for Lorentz signature and complex manifolds, we can fix the almost-complex structure $J$ 
to be such that the corresponding Lee form is a gradient \eqref{fgradient}. 
On the other hand, for Riemann and split signature, the spinor $o^{A}$ alone already fixes $J$, 
so it is not necessarily the case that \eqref{fgradient} holds (although \eqref{LeeForm2} is still true).
In fact, for these cases, \eqref{fgradient} would imply that the manifold is conformally K\"ahler
(since $J$ is integrable), which, as far as we can see, is not guaranteed to be the case from 
our current assumptions. 
However, as stressed by remark \ref{Remark:LeeForm}, the condition \eqref{fgradient} 
is only needed in sections \ref{Sec:Einstein} and \ref{Sec:Reconstruction}, 
where the Einstein-Weyl equations are assumed. 
It turns out that these equations, in the Riemannian and split cases, are sufficient to show 
that \eqref{fgradient} actually holds, see remark \ref{Remark:LeeForm2}.

\subsection{Formal solutions and constants of integration}\label{Sec:ConstantsOfIntegration}

In the next sections we will be concerned with integrating the field equations for several systems. 
As in ordinary calculus, there will be `constants of integration' involved, 
a feature which already appears, of course, in the original HH construction \cite{PlebanskiRobinson, PlebanskiFinley}.
The HH construction is based on a set of four special complex coordinates, and the constants of integration there 
are scalar (complex) functions that depend on only two of these coordinates. 
More specifically, Pleba\'nski {\it et al} use coordinates $(\tilde{z}^{\bf A'}, z^{\bf A'})$ 
where indices ${\bf A'}, {\bf B'},...=0',1'$ are {\em concrete numerical indices}, see section \ref{Sec:Coordinates} below.
(In the notation of \cite{PlebanskiFinley}, one has $\tilde{z}^{\bf A'} \equiv p^{\dot{A}}$ and $z^{\bf A'}\equiv q^{\dot{A}}$.) 
The equations to be integrated in the HH formulation involve only partial derivatives of $\tilde{z}^{\bf A'}$, so the result of an 
integration involves arbitrary functions $f(z^{\bf A'})$.
As a typical example, the solution to the equation $\partial_{\tilde{z}^{\bf A'}}\partial_{\tilde{z}^{\bf B'}}F=0$
is $F(\tilde{z},z)=J_{\bf A'}(z)\tilde{z}^{\bf A'}+\kappa(z)$, where $J_{\bf A'}(z)$ and $\kappa(z)$ 
are three arbitrary functions of $z^{\bf A'}$, see e.g. \cite[Eq. (3.1)]{PlebanskiFinley}.
More in general, solutions to equations of the form $\partial_{\tilde{z}^{\bf A'}}...\partial_{\tilde{z}^{\bf K'}}F=0$ 
are polynomials in $\tilde{z}^{\bf A'}$ with coefficients (`constants of integration') depending on $z^{\bf A'}$;
equivalently, the coefficients are in the kernel of $\partial_{\tilde{z}^{\bf A'}}$.

\smallskip
Our formulation is coordinate-free and it only involves abstract spinor fields, 
so we need an abstract version of the above, namely of both the constants of integration and 
of the coordinates $\tilde{z}^{\bf A'}$. First of all, 
the vector fields $\partial_{\tilde{z}^{\bf A'}}$ belong to the eigenspace $\tilde{L}$ of $J$ 
(they are anti-holomorphic; see eqs. \eqref{Coordinates2} and \eqref{JinCoordinates} below), 
so in our construction they correspond to the abstract operator $\tC_{A'}$, acting on sections of $\mathbb{S}_{p,w}$. 
Thus, the analogue of functions with dependence only on $z^{\bf A'}$ are spinor fields in the kernel of $\tC_{A'}$:

\begin{definition}\label{Def:Constants}
In the context of an integration procedure, we will say that a spinor field $K^{\mathscr{A}}$ is a ``constant of integration'' if 
\begin{equation}
 \tC_{A'}K^{\mathscr{A}} = 0. \label{ConstantOfIntegration}
\end{equation}
\end{definition}

On the other hand, starting from the observation $\partial_{\tilde{z}^{\bf A'}}\tilde{z}^{\bf B'}=\delta^{\bf B'}_{\bf A'}$,
we can think of the abstract analogue of the two coordinates $\tilde{z}^{\bf A'}$ as 
a spinor field $\pi^{A'}$ satisfying the twistor-like equation 
\begin{equation}
 \tC_{A'}\pi^{B'} = \delta^{B'}_{A'}. \label{TwistorLikeEq}
\end{equation}
The weights of $\pi^{A'}$ must be $p(\pi^{A'})=-1$, $w(\pi^{A'})=0$. 
Later, in section \ref{Sec:Coordinates}, we will see the explicit relation between a solution of \eqref{TwistorLikeEq} 
and the coordinates $\tilde{z}^{\bf A'}$, see eq. \eqref{RelationPiZetaTilde}.
Any other solution to \eqref{TwistorLikeEq} can be expressed in terms of $\pi^{A'}$ and constants of integration:
\begin{proposition}\label{Prop:TwistorLikeEq}
Let $\tau_{A'}$ be a spinor field such that $\tC_{(A'}\tau_{B')}=0$. 
Then $\tau_{A'}$ can be expressed as $\tau_{A'}=a\pi_{A'}+b_{A'}$, where $\tC_{A'}a=0=\tC_{A'}b_{B'}$.
\end{proposition}

\begin{proof}
The equation $\tC_{(A'}\tau_{B')}=0$ implies that $\tC_{A'}\tau_{B'}=a\epsilon_{A'B'}$ for some function $a$. 
Contracting with $\tC^{A'}$, we see that $\tC_{A'}a=0$.
Therefore: $\tC_{A'}\tau^{B'}=a\delta^{B'}_{A'}=a\tC_{A'}\pi^{B'}=\tC_{A'}(a\pi^{B'})$, 
so we have $\tC_{A'}(\tau^{B'}-a\pi^{B'})=0$. Defining $b^{B'}=\tau^{B'}-a\pi^{B'}$, the result follows.
\end{proof}

The `constants of integration' \eqref{ConstantOfIntegration} together with the spinor field $\pi^{A'}$ in \eqref{TwistorLikeEq} 
allow to give a formal, abstract solution to equations of the form $\tC_{A'}...\tC_{K'}F_{L'...Q'}=0$
or with symmetrizations.
The equations we will need in later sections are
\begin{align}
 \tC_{A'_{1}}...\tC_{A'_{m}}F ={}& 0, \label{SpecialEqs1} \\
 \tC_{(A'_{1}}...\tC_{A'_{k}}\tau_{A'_{k+1}...A'_{n})} ={}& 0. \label{SpecialEqs2}
\end{align}
For \eqref{SpecialEqs1}, it is straightforward to check that the solution can be written as
\begin{equation}
 F = K^{0} + K^{1}_{A'}\pi^{A'} + K^{2}_{A'_{1}A'_{2}}\pi^{A'_{1}}\pi^{A'_{2}} + ... + K^{m-1}_{A'_{1}...A'_{m-1}}\pi^{A'_{1}}...\pi^{A'_{m-1}},
 \label{SolSpecialEqs1}
\end{equation}
where $K^{k}_{A'_{1}...A'_{k}}$ (no sum over $k$) are arbitrary `constants of integration' in the sense above, 
namely $\tC_{B'}K^{k}_{A'_{1}...A'_{k}}=0$.
For \eqref{SpecialEqs2}, it is convenient to first analyse the case in which there is only one $\tC_{A'}$ derivative:

\begin{proposition}\label{Prop:GeneralTwistorLikeEq}
Let $n\in\mathbb{N}$. Let $\tau_{A'_{1}...A'_{n}}$ be a totally symmetric spinor field satisfying
the equation $\tC_{(A'_{1}}\tau_{A'_{2}...A'_{n+1})}=0$. Then $\tau_{A'_{1}...A'_{n}}$ can be written as 
\begin{equation}
 \tau_{A'_{1}...A'_{n}} = 
 c^{0}\pi_{A'_{1}}...\pi_{A'_{n}} + c^{1}_{(A'_{1}}\pi_{A'_{2}}...\pi_{A'_{n})} + c^{2}_{(A'_{1}A'_{2}}\pi_{A'_{3}}...\pi_{A'_{n})} 
  +...+c^{n-1}_{(A'_{1}...A'_{n-1}}\pi_{A'_{n)}} + c^{n}_{A'_{1}...A'_{n}}
  \label{SolGeneralTwistorLikeEq}
\end{equation}
where $\tC_{B'}c^{k}_{A'_{1}...A'_{k}}=0$ (no sum over $k$) for all $k=0,...,n$.
\end{proposition}

\begin{proof}
By induction in $n$. The case $n=1$ was already shown in proposition \ref{Prop:TwistorLikeEq}.
Now suppose that \eqref{SolGeneralTwistorLikeEq} is true for $n-1$.
Let $\tau_{A'_{1}...A'_{n}}$ be such that $\tC_{(A'_{1}}\tau_{A'_{2}...A'_{n+1})}=0$. Contracting with 
$\tC^{A'_{n+1}}$, we get $\tC_{(A'_{1}}\rho_{A'_{2}...A'_{n})}=0$, where 
$\rho_{A'_{2}...A'_{n}}=\tC^{A'_{n+1}}\tau_{A'_{2}...A'_{n}A'_{n+1}}$. Since we are assuming that \eqref{SolGeneralTwistorLikeEq}
is true for $n-1$, $\rho_{A'_{1}...A'_{n-1}}$ can be expressed in terms of $\pi^{A'}$ and constants of integration 
$b^{k}_{A'_{1}...A'_{k}}$ as in the right side of \eqref{SolGeneralTwistorLikeEq}:
\begin{equation*}
 \rho_{A'_{1}...A'_{n-1}} = 
 b^{0}\pi_{A'_{1}}...\pi_{A'_{n}} + b^{1}_{(A'_{1}}\pi_{A'_{2}}...\pi_{A'_{n})} + 
  ...+b^{n-2}_{(A'_{1}...A'_{n-2}}\pi_{A'_{n-1)}} + b^{n-1}_{A'_{1}...A'_{n-1}}.
\end{equation*}
Each term in the right hand side of this equation can be written as 
\begin{align*}
 b^{k}_{(A'_{1}...A'_{k}}\pi_{A'_{k+1}}...\pi_{A'_{n-1})} = m_{k} \tC^{A'_{n}}(b^{k}_{(A'_{1}...A'_{k}}\pi_{A'_{k+1}}...\pi_{A'_{n})})
\end{align*}
where $m_{k}$ is some numerical factor, so $\rho_{A'_{1}...A'_{n-1}}=\tC^{A'_{n}}E_{A'_{1}...A'_{n-1}A'_{n}}$
where $E_{A'_{1}...A'_{n}}$ is the expression that results from these replacements.
But we also have $\rho_{A'_{1}...A'_{n-1}}=\tC^{A'_{n}}\tau_{A'_{1}...A'_{n-1}A'_{n}}$, 
thus, $\tau_{A'_{1}...A'_{n}}=E_{A'_{1}...A'_{n}}+\tC_{A'_{1}}...\tC_{A'_{n}}A$ for some $A$. 
Since $\tC_{(A'_{1}}\tau_{A'_{2}...A'_{n+1})}=0=\tC_{(A'_{1}}E_{A'_{2}...A'_{n+1})}$, we have 
$\tC_{A'_{1}}c^{n}_{A'_{2}...A'_{n+1}}=0$, where $c^{n}_{A'_{1}...A'_{n}}=\tC_{A'_{1}}...\tC_{A'_{n}}A$.
Replacing everything, the result follows.
\end{proof}

The result of proposition \ref{Prop:GeneralTwistorLikeEq} can be used to find the solution to \eqref{SpecialEqs2}
for higher order derivatives. The specific case we will need later is $\tC_{(A'}\tC_{B'}\tau_{C'D')}=0$, 
so we give the solution to this equation (the general case can be analysed following the same reasoning):

\begin{proposition}
Let $\tau_{A'B'}$ be a symmetric spinor field such that $\tC_{(A'}\tC_{B'}\tau_{C'D')}=0$.
Then $\tau_{A'B'}$ can be expressed as
\begin{align}
 \tau_{A'B'} = \pi_{A'}\pi_{B'}\pi^{C'}a^{1}_{C'} + \pi_{(A'}\pi^{C'}a^{2}_{B')C'}+\pi^{C'}a^{3}_{A'B'C'} + a^{4}_{A'B'} 
 \label{DoubleTwistorLikeEq}
\end{align}
where $\pi^{A'}$ is defined by \eqref{TwistorLikeEq} and the rest of the spinors in the right are constants of integration.
\end{proposition}

\begin{proof}
Putting $\nu_{B'C'D'}=\tC_{(B'}\tau_{C'D')}$, we have $\tC_{(A'}\nu_{B'C'D')}=0$, 
so $\nu_{A'B'C'}$ can be expressed as in the right hand side of \eqref{SolGeneralTwistorLikeEq}. 
In addition, the condition $\nu_{A'B'C'}=\tC_{(A'}\tau_{B'C')}$ gives $\tC^{A'}\tC^{B'}\tC^{C'}\nu_{A'B'C'}=0$, 
which implies that the term of highest order in $\pi_{A'}$ in \eqref{SolGeneralTwistorLikeEq} is absent  for $\nu_{A'B'C'}$, so 
\begin{align*}
 \nu_{A'B'C'} = b^{1}_{(A'}\pi_{B'}\pi_{C')} + b^{2}_{(A'B'}\pi_{C')} + b^{3}_{A'B'C'}.
\end{align*}
We have $b^{1}_{(A'}\pi_{B'}\pi_{C')}=\tC_{(A'}(\pi_{B'}\pi_{C')}b^{1}_{D'}\pi^{D'})$, etc., which leads to $\tC_{(A'}K_{B'C')}=0$, 
where $K_{A'B'}=\tau_{A'B'}-\pi_{A'}\pi_{B'}b^{1}_{D'}\pi^{D'}-\pi_{(A'}b^{2}_{B')D'}\pi^{D'}-b^{3}_{A'B'D'}\pi^{D'}$.
Using proposition \ref{Prop:GeneralTwistorLikeEq}, we get $K_{A'B'}=c^{0}\pi_{A'}\pi_{B'}+c^{1}_{(A'}\pi_{B')}+c^{2}_{A'B'}$, so
\begin{align*}
 \tau_{A'B'} = c^{0}\pi_{A'}\pi_{B'}+c^{1}_{(A'}\pi_{B')}+c^{2}_{A'B'}
 +\pi_{A'}\pi_{B'}b^{1}_{D'}\pi^{D'}+\pi_{(A'}b^{2}_{B')D'}\pi^{D'}+b^{3}_{A'B'D'}\pi^{D'}.
\end{align*}
Renaming the constants of integration and reorganizing, \eqref{DoubleTwistorLikeEq} follows.
\end{proof}

We will need the result \eqref{DoubleTwistorLikeEq} in section \ref{Sec:Reconstruction} 
(see proposition \ref{Prop:RelationSpinorAndScalarPotentials}).
For higher order twistor-like equations \eqref{SpecialEqs2}, the solution can be found in a similar way.

%%=======================================================

\section{Linear fields}\label{Sec:LinearFields}

In this section we assume that the conditions \eqref{SFR} and \eqref{rPND} are satisfied. 

\subsection{Dirac-Weyl fields}

The right-handed (RH) spin $1/2$ equation is $\nabla_{A}{}^{A'}\phi_{A'}=0$.
The weights of the field are $p(\phi_{A'})=0$, $w(\phi_{A'})=-1$.
Using identity \eqref{identityRHF}, we can write the field equation as $\mathcal{C}_{A}{}^{A'}\phi_{A'}=0$. 
Contracting with $o^{A}$, we get $\tC^{A'}\phi_{A'}=0$. Applying then lemma \ref{Lemma:Potentials}, 
we deduce that there exists, locally, a scalar field $\Phi$ with weights $p(\Phi)=-1$ and $w(\Phi)=-1$, 
such that $\phi_{A'}=\tC_{A'}\Phi$.
Notice that $\Phi$ is defined only up to the addition of a function $g$ such that $\tC_{A'}g=0$, 
see remark \ref{Remark:InternalGaugeFreedom}.
Contracting $\C_{A}{}^{A'}\phi_{A'}=0$ with $\iota^{A}$ and replacing $\phi_{A'}=\tC_{A'}\Phi$
gives $\C^{A'}\tC_{A'}\Phi=0$. Using identity \eqref{C-tC-Scalars} we get
$(\Box^{\mathcal{C}} + \mathcal{R}/4)\Phi=0$, 
where $\mathcal{R}$ is the `scalar curvature' of $\C_{a}$, see \eqref{ScalarCurvatureC}, \eqref{ScalarCurvatureCSpecial}.
Since $\phi_{A'}$ was arbitrary, the following result is obtained:
\begin{theorem}\label{Thm:DiracWeyl}
Any solution to the RH spin $1/2$ equation can be locally written as $\phi_{A'}=\tC_{A'}\Phi$, where $\Phi$ 
satisfies $(\Box^{\mathcal{C}} + \mathcal{R}/4)\Phi=0$ and is defined only up to the freedom \eqref{InternalGaugeFreedom}.
\end{theorem}

\begin{remark}
Using the expression \eqref{boxC2} for $\Box^{\C}$ in terms of GHP operators, we have
\begin{equation}
 (\Box^{\C}+\mathcal{R}/4)\Phi=2[(\tho'-\tilde\rho')\tho-(\edt'-\tilde\t)\edt]\Phi.  \label{DebyeEqDiracGHP}
\end{equation}
\end{remark}

\subsection{Maxwell fields}

\subsubsection{Preliminaries}\label{Sec:MaxwellPreliminaries}

Maxwell equations for a 2-form $F_{ab}$ are
\begin{subequations}
\begin{align}
 \nabla^{a}F_{ab} ={}& 4\pi J_{b}, \label{MaxwellEqs1} \\
 \nabla_{[a}F_{bc]} ={}& 0,\label{MaxwellEqs2}
\end{align}
\end{subequations}
where $J_{a}$ is assumed to be given and it physically represents the source.
Since we are assuming the manifold to be orientable, \eqref{MaxwellEqs2} can be written as $\nabla^{a} {}^{*}F_{ab}=0$, 
where ${}^{*}F_{ab}=\frac{1}{2}\varepsilon_{ab}{}^{cd}F_{cd}$. The spinor decomposition of the field strength is
\begin{align}
 F_{ab} = \phi_{A'B'}\epsilon_{AB} + \varphi_{AB}\epsilon_{A'B'}. \label{SpinorDecMaxwell}
\end{align}
$\phi_{A'B'}$ and $\varphi_{AB}$ are respectively the SD and ASD parts of $F_{ab}$, 
and we are ignoring reality conditions for the moment. 
Thus, Maxwell equations can be written in spinor terms as
\begin{subequations}
\begin{align}
 \nabla_{B}{}^{A'}\phi_{A'B'} + \nabla_{B'}{}^{A}\varphi_{AB} ={}& 4\pi J_{BB'}, \label{MaxwellEqsSpinors1} \\
 \nabla_{B}{}^{A'}\phi_{A'B'} - \nabla_{B'}{}^{A}\varphi_{AB} ={}& 0. \label{MaxwellEqsSpinors2}
\end{align}
\end{subequations}
Imposing \eqref{MaxwellEqsSpinors2}, and replacing in \eqref{MaxwellEqsSpinors1}:
\begin{equation}
 \nabla_{B}{}^{A'}\phi_{A'B'} = 2\pi J_{BB'}. \label{MaxwellEqs3}
\end{equation}
Although \eqref{MaxwellEqs3} may suggest that the ASD part $\varphi_{AB}$ disappears from the picture and 
is ``decoupled'' from the SD part, it is actually `hidden' in the fact that we are already using \eqref{MaxwellEqsSpinors2}:
this equation implies that $\nabla_{[a}F_{bc]}=0$, so $F_{ab}$ can (at least locally) be expressed in terms of
a vector potential $A_{a}$ by $F_{ab}=2\nabla_{[a}A_{b]}$. In spinor terms this is
\begin{equation}
 \phi_{A'B'} = \nabla_{A(A'}A^{A}_{B')}, \qquad \varphi_{AB} = \nabla_{A'(A}A^{A'}_{B)} \label{VectorPotential}
\end{equation}
so the SD and ASD parts are related by the vector potential.
Roughly speaking, if we solve \eqref{MaxwellEqs3} for $\phi_{A'B'}$, this should give us an expression 
for $A_{AA'}$, which then, replacing in the second equation in \eqref{VectorPotential}, should in turn give us 
an expression for $\varphi_{AB}$.

\medskip
Now, let us take into account reality conditions. 
For complex manifolds, the 2-form $F_{ab}$ is complex, and, if we do not impose anything else on $F_{ab}$, 
the SD and ASD parts of $F_{ab}$ are independent.
However, if we impose \eqref{MaxwellEqs2}, then $\phi_{A'B'}$ and $\varphi_{AB}$ are actually related by \eqref{VectorPotential}.
So, if we study the (complex) Maxwell equations only in terms of $\phi_{A'B'}$ via \eqref{MaxwellEqs3}, 
we can still get information of $\varphi_{AB}$ by using \eqref{VectorPotential}.
This point will be particularly important not only for Maxwell fields but also for the other more general fields 
that we will study in later sections.

\smallskip
For real manifolds, we want to impose $F_{ab}$ to be a real 2-form. Since we are working with spinors, 
this reality condition will take different forms according to signature.
For Riemann signature, we mentioned that there is a notion of `real spinors' via the involutive 
Euclidean spinor complex conjugation $\dagger$.
For split signature, spinors are themselves real objects.
So, for these two cases, the discussion about the SD and ASD parts of $F_{ab}$ is very similar to the 
complex case: we can study the Maxwell equation \eqref{MaxwellEqs3} for the SD part $\phi_{A'B'}$, 
and the ASD part will be given by the second eq. in \eqref{VectorPotential} once we determine an expression for $A_{AA'}$.
The only difference is that the spinors are now required to be real.

\smallskip
For Lorentz signature, there are no real spinors, since complex conjugation changes chirality.
The reality requirement for the 2-form $F_{ab}$ translates into the condition 
that the SD and ASD parts must be complex-conjugated: $\varphi_{AB}=\bar{\phi}_{AB}$.
This means that, for this case, we can focus entirely on $\phi_{A'B'}$: once we solve \eqref{MaxwellEqs3} for it, 
the full Maxwell field $F_{ab}$ will be given by \eqref{SpinorDecMaxwell} with $\varphi_{AB}=\bar{\phi}_{AB}$.

\medskip
Since we are leaving the signature arbitrary, in what follows we will study the SD and ASD parts separately.
This will also be useful for later sections, where the treatment will be very similar.
For the SD part, we will consider both trivial and non-trivial sources; but for the ASD part 
we will focus only on the source-free case.
In addition, for the treatment of the ASD field we will need to assume that this part is algebraically special 
(except in Lorentz signature).

\subsubsection{The self-dual part}\label{Sec:SDMaxwell}

Here we will study the system 
\begin{equation}
 \nabla_{B}{}^{A'}\phi_{A'B'} = 0. \label{RHspin1}
\end{equation}
Notice that this equation, by itself, is equivalent to the vacuum Maxwell equations {\em only in Lorentz signature, 
and not in the other cases}. For this reason, we will refer to \eqref{RHspin1} as the ``spin 1 system''.
This difference will be important for the study of the ASD field $\varphi_{AB}$ in section \ref{Sec:ASDMaxwell}.

\subsubsection*{The source-free case}

\begin{theorem}\label{Theorem:Maxwell:Sourcefree}
Any solution to the spin 1 system \eqref{RHspin1} can be locally written as 
\begin{equation}\label{spin1-vacuum}
 \phi_{A'B'}=\tilde{\mathcal{C}}_{A'}\tilde{\mathcal{C}}_{B'}\Phi
\end{equation}
where $\Phi$ is a scalar field with weights $p(\Phi)=-2$ and $w(\Phi)=-1$, defined only up to the 
freedom \eqref{InternalGaugeFreedom}, that solves the equation 
\begin{equation}\label{debye-spin1}
 (\Box^{\mathcal{C}}+\mathcal{R}/2)\Phi=0,
\end{equation}
where $\mathcal{R}$ is the scalar curvature of $\C_{a}$ \eqref{ScalarCurvatureC}-\eqref{ScalarCurvatureCSpecial}.
\end{theorem}

\begin{remark}
Using identity \eqref{boxC2}, in GHP notation we have
\begin{equation}
 (\Box^{\C}+\mathcal{R}/2)\Phi=2[(\tho'-\tilde\rho')(\tho+\rho)-(\edt'-\tilde\t)(\edt+\t)]\Phi. \label{DebyeEqMaxwellGHP}
\end{equation}
\end{remark}

\begin{proof}[Proof of Theorem \ref{Theorem:Maxwell:Sourcefree}]
Let $\phi_{A'B'}$ be a solution to \eqref{RHspin1}.
From \eqref{identityRHF} we know $\nabla_{B}{}^{A'}\phi_{A'B'}=\mathcal{C}_{B}{}^{A'}\phi_{A'B'}$. 
Contraction with $o^{B}$ leads to $\tC^{A'}\phi_{A'B'}=0$.
Applying lemma \ref{Lemma:Potentials}, and using the symmetry $\phi_{A'B'}=\phi_{(A'B')}$,
we deduce that there exists, locally, 
a scalar field $\mr{\Phi}$, with weights $p(\mr\Phi)=-2$ and $w(\mathring{\Phi})=-1$, such that
\begin{equation}
 \phi_{A'B'}=\tilde{\C}_{A'}\tilde{\C}_{B'}\mr\Phi. \label{spin1-0}
\end{equation}
Contracting now $\mathcal{C}_{B}{}^{A'}\phi_{A'B'}=0$ with $\iota^{B}$, we get $\mathcal{C}^{A'}\phi_{A'B'}=0$, 
so replacing the expression above for $\phi_{A'B'}$:
\begin{equation}
 0 = \C^{A'}\phi_{A'B'} =\C^{A'}\tC_{A'}\tilde{\C}_{B'}\mr\Phi  = \tC_{B'}\C^{A'}\tC_{A'}\mr\Phi + [\tC_{B'},\C_{A'}]\tC^{A'}\mr\Phi.
 \label{EqWithCommutator1}
\end{equation}
Noticing that $p(\tC^{A'}\mr\Phi)=-1$ and applying identity \eqref{ContractedCommutatorPS}, 
the commutator term on the right vanishes. Using then \eqref{C-tC-Scalars}, we get:
\begin{equation}
 \tilde{\C}_{B'}(\Box^{\C}+\mathcal{R}/2)\mr\Phi=0. \label{cond1}
\end{equation}
Thus, $(\Box^{\C}+\mathcal{R}/2)\mr\Phi \equiv 2K$, where $K$ is a constant of integration: $\tilde{\C}_{A'}K=0$ 
(see section \ref{Sec:ConstantsOfIntegration}).
Notice that $K=\C^{A'}\tilde{\C}_{A'}\mr\Phi$.
Now, similarly to the ``gauge freedom'' \eqref{InternalGaugeFreedom},
in the expression \eqref{spin1-0} for $\phi_{A'B'}$ we see that we can perform a ``gauge transformation''
\begin{equation}
 \tC_{B'}\mr\Phi \to \tC_{B'}\mr\Phi + g_{B'}, \label{InternalGaugeSpin1}
\end{equation}
where $\tC_{(A'}g_{B')}=0$, and the field \eqref{spin1-0} remains invariant, but $K$ is changed according to
\begin{equation*}
 K=\C^{A'}\tC_{A'}\mr\Phi \;\; \to \;\; K' = \C^{A'}\tC_{A'}\mr\Phi+\C^{A'}g_{A'}=K+\C^{A'}g_{A'}.
\end{equation*}
In particular, choose $g_{A'}=\tC_{A'}a$ for some auxiliary field $a$.
Then $K'=K+\C^{A'}\tC_{A'}a = K+\frac{1}{2}(\Box^{\C}+\mathcal{R}/2)a$. 
The solution to the problem $(\Box^{\C}+\mathcal{R}/2)a=-2K$ always exists (at least locally), 
since it is a wave equation with source.
In addition, replacing in \eqref{InternalGaugeSpin1} we have $\tC_{B'}\mr\Phi + g_{B'} = \tC_{B'}\Phi$, 
where $\Phi=\mr\Phi+a$. 
The field \eqref{spin1-0} is then $\phi_{A'B'}=\tC_{A'}\tC_{B'}\Phi$, 
and $\Phi$ satisfies $\C^{A'}\tC_{A'}\Phi=\frac{1}{2}(\Box^{\C}+\mathcal{R}/2)\Phi=0$, 
thus the result follows.
\end{proof}

\subsubsection*{The case with sources}

For general sources in \eqref{MaxwellEqs3}, it is clear that a field of the form \eqref{spin1-vacuum} is not a solution, 
since contracting with $o^{B}$ and replacing \eqref{spin1-vacuum} leads to
\begin{equation}
 2\pi o^{B}J_{BB'}=o^{B}\nabla_{B}{}^{A'}\phi_{A'B'}=\tC^{A'}\phi_{A'B'} = \tC^{A'}\tC_{A'}\tC_{B'}\Phi \equiv 0,
\end{equation}
therefore, if $o^{B}J_{BB'}\neq0$, this equation is inconsistent. 
In order to solve this issue, we follow a strategy similar to the one adopted in \cite{Hollands2020}, 
which involves the introduction of a `corrector' potential.
Our approach is however slightly different in that it does not assume any particular form for the 
`corrector term', and it only uses the integration machinery given in Section \ref{Sec:deRham}.

\begin{theorem}
The solution $\phi_{A'B'}$ to the equation \eqref{MaxwellEqs3} for a given source $J_{BB'}$ can be locally written as 
\begin{equation}\label{solution:MaxwellWithSources}
 \phi_{A'B'}=\tC_{A'}\tC_{B'}\Phi+x_{A'B'},
\end{equation}
where $x_{A'B'}$ is determined from $J_{BB'}$ as a particular solution to the equation
\begin{equation}\label{xJ}
 \tC^{A'}x_{A'B'}=2\pi o^{B}J_{BB'},
\end{equation}
and $\Phi$ is a scalar field with weights $p(\Phi)=-2$ and $w(\Phi)=-1$ that solves the inhomogeneous equation
\begin{equation}
 \tfrac{1}{2}(\Box^{\mathcal{C}}+\mathcal{R}/2)\Phi=T
\end{equation}
with the source $T$ determined by $x_{A'B'}$.
\end{theorem}

\begin{proof}
Suppose that the source $J_{BB'}$ in \eqref{MaxwellEqs3} is given. 
We want to find the general solution $\phi_{A'B'}$ to this equation.
Applying item 2. in Lemma \ref{Lemma:Potentials} for the field $\theta^{\mathscr{A}}\equiv 2\pi o^{B}J_{BB'}$,
we know that there exists $x_{A'B'}=x_{(A'B')}$ such that eq. \eqref{xJ} holds.
We can choose any particular solution to this equation.
Using identity \eqref{identityRHF}, eq. \eqref{xJ}
implies that $o^{B}\nabla_{B}{}^{A'}x_{A'B'}=2\pi o^{B}J_{BB'}$, therefore
\begin{equation}\label{difference}
 o^{B}\nabla_{B}{}^{A'}(\phi_{A'B'}-x_{A'B'})=0.
\end{equation}
Now define
\begin{equation*}
 S_{BB'}:=2\pi J_{BB'}-\nabla_{B}{}^{A'}x_{A'B'}.
\end{equation*}
From \eqref{MaxwellEqs3} and \eqref{difference} we see that $o^{B}S_{BB'}=0$, therefore $S_{BB'}=o_{B}y_{B'}$
where $y_{B'}=\iota^{B}S_{BB'}=2\pi J_{1B'}-\iota^{B}\nabla_{B}{}^{A'}x_{A'B'}$. 
This $y_{B'}$ is known once we integrate \eqref{xJ} to find $x_{A'B'}$.
Now, a straightforward computation shows that $\nabla^{BB'}S_{BB'}=0$.
But we also have the identity $\nabla^{BB'}S_{BB'}=\C^{BB'}S_{BB'}$, thus, replacing $S_{BB'}=o_{B}y_{B'}$, 
we get $\tC^{B'}y_{B'}=0$, from where we deduce that there exists, locally, a scalar field $T$ such that
\begin{equation}\label{Defzeta}
 y_{B'}=\tC_{B'}T.
\end{equation}
Now we use $T$ as a source for the equation
\begin{equation*}
 T=\C^{A'}\tC_{A'}\Phi=\tfrac{1}{2}(\Box^{\C}+\mathcal{R}/2)\Phi
\end{equation*}
where $\Phi$ is undetermined at the moment.
Since $T$ is known from \eqref{Defzeta} (assuming we already determined $x_{A'B'}$), this equation can 
be solved for $\Phi$. Therefore:
\begin{align*}
 S_{BB'}=o_{B}\tC_{B'}\C^{A'}\tC_{A'}\Phi = 2\pi J_{BB'}-\nabla_{B}{}^{A'}x_{A'B'},
\end{align*}
from which we get, after a straightforward calculation, $2\pi J_{BB'}=\nabla_{B}{}^{A'}\mr{\phi}_{A'B'}+\nabla_{B}{}^{A'}x_{A'B'}$, 
where $\mr{\phi}_{A'B'}=\tC_{A'}\tC_{B'}\Phi$. Thus the result \eqref{solution:MaxwellWithSources} follows.
\end{proof}

\subsubsection{The anti-self-dual part}\label{Sec:ASDMaxwell}

From the discussion in section \ref{Sec:MaxwellPreliminaries}, recall that for Lorentz signature 
we do not need to worry about different treatments for the SD and ASD parts of the Maxwell field, 
since they are related by complex conjugation. 
However, for all the other cases, even if we solve the problem \eqref{RHspin1} for $\phi_{A'B'}$ 
we still must consider separately the issue of determining the ASD field $\varphi_{AB}$. 
We will see here that we can still ``determine'' the ASD part $\varphi_{AB}$ {\em as long as it is algebraically special}.
More precisely, provided this is the case, $\varphi_{AB}$ will be given in terms of arbitrary constants of integration.

\smallskip
First, notice that since $A_{a}$ has vanishing $p$ and $w$ weights, we have $\nabla_{[a}A_{b]}=\C_{[a}A_{b]}$, 
so the SD and ASD parts of $F_{ab}$ can be expressed as
$\phi_{A'B'}=\C_{A(A'}A^{A}_{B')}$, $\varphi_{AB}=\C_{A'(A}A^{A'}_{B)}$. Equivalently:
\begin{subequations}
\begin{align}
 \phi_{A'B'} ={}& \tC_{(A'}A_{B')} - \C_{(A'}\tilde{A}_{B')}, \label{SDMaxwell0} \\
 \varphi_{AB} ={}& o_{A}o_{B}(\C_{A'}A^{A'}-\sigma^{1}_{A'}\tilde{A}^{A'}) -o_{(A}\iota_{B)}(\C_{A'}\tilde{A}^{A'}+\tC_{A'}A^{A'})
 +\iota_{A}\iota_{B}\tC_{A'}\tilde{A}^{A'}, \label{ASDMaxwell0}
\end{align}
\end{subequations}
where $\tilde{A}_{A'}=o^{A}A_{AA'}$, $A_{A'}=\iota^{A}A_{AA'}$, and $\sigma^{1}_{A'}$ is defined in \eqref{sigma1}.
Now, {\em assume that $\varphi_{AB}$ is algebraically special} with repeated principal spinor $o^{A}$, 
that is $o^{A}o^{B}\varphi_{AB}=0$.
From \eqref{ASDMaxwell0} we then see that $\tC_{A'}\tilde{A}^{A'}=0$, so there exists (locally)
a scalar field $\mu$, with weights $p(\mu)=0$, $w(\mu)=0$, such that $\tilde{A}_{A'}=\tC_{A'}\mu$. 
But then we can get rid of $\mu$ by a gauge transformation $A_{a} \to A_{a} - \nabla_{a}\mu$, 
so we can simply put $\tilde{A}_{A'}=0$. 
The vector potential is then $A_{AA'}=o_{A}A_{A'}$, and we can focus on the determination of $A_{A'}$.
The Maxwell field becomes
\begin{subequations}
\begin{align}
 \phi_{A'B'} ={}& \tC_{(A'}A_{B')}, \label{SDMaxwell1} \\
 \varphi_{AB} ={}& o_{A}o_{B}\C_{A'}A^{A'} -o_{(A}\iota_{B)}\tC_{A'}A^{A'}. \label{ASDMaxwell1}
\end{align}
\end{subequations}

Now, the vacuum Maxwell equations are $\nabla_{B}{}^{A'}\phi_{A'B'}=\nabla_{B'}{}^{A}\varphi_{AB}$
and $\nabla_{B}{}^{A'}\phi_{A'B'}=0$.  From the latter equation, we know that 
$\phi_{A'B'}=\tC_{A'}\tC_{B'}\Phi$, for some locally defined scalar field $\Phi$.
Comparing this expression to \eqref{SDMaxwell1}, we deduce
\begin{equation}
 A_{A'} = \tC_{A'}\Phi + \tau_{A'}, \label{ConnectionMaxwell}
\end{equation}
where $\tau_{A'}$ is defined only by the condition
\begin{equation}
 \tC_{(A'}\tau_{B')} = 0. \label{tauMaxwell}
\end{equation}
The formal solution to this equation is given by proposition \ref{Prop:TwistorLikeEq} 
in terms of constants of integration and the solution to \eqref{TwistorLikeEq}.
On the other hand, as in the proof of theorem \ref{Theorem:Maxwell:Sourcefree}, 
the equation $\C^{A'}\phi_{A'B'}=0$
leads to $\C^{A'}\tC_{A'}\Phi=K$, where $K$ is a constant of integration, $\tC_{A'}K=0$.
For the spin 1 system \eqref{RHspin1}, we noticed that we had the ``internal gauge freedom'' \eqref{InternalGaugeSpin1}, 
which we used to set $K=0$.
In our current situation, equations \eqref{ConnectionMaxwell}-\eqref{tauMaxwell} seem to be 
the analogue of \eqref{InternalGaugeSpin1}, however, in this case {\em this is not a `gauge freedom'}: 
while the SD curvature $\phi_{A'B'}$ is insensitive to $\tau_{A'}$, the field $\tau_{A'}$ appears explicitly in 
the ASD part $\varphi_{AB}$, as can be seen from replacing \eqref{ConnectionMaxwell} into \eqref{ASDMaxwell1}.
More explicitly, using $\C_{A'}\tC^{A'}\Phi=-K$, we find:
\begin{equation}
 \varphi_{AB} = o_{A}o_{B}(-K+\C_{A'}\tau^{A'}) - o_{(A}\iota_{B)}\tC_{A'}\tau^{A'}. \label{ASDMaxwell2}
\end{equation}
In particular, the trick that we used in the spin 1 case (i.e. $\tau^{A'}=\tC^{A'}a$ and $\C^{A'}\tC_{A'}a=-K$)
now leads to a self-dual Maxwell field, that is $\varphi_{AB}\equiv 0$.
In the present section, however, we do not have this freedom, and the ASD Maxwell field $\varphi_{AB}$ 
depends on the constants of integration $K$ and $\tau_{A'}$.
This difference between the SD and ASD fields will also appear 
in the Yang-Mills and gravitational systems.

\subsubsection{Parallel frames}\label{Sec:ParallelFramesMaxwell}

Although for the Maxwell case the treatment of the previous sections is sufficient, 
an analogous strategy for the Yang-Mills and Einstein-Weyl systems is much harder,
so it is convenient to give an alternative formulation of the analysis above in terms of ``parallel frames'', 
since this will extend to the other more complicated cases.

\smallskip
We now think of Maxwell theory as a gauge theory. 
That is, we consider a principal bundle $P^{1}$ over $M$ with a 1-dimensional structure group $G^{1}$;
the vector potential then arises by assuming that $P^{1}$ is equipped with a local connection 1-form, $A_{a}$. 
The `charge $e$ representations' $\rho_{e}:G^{1}\to{\rm GL}(\mathbb{C})$ are given by
$\rho_{e}(a)\psi = a^{e}\psi$, with $a\in G^{1}$, $\psi\in\mathbb{C}$, 
and a charged scalar field with charge $e$ is simply a section of the associated vector bundle 
$E_{e}=P^{1}\times_{\rho_{e}}\mathbb{C}$.
By taking the tensor product of $E_{e}$ with the various spinor bundles, we obtain 
charged spinor fields as the corresponding sections.
The fact that $A_{a}$ is now seen as a connection form implies that it induces a covariant derivative 
on these bundles, ${\rm D}_{a}$. 
That is, if a spinor field $\psi^{\mathscr{A}}$ has charge $e$, then
${\rm D}_{a}\psi^{\mathscr{A}} = \nabla_{a}\psi^{\mathscr{A}}+e A_{a}\psi^{\mathscr{A}}$
The curvature of ${\rm D}_{a}$, given by the commutator $[{\rm D}_{a},{\rm D}_{b}]$, involves both the spacetime curvature 
and the curvature of the connection $A_{a}$, the latter being given by $2\nabla_{[a}A_{b]}=F_{ab}$.

\smallskip
We can then consider spinor fields which, in addition to the electromagnetic charge $e$, have also weights $p$ and $w$.
The covariant derivative on such objects can be obtained by simply replacing $\nabla_{a}$ by 
${\rm D}_{a}$ in the previous expressions for the complex-conformal connection $\C_{a}$, 
that is, one defines 
\begin{equation}
 \Cym_{a} := \C_{a}+ eA_{a}.
\end{equation}
For fields with charge $e=0$, $\Cym_{a}$ reduces to $\C_{a}$.
We also introduce the projections $\tilde\Cym_{A'} = o^{A}\Cym_{AA'}$ and $\Cym_{A'} = \iota^{A}\Cym_{AA'}$.
The curvature is $[\Cym_{a},\Cym_{b}] = [\C_{a},\C_{b}]+eF_{ab}$.
Let $\alpha$ be a scalar field with charge $e(\alpha)=1$ and weights $p(\alpha)=0$, $w(\alpha)=0$,
then $[\Cym_{a},\Cym_{b}] \alpha = F_{ab}\alpha$. Assuming $\alpha\neq 0$ in some region $U$, 
we get $F_{ab}=\alpha^{-1}[\Cym_{a},\Cym_{b}] \alpha$.
The SD and ASD fields are 
\begin{subequations}
\begin{align}
 \phi_{A'B'} ={}& \alpha^{-1}\Box^{\Cym}_{A'B'}\alpha, \label{SDBoxCMaxwell} \\
 \varphi_{AB} ={}& \alpha^{-1}\Box^{\Cym}_{AB}\alpha. \label{ASDBoxCMaxwell}
\end{align}
\end{subequations}

Now, from $[\Cym_{a},\Cym_{b}] = [\C_{a},\C_{b}]+eF_{ab}$ we see that the new operator $\tilde\Cym_{A'}$ satisfies
$[\tilde\Cym_{A'},\tilde\Cym_{B'}] = \epsilon_{A'B'} e o^{A}o^{B} \varphi_{AB}$.
Therefore, we deduce that $[\tilde\Cym_{A'},\tilde\Cym_{B'}]=0$ if and only if the ASD part 
is algebraically special, $o^{A}o^{B} \varphi_{AB}=0$. 
In such case, the integration technique from section \ref{Sec:deRham} can be generalized 
to the operator $\tilde\Cym_{A'}$.
In addition, a key point is that the equation $\tilde\Cym^{A'}\tilde\Cym_{A'}=0$ allows us to consider 
fields that are {\em parallel} under $\tilde\Cym_{A'}$, so we can take
\begin{equation}
 \tilde\Cym_{A'}\alpha=0. \label{ParallelFrameMaxwell}
\end{equation}
The generalization of this equation will correspond in later sections to {\em parallel frames} 
(see e.g. \cite[Theorem 2.4.1]{WW} for the connection between flat connections and parallel frames).
Using identities \eqref{BoxYMSD0}-\eqref{BoxYMASD0}, the SD and ASD fields are then:
\begin{subequations}
\begin{align}
 \phi_{A'B'} ={}& \tilde\Cym_{(A'}\gamma_{B')}, \label{CSDMaxwell1} \\
 \varphi_{AB} ={}& o_{A}o_{B}\Cym_{A'}\gamma^{A'} - o_{(A}\iota_{B)}\tilde\Cym_{A'}\gamma^{A'}, \label{CASDMaxwell1}
\end{align}
\end{subequations}
where we define
\begin{equation}
 \gamma_{A'}:=\alpha^{-1}\Cym_{A'}\alpha. \label{gammaMaxwell0}
\end{equation}
The field $\gamma_{A'}$ is essentially\footnote{From the definition of $\Cym_{a}$, 
we have $\gamma_{A'}=\iota^{A}\partial_{AA'}\log\alpha+A_{A'}$, so $\gamma_{A'}$ is gauge-equivalent to $A_{A'}$.} 
$A_{A'}$ from section \ref{Sec:ASDMaxwell}, but now expressed 
in terms of a ``parallel frame'' \eqref{ParallelFrameMaxwell}.
The charge of $\gamma_{A'}$ is $e=0$, so in \eqref{CSDMaxwell1}-\eqref{CASDMaxwell1} we can 
replace $\Cym_{a}$ by $\C_{a}$ and, thus, these formulas coincide 
with the earlier expressions \eqref{SDMaxwell1}-\eqref{ASDMaxwell1}, and the analysis is the same as before.
The current formulation has the advantage that it will extend to the Yang-Mills and Einstein-Weyl systems.

\begin{remark}[Charged Dirac fields]
The formulation in this section allows us also to consider the possibility of spacetime fields 
that are charged with respect to the electromagnetic field. 
For example, the charged Dirac-Weyl equation is
\begin{equation}
 {\rm D}_{A}{}^{A'}\phi_{A'} = 0, \label{ChargedDiracEq}
\end{equation}
where the field $\phi_{A'}$ is assumed to have a non-trivial charge $e(\phi_{A'})=e$. 
Equation \eqref{ChargedDiracEq} can be written as $\Cym_{A}{}^{A'}\phi_{A'}=0$, or
$\tilde\Cym^{A'}\phi_{A'}=0$ and $\Cym^{A'}\phi_{A'}=0$. 
The former equation leads to $\phi_{A'}=\tilde\Cym_{A'}\chi$, where $p(\chi)=-1$, $w(\chi)=-1$, and $e(\chi)=e$.
The other equation is $\Cym^{A'}\tilde\Cym_{A'}\chi=0$. In terms of wave operators and curvature, 
a short calculation shows that this is
\begin{equation}
 (\Box^{\Cym}+\mathcal{R}/4 - 2e\varphi_{01})\chi=0,
\end{equation}
where $\Box^{\Cym}=g^{ab}\Cym_{a}\Cym_{b}$ and $\varphi_{01}=o^{A}\iota^{B}\varphi_{AB}$.
\end{remark}

\section{Yang-Mills fields}\label{Sec:YangMills}

We assume that conditions \eqref{SFR} and \eqref{rPND} are satisfied.

\subsection{Preliminaries}

Let $P$ be a principal bundle over $M$, whose structure group is some Lie group $G$.
We consider the natural representation of $G$ in which its elements act as matrices on some $n$-dimensional vector space $V$.
We will use abstract indices $i, j, k,...$ to denote elements of $V$.
We refer to this kind of indices as `Yang-Mills' or `internal' indices, and we make the convention 
that they {\em do not} include spacetime indices.
The action of $m\in G$ on an element $\mu^{i}\in V$ is  
$\mu^{i}\mapsto m_{j}{}^{i}\mu^{j}$. 
We then take the associated vector bundles $E=P\times_{G}V$, and a spacetime field that 
is `charged with respect to $G$' is a section of $E$. 
As usual, fields with more Yang-Mills indices are obtained by simply considering tensor products of $V$ and its dual, 
together with the corresponding representations of $G$.
More generally, fields with mixed spacetime/internal indices are sections of 
vector bundles that are obtained as a tensor product of the corresponding spinor/tensor bundle and 
the Yang-Mills bundle.

\smallskip
Assume that $P$ is equipped with a local connection 1-form $A$. By definition, this takes values 
in the Lie algebra ${\rm Lie}(G)$, so in abstract indices we can write $A_{a i }{}^{j}$.
The covariant derivative induced on associated vector bundles, which we denote by ${\rm D}_{a}$
as in the Maxwell case, is given by
\begin{equation}\label{Def:YMderivative}
{\rm D}_{a}\mu^{\mathscr{A} i} = \nabla_{a}\mu^{\mathscr{A} i} + A_{a j}{}^{i}\mu^{\mathscr{A} j},
\end{equation}
where $\mu^{\mathscr{A}i}$ is arbitrary, and $\nabla_{a}$ on the right hand side acts only on spacetime indices $\mathscr{A}$.
The curvature of ${\rm D}_{a}$ can be computed from 
\eqref{Def:YMderivative} by applying an additional derivative and taking the commutator:
\begin{equation*}
 [{\rm D}_{a},{\rm D}_{b}] \mu^{\mathscr{A}i}
 = [\nabla_{a},\nabla_{b}]\mu^{\mathscr{A}i} + F_{ab j}{}^{i}\mu^{\mathscr{A} j},
\end{equation*}
where we define the Yang-Mills curvature by 
\begin{equation}\label{Def:YMcurvature}
  F_{abi}{}^{j} = 2\nabla_{[a}A_{b]i}{}^{j} + 2A_{[a|k|}{}^{j}A_{b]i}{}^{k}.
\end{equation}
The skew-symmetry $F_{ab i}{}^{j}=F_{[ab]i}{}^{j}$ allows 
the following spinor decomposition:
\begin{equation}\label{YM-spinordecomposition}
F_{ab i}{}^{j} = \chi_{A'B' i }{}^{j}\epsilon_{AB} + \varphi_{AB i}{}^{j}\epsilon_{A'B'},
\end{equation}
where $\chi_{A'B' i}{}^{j} = \tfrac{1}{2}\epsilon^{AB}F_{ab i}{}^{j}$ 
and $\varphi_{AB i}{}^{j} = \tfrac{1}{2}\epsilon^{A'B'}F_{abi}{}^{j}$
are, respectively, the SD and ASD pieces of $F_{abi}{}^{j}$.
Notice that reality conditions for $F_{abi}{}^{j}$ are subtle 
since they depend on the gauge group $G$, which we are here leaving unspecified. 
Thus, we will consider $\chi_{A'B'i}{}^{j}$ and $\varphi_{ABi}{}^{j}$ 
to be independent, so in this sense the situation is analogous to the Maxwell case in section \ref{Sec:ASDMaxwell}.

Using the identity $[{\rm D}_{[c},{\rm D}_{a}]{\rm D}_{b]}\mu^{i} = {\rm D}_{[c}[{\rm D}_{a},{\rm D}_{b]}]\mu^{i}$,
one deduces\footnote{On spacetime indices, $[{\rm D}_{a},{\rm D}_{b}]$ acts as usual via the Riemann tensor, 
which satisfies $R_{[cab]}{}^{d}=0$.} 
that $F_{abi}{}^{j}$ satisfies the Bianchi identities: ${\rm D}_{[c}F_{ab]i}{}^{j}=0$.
Alternatively, these can be expressed as ${\rm D}^{a} {}^{*}F_{abi}{}^{j}=0$, or in spinors:
${\rm D}_{B}{}^{A'}\chi_{A'B'i}{}^{j} = {\rm D}_{B'}{}^{A}\varphi_{ABi}{}^{j}$. 
The Yang-Mills field equations are ${\rm D}^{a}F_{abi}{}^{j}=0$. In spinor terms
(using Bianchi identities):
\begin{align}
  {\rm D}_{B}{}^{A'}\chi_{A'B'i}{}^{j} ={}& 0, \label{YMFieldEq} \\
  {\rm D}_{B'}{}^{A}\varphi_{ABi}{}^{j} ={}& 0. \label{YMFieldEqASD}
\end{align}

The complex-conformal connection of section \ref{Sec:ConnectionC} can be generalized 
to act covariantly also on Yang-Mills indices. This generalization is achieved by simply defining 
\begin{equation}\label{YMC}
 \Cym_{a}\mu^{\mathscr{A}i} = \C_{a}\mu^{\mathscr{A}i} + A_{aj}{}^{i}\mu^{\mathscr{A}j},
\end{equation}
together with the obvious extension for more Yang-Mills indices.
We also define the projections $\tilde\Cym_{A'} = o^{A}\Cym_{AA'}$ and $\Cym_{A'} = \iota^{A}\Cym_{AA'}$.
The weights of the Yang-Mills potential are $p(A_{ai}{}^{j})=0$, $w(A_{ai}{}^{j})=0$.
The curvature is $[\Cym_{a},\Cym_{b}]\mu^{\mathscr{A}i} 
 = [\C_{a},\C_{b}]\mu^{\mathscr{A}i}+F_{abj}{}^{i}\mu^{\mathscr{A}j}$.
Similarly to the Maxwell case, we find $[\tilde\Cym_{A'},\tilde\Cym_{B'}]\mu^{\mathscr{A}i} 
 =\epsilon_{A'B'}o^{A}o^{B} \varphi_{ABj}{}^{i}\mu^{\mathscr{A}j}$.
\begin{proposition}\label{Prop:YM}
Assume that $o^{A}$ satisfies \eqref{SFR} and \eqref{rPND}, and
that the ASD Yang-Mills curvature satisfies 
\begin{equation}\label{conditionASDYM}
 o^{A}o^{B}\varphi_{ABi}{}^{j}=0.
\end{equation}
Then for any field $\mu^{\mathscr{A}i...}_{j...}$, it holds
\begin{equation}\label{CtYMflat}
 [\tilde\Cym_{A'}, \tilde\Cym_{B'}]\mu^{\mathscr{A}i...}_{j...}=0. 
\end{equation}
\end{proposition}

Using this result, we now proceed to impose the Yang-Mills field equations \eqref{YMFieldEq}-\eqref{YMFieldEqASD}
and analyse the potentialization of the system.

\subsection{The self-dual part of the curvature}\label{Sec:YMSD}

The Yang-Mills curvature has weights $p(F_{abi}{}^{j})=0$, $w(F_{abi}{}^{j})=0$. 
It then follows that $\Cym_{[a}F_{bc]i}{}^{j}={\rm D}_{[a}F_{bc]i}{}^{j}$, 
so we can write the Bianchi identities as $\Cym_{[a}F_{bc]i}{}^{j}=0$, or 
$\Cym^{a} {}^{*}F_{abi}{}^{j}=0$. 
Similarly, the field equations can be written as $\Cym^{a} F_{abi}{}^{j}=0$. 
In spinor terms:
\begin{align}
 \Cym_{B}{}^{A'}\chi_{A'B'i}{}^{j} ={}& 0, \label{YMFE1} \\
 \Cym_{B'}{}^{A}\varphi_{ABi}{}^{j} ={}& 0. \label{YMFE2}
\end{align}
Notice that $p(\chi_{A'B'i}{}^{j})=0$ and $w(\chi_{A'B'i}{}^{j})=-1$, 
and the same for $\varphi_{ABi}{}^{j}$.

We can summarize the results of this subsection as follows:
\begin{theorem}\label{Thm:YM}
Assume that the Yang-Mills curvature satisfies \eqref{conditionASDYM},  
and denote $\varphi_{01i}{}^{j}=o^{A}\iota^{B}\varphi_{ABi}{}^{j}$.
Any solution to the Yang-Mills field equation \eqref{YMFieldEq} can be locally written as
\begin{equation}\label{PotentialYM}
 \chi_{A'B'i}{}^{j} = \tilde\Cym_{A'}\tilde\Cym_{B'}\Phi_{i}{}^{j},
\end{equation}
where $\Phi_{i}{}^{j}$ is a (spacetime scalar) field with weights 
$p(\Phi_{i}{}^{j})=-2$ and $w(\Phi_{i}{}^{j})=-1$, such that 
\begin{equation}\label{HHYM}
 (\Box^{\Cym} +\mathcal{R}/2)\Phi_{i}{}^{j} 
 - 2(\tilde\Cym_{A'}\Phi_{i}{}^{k})(\tilde\Cym^{A'}\Phi_{k}{}^{j})
 +4(\varphi_{01i}{}^{k}\Phi_{k}{}^{j}-\varphi_{01k}{}^{j}\Phi_{i}{}^{k})
 = 2K_{i}{}^{j} 
\end{equation}
where $\Box^{\Cym}=g^{ab}\Cym_{a}\Cym_{b}$, and $K_{i}{}^{j}$ 
is a constant of integration, $\tilde\Cym_{A'}K_{i}{}^{j} =0$.
\end{theorem}

\begin{remark}
The potentialization \eqref{PotentialYM} of the Yang-Mills field, together with 
the equation \eqref{HHYM} for the potential, 
are the generalization of the Maxwell results \eqref{spin1-vacuum}-\eqref{debye-spin1} to a non-abelian gauge theory. 
(Notice that for the abelian case, the non-linear term in \eqref{HHYM}, as well as the coupling term between 
$\varphi_{01i}{}^{j}$ and $\Phi_{i}{}^{j}$, are absent.)
These results are a conformally invariant generalization of Jeffryes' results \cite[Eqs. (4.11) and (4.19)]{Jeffryes}.
\end{remark}

In the rest of this subsection we will prove the result \ref{Thm:YM}. 
The proof is similar to the Maxwell case, theorem \ref{Theorem:Maxwell:Sourcefree}.
The first observation is that, contracting \eqref{YMFE1} with $o^{B}$, we get $\tilde\Cym^{A'}\chi_{A'B'i}{}^{j}=0$, 
so the potentialization \eqref{PotentialYM} follows. 
Contracting then \eqref{YMFE1} with $\iota^{B}$, one gets
\begin{equation}
 0 = \Cym^{A'}\tilde\Cym_{A'}\tilde\Cym_{B'}\Phi_{i}{}^{j}
 = \tilde\Cym_{B'}\Cym^{A'}\tilde\Cym_{A'}\Phi_{i}{}^{j} 
 + [\tilde\Cym_{B'},\Cym_{A'}]\tilde\Cym^{A'}\Phi_{i}{}^{j},
 \label{EqWithCommutator2}
\end{equation}
so we need to compute the commutator on the right hand side. 
This is analogous to the calculation we needed in eq. \eqref{EqWithCommutator1},
but here it is a little bit longer.
In this case we need the identity \eqref{IdentityCommutatorYM}, and we put
$\mu^{A'}{}_{i}{}^{j}=\tilde\Cym^{A'}\Phi_{i}{}^{j}$, 
so $p(\mu^{A'}{}_{i}{}^{j})=-1$. This gives immediately:
\begin{align*}
 [\tilde\Cym_{B'},\Cym_{A'}]\tilde\Cym^{A'}\Phi_{i}{}^{j} ={}& 
 -\chi_{B'A'i}{}^{k}\tilde\Cym^{A'}\Phi_{k}{}^{j} +\chi_{B'A'k}{}^{j}\tilde\Cym^{A'}\Phi_{i}{}^{k} 
  + \varphi_{01i}{}^{k}\tilde\Cym_{B'}\Phi_{k}{}^{j} - \varphi_{01k}{}^{j}\tilde\Cym_{B'}\Phi_{i}{}^{k}.
\end{align*}
Now, we already established that \eqref{PotentialYM} is true. 
Replacing this in the terms involving $\chi_{A'B'i}{}^{j}$ in the equation above, it is easy to see that
\begin{equation*}
 -\chi_{A'B'i}{}^{k}\tilde\Cym^{A'}\Phi_{k}{}^{j} + \chi_{A'B'k}{}^{j}\tilde\Cym^{A'}\Phi_{i}{}^{k}
 = -\tilde\Cym_{B'}\left[ (\tilde\Cym_{A'}\Phi_{i}{}^{k})(\tilde\Cym^{A'}\Phi_{k}{}^{j}) \right].
\end{equation*}
On the other hand, contracting \eqref{YMFE2} with $o^{B}$ and using \eqref{conditionASDYM}, it follows 
that $\tilde\Cym_{B'}\varphi_{01i}{}^{j}=0$, so the terms involving $\varphi_{01i}{}^{j}$ are 
\begin{equation*}
 \varphi_{01i}{}^{k}\tilde\Cym_{B'}\Phi_{k}{}^{j} - \varphi_{01k}{}^{j}\tilde\Cym_{B'}\Phi_{i}{}^{k}
 = \tilde\Cym_{B'}\left( \varphi_{01i}{}^{k}\Phi_{k}{}^{j} - \varphi_{01k}{}^{j}\Phi_{i}{}^{k} \right).
\end{equation*}
Putting all together, we get:
\begin{equation*}
 \tilde\Cym_{B'} \left[ \Cym^{A'}\tilde\Cym_{A'}\Phi_{i}{}^{j} - (\tilde\Cym_{A'}\Phi_{i}{}^{k})(\tilde\Cym^{A'}\Phi_{k}{}^{j}) 
 +\varphi_{01i}{}^{k}\Phi_{k}{}^{j} -\varphi_{01k}{}^{j}\Phi_{i}{}^{k} \right] = 0,
\end{equation*}
which gives
\begin{equation}
\Cym^{A'}\tilde\Cym_{A'}\Phi_{i}{}^{j} -(\tilde\Cym_{A'}\Phi_{i}{}^{k})(\tilde\Cym^{A'}\Phi_{k}{}^{j}) 
 +\varphi_{01i}{}^{k}\Phi_{k}{}^{j} - \varphi_{01k}{}^{j}\Phi_{i}{}^{k} = K_{i}{}^{j} \label{HHYM2}
\end{equation}
where $K_{i}{}^{j} $ is a constant of integration, $\tilde\Cym_{A'}K_{i}{}^{j} =0$. 
(Compare to Jeffryes' result \cite[Eq. (4.19)]{Jeffryes}.)
Using now identities \eqref{Cym-tCym} and \eqref{ASDBoxYM} for $\Cym^{A'}\tilde\Cym_{A'}$, 
equation \eqref{HHYM} follows.

\subsection{The anti-self dual part of the curvature}\label{Sec:YMASD}

While the analysis of the SD Yang-Mills curvature is similar to what we did in section \ref{Sec:SDMaxwell}
for Maxwell, for the ASD curvature the method is analogous to the discussion of section \ref{Sec:ParallelFramesMaxwell}.
Recall that in that case, the point was to use a ``frame'' $\alpha$ to express the curvatures 
$\Box^{\Cym}_{A'B'}$ and $\Box^{\Cym}_{AB}$ (equations \eqref{SDBoxCMaxwell}-\eqref{ASDBoxCMaxwell}), 
and then to take advantage of the fact that, due to the flatness of $\tilde\Cym_{A'}$, 
we can use a frame that is parallel under $\tilde\Cym_{A'}$.
The logic now is the same.

\smallskip
We denote concrete numerical Yang-Mills indices by ${\bf i}, {\bf j}, {\bf k},... =1,...,n$.
Let $\alpha^{i}_{\bf j}$ be a set of linearly independent sections of $E$ over some neighbourhood $U\subset M$. 
These fields form a basis for the fiber $E_x$ for all $x\in U$. We denote the dual basis by $\alpha^{\bf j}_{i}$, 
so that it holds $\alpha^{\bf i}_{i}\alpha^{i}_{\bf j}=\delta^{\bf i}_{\bf j}$ and $\alpha^{\bf j}_{i}\alpha^{j}_{\bf j}=\delta^{j}_{i}$.
The action of the Yang-Mills spinor curvature operators $\Box^{\Cym}_{A'B'}$ and $\Box^{\Cym}_{AB}$
on $\alpha^{i}_{\bf k}$ is given by formulas \eqref{BoxYMSD1}-\eqref{BoxYMASD1}. 
Choosing $p(\alpha^{i}_{\bf k})=0$ and multiplying by $\alpha^{\bf k}_{i}$, we get:
\begin{subequations}
\begin{align}
 \chi_{A'B'i}{}^{j} ={}& \alpha^{\bf k}_{i}\Box^{\Cym}_{A'B'}\alpha^{j}_{\bf k}, \label{SDYM1} \\
 \varphi_{ABi}{}^{j} ={}& \alpha^{\bf k}_{i}\Box^{\Cym}_{AB}\alpha^{j}_{\bf k}. \label{ASDYM1}
\end{align}
\end{subequations}
These are the analogues of eqs. \eqref{SDBoxCMaxwell}-\eqref{ASDBoxCMaxwell} for Maxwell.
Now, since the connection $\tilde\Cym_{A'}$ is flat, we can choose $\alpha^{i}_{\bf k}$ to satisfy
\begin{equation}
 \tilde\Cym_{A'}\alpha^{i}_{\bf k} = 0.
\end{equation}
Define now the generalization of \eqref{gammaMaxwell0} as
\begin{equation}
 \gamma_{A'i}{}^{j} := \alpha^{\bf k}_{i}\Cym_{A'}\alpha^{j}_{\bf k}. \label{gammaYM0}
\end{equation}
Using identities \eqref{BoxYMSD0}-\eqref{BoxYMASD0}, a short calculation shows that 
\begin{subequations}
\begin{align}
 \chi_{A'B'i}{}^{j} ={}& \tilde\Cym_{(A'}\gamma_{B')i}{}^{j}, \label{SDYM} \\
 \varphi_{ABi}{}^{j} ={}& 
 o_{A}o_{B}(\Cym_{A'}\gamma^{A'}{}_{i}{}^{j}+\gamma_{A'i}{}^{k}\gamma^{A'}{}_{k}{}^{j})
 -o_{(A}\iota_{B)}\tilde\Cym_{A'}\gamma^{A'}{}_{i}{}^{j}. \label{ASDYM}
\end{align}
\end{subequations}
These are the non-abelian generalization of \eqref{CSDMaxwell1}-\eqref{CASDMaxwell1}.

\smallskip
Notice that, in deducing formulas \eqref{SDYM}-\eqref{ASDYM}, we did not use the field equations; 
rather, \eqref{SDYM}-\eqref{ASDYM} are a consequence of the geometric structure. 
Imposing now the field equations, we can compare \eqref{SDYM} and \eqref{PotentialYM}, so we deduce that 
\begin{equation}
 \gamma_{A'i}{}^{j} = \tilde\Cym_{A'}\Phi_{i}{}^{j} + \tau_{A'i}{}^{j}, \label{gammaYM1}
\end{equation}
where $\tau_{A'i}{}^{j}$ is such that 
\begin{equation}
 \tilde\Cym_{(A'}\tau_{B')i}{}^{j} = 0. \label{TwistorLikeEqYM}
\end{equation}
The solution to \eqref{TwistorLikeEqYM} is given by the obvious generalization of proposition \ref{Prop:TwistorLikeEq}
to the Yang-Mills case.
Replacing \eqref{gammaYM1} into \eqref{ASDYM}, we get $\varphi_{ABi}{}^{j}$ in terms of 
the potential $\Phi_{i}{}^{j}$ and the constants of integration involved in $\tau_{A'i}{}^{j}$.

\begin{remark}
After the integration procedure, the ASD Yang-Mills curvature
depends explicitly on the solution $\Phi_{i}{}^{j}$ to \eqref{HHYM}, 
as can be seen from replacing \eqref{gammaYM1} into \eqref{ASDYM}. 
This is a consequence of the non-abelian character of the theory, 
and is in contrast to the Maxwell case, eq. \eqref{ASDMaxwell2}, where the ASD Maxwell field 
was given by constants of integration.
If we set $\tau_{A'i}{}^{j}\equiv 0$ in \eqref{gammaYM1} and $K_{i}{}^{j}\equiv 0$
in \eqref{HHYM}, the resulting Yang-Mills field is self-dual, $\varphi_{ABi}{}^{j}=0$.
\end{remark}

\section{The Einstein-Weyl equations}\label{Sec:Einstein}

Recall that the Einstein-Weyl equations for a Weyl connection $\nabla^{\rm w}_{a}$ are \eqref{EWeq}. 
In this section we will study these equations for the case in which the Weyl 1-form is closed: we have, locally
\begin{equation}
 {\rm w}_{a}=\partial_{a}\log u \equiv \partial_{a}\log\mr\Omega^{-1}. \label{DefOmega0}
\end{equation}
The introduction of the field $\mr\Omega\equiv u^{-1}$ is simply a matter of convenience: 
we can interpret $\mr\Omega$ as a conformal factor (its conformal weight is $w(\mr\Omega)=1$).
The Einstein-Weyl equations can be expressed in spinor terms as
\begin{equation}
 \Phi^{\rm w}_{ABC'D'} = 0, \label{EinsteinWeylEquations0}
\end{equation}
where $\Phi^{\rm w}_{ABC'D'}$ is defined by \eqref{SpinorDecCurvatureWeylC}, and given 
in terms of a Levi-Civita connection by \eqref{Weyl-LC2}.

\medskip
We assume that the spinor field $o^{A}$ satisfies \eqref{SFR} and \eqref{rPND}.

\begin{remark}\label{Remark:LeeForm2}
In this section we assume that the Lee form $f_{a}$ has the expression \eqref{fgradient}.
For Lorentz signature and complex manifolds, this can always be achieved,
as long as \eqref{SFR} and \eqref{rPND} hold, see proposition \ref{Prop:LeeForm}.
For Riemann and split signature, on the other hand, conditions \eqref{SFR} and \eqref{rPND} {\rm do not} imply 
that \eqref{fgradient} holds, see the comment below remark \ref{Remark:LeeForm}.
However, in these signatures, the condition \eqref{rPND} implies that the Weyl spinor $\Psi_{ABCD}$ is actually type D.
From the Bianchi identities expressed in the form of eq. \eqref{BianchiAuxWeylPsi} below, 
together with the Einstein-Weyl equations \eqref{EinsteinWeylEquations0}, it then follows that 
\begin{equation}
 \C_{a}(\mr\Omega^{-1}\Psi_{2}) = 0, \label{EqForPsi2}
\end{equation}
where $\Psi_{2}=\Psi_{ABCD}o^{A}o^{B}\iota^{C}\iota^{D}$. 
Since $p(\mr\Omega^{-1}\Psi_{2})=0$, $w(\mr\Omega^{-1}\Psi_{2})=-3$,
equation \eqref{EqForPsi2} implies that 
\begin{equation}
 f_{a}=\nabla_{a}\log(\mr\Omega^{-1}\Psi_{2})^{1/3}, \label{LeeForm4}
\end{equation}
so \eqref{fgradient} holds for these cases as well.
\end{remark}

\subsection{Relations between the two Weyl connections}\label{Sec:AdditionalWeylConnection}

For the study of the Einstein-Weyl system, there are two Weyl connections involved. 
One of them, $\nabla^{\rm w}_{a}$, is the one that defines the field equations \eqref{EinsteinWeylEquations0} 
of the Einstein-Weyl system we want to study, 
and the other one, $\nabla^{\rm f}_{a}$, is associated to the almost-complex structure (section \ref{Sec:ConnectionC})
and is not related to field equations. For both of them, the Weyl 1-forms are closed. 

\smallskip
Some general identities for Weyl connections, with closed Weyl 1-forms, are given in Appendix \ref{App:WeylConnections}. 
We denote the curvature spinors associated to $\nabla^{\rm w}_{a}$ by 
$\tilde{X}^{\rm w}_{A'B'C'D'}$, $X^{\rm w}_{ABCD}$, $\Phi^{\rm w}_{ABC'D'}$ and $\Lambda^{\rm w}$. 
The definitions of these objects are given by \eqref{SpinorDecCurvatureWeylC}--\eqref{CrucialRelationsCurvatureWeylC},
and the relations between them and the curvature spinors of the Levi-Civita connection of an arbitrary metric in the conformal 
class are given by formulas \eqref{Weyl-LC1}--\eqref{Weyl-LC4}. 
The corresponding objects for $\nabla^{\rm f}_{a}$ will be denoted by 
$\tilde{X}^{\rm f}_{A'B'C'D'}$, $X^{\rm f}_{ABCD}$, $\Phi^{\rm f}_{ABC'D'}$ and $\Lambda^{\rm f}$. 

\smallskip
We recall that the totally symmetric parts of $\tilde{X}^{\rm w}_{A'B'C'D'}$ and $\tilde{X}^{\rm f}_{A'B'C'D'}$ 
coincide with the ordinary SD Weyl spinor $\tilde{\Psi}_{A'B'C'D'}$, 
and similarly $X^{\rm w}_{(ABCD)}=X^{\rm f}_{(ABCD)}=\Psi_{ABCD}$.

\smallskip
We will need the relations between the curvature spinors of $\nabla^{\rm w}_{a}$ and those of $\nabla^{\rm f}_{a}$,
as well as the Bianchi identities expressed in terms of $\C_{a}$.
All these relations can be obtained by using formulas \eqref{Weyl-LC1}--\eqref{Weyl-LC4}
and \eqref{BianchiWeyl1}--\eqref{ContractedBianchiWeyl} in appendix \ref{App:WeylConnections}.
After some straightforward calculations, we get:

\begin{proposition}
Assume that \eqref{SFR}, \eqref{rPND}, \eqref{fgradient} hold, and let $\mr\Omega$ be given by \eqref{DefOmega0}.
\begin{enumerate}
\item We have the identities
\begin{subequations}
\begin{align}
 \Lambda^{\rm w} = {}&  \Lambda^{\rm f} + \tfrac{1}{4} \mr\Omega\Box^{\mathcal{C}}\mr\Omega^{-1}, \\
 \Phi^{\rm w}_{ABC'D'} ={}& \Phi^{\rm f}_{ABC'D'} + \mr{\Omega}^{-1}\mathcal{C}_{A(C'}\mathcal{C}_{D')B}\mr{\Omega}.
\end{align}
\end{subequations}
In particular:
\begin{subequations}
\begin{align}
 o^{A}o^{B}\Phi^{\rm w}_{ABC'D'} ={}& o^{A}o^{B}\Phi^{\rm f}_{ABC'D'} + \mr\Omega^{-1}\tC_{C'}\tC_{D'}\mr\Omega, \label{ooPhi0} \\
 o^{A}\iota^{B}\Phi^{\rm w}_{ABC'D'} ={}& o^{A}\iota^{B}\Phi^{\rm f}_{ABC'D'} + \mr\Omega^{-1}\tC_{(C'}\C_{D')}\mr\Omega, \label{oiPhi0} \\
 \iota^{A}\iota^{B}\Phi^{\rm w}_{ABC'D'} ={}& \iota^{A}\iota^{B}\Phi^{\rm f}_{ABC'D'} 
 + \mr\Omega^{-1}(\C_{(C'}\C_{D')}-\sigma^{1}_{(C'}\tC_{D')})\mr\Omega, \label{iiPhi0}
\end{align}
\end{subequations}
where $\sigma^{1}_{A'}$ is defined in \eqref{Ci}.
\item The Bianchi identities for the curvature spinors of $\nabla^{\rm w}_{a}$ can be expressed as:
\begin{subequations}
\begin{align}
 & \mr\Omega \; \C_{B}{}^{A'}[\mr\Omega^{-1}\tilde{\Psi}_{A'B'C'D'}]
 = \mr\Omega^{-1}\C_{(B'}{}^{A}[\mr\Omega\Phi^{\rm w}_{C'D')AB}], \label{BianchiAuxWeyl2}  \\
 & \mr\Omega \; \C_{B'}{}^{A}[\mr\Omega^{-1}\Psi_{ABCD}]
 = \mr\Omega^{-1}\C_{(B}{}^{A'}[\mr\Omega\Phi^{\rm w}_{CD)A'B'}],  \label{BianchiAuxWeylPsi} \\ 
 & \mr\Omega^{2}\C^{AC'}(\mr\Omega^{-2}\Phi^{\rm w}_{ABC'D'}) 
 + 3\mr\Omega^{-2} \; \C_{BD'}(\mr\Omega^{2}\Lambda^{\rm w}) = 0.
\end{align}
\end{subequations}
\end{enumerate}
\end{proposition}

\subsection{The conformally invariant HH equation}\label{Sec:CHHequation}

\begin{theorem}\label{Thm:ConformalHHequation}
Assume that the Einstein-Weyl equations \eqref{EinsteinWeylEquations0} 
are satisfied. Then there exists, locally, a scalar field $\Phi$ with weights $p(\Phi)=-4$ and $w(\Phi)=-1$
such that the SD Weyl curvature spinor of the conformal structure is 
\begin{equation}
  \tilde{\Psi}_{A'B'C'D'} = \mr{\Omega}\tC_{A'}\tC_{B'}\tC_{C'}\tC_{D'}\Phi, \label{PotentializationWeyl}
\end{equation}
and such that $\Phi$ satisfies the conformally invariant equation
\begin{equation}
 (\Box^{\C} + 3\mathcal{R}/2)\Phi +\mr\Omega(\tC_{A'}\tC_{B'}\Phi)(\tC^{A'}\tC^{B'}\Phi)
 -4(\tC^{A'}\mr\Omega)(\tC^{B'}\Phi)(\tC_{A'}\tC_{B'}\Phi) = 2K, \label{conformalHHeq}
\end{equation}
where $\Box^{\C}=g^{ab}\C_{a}\C_{b}$, $\mathcal{R}$ is the `scalar curvature' of $\C_{a}$ 
\eqref{ScalarCurvatureC}-\eqref{ScalarCurvatureCSpecial}, and $K$ is an arbitrary scalar field defined 
only by the condition $\tC_{A'}\tC_{B'}\tC_{C'}K=0$.
We call \eqref{conformalHHeq} the ``conformally invariant hyper-heavenly (HH) equation''.
\end{theorem}

\begin{remark}
The arbitrary function $K$ in \eqref{conformalHHeq} can be expressed in terms of $\pi^{A'}$ (defined by \eqref{TwistorLikeEq})
and three constants of integration as in equation \eqref{SolSpecialEqs1} for $m=3$.
\end{remark}

\begin{remark}
The potentialization \eqref{PotentializationWeyl} of the SD Weyl spinor, together with equation \eqref{conformalHHeq}, 
are the generalization to the Einstein-Weyl system of the corresponding results for Dirac, Maxwell, and Yang-Mills fields.
In principle, we could also think of \eqref{PotentializationWeyl}-\eqref{conformalHHeq} as a `particular case' of a 
Yang-Mills system, by making the identifications $\chi_{A'B'i}{}^{j} \equiv \mr\Omega^{-1}\tilde\Psi_{A'B'C'}{}^{D'}$ 
and $\Phi_{i}{}^{j} \equiv \tC_{C'}\tC^{D'}\Phi$.
More formally, this would look like a Yang-Mills theory for the gauge group $G={\rm SL}(2,\mathbb{C})$.
We notice, however, that in our treatment it is the field $\tilde\Psi_{A'B'C'}{}^{D'}$, not $\mr\Omega^{-1}\tilde\Psi_{A'B'C'}{}^{D'}$,
which arises as the curvature of a connection.
\end{remark}

The proof of theorem \ref{Thm:ConformalHHequation} is very similar to the proof of the corresponding
Maxwell and Yang-Mills results, seen in sections \ref{Sec:SDMaxwell} and \ref{Sec:YMSD} respectively;
the only difference is that the computations are harder.
We leave most of the details to appendix \ref{App:HH}, and here we simply give the main points for this proof.

First of all, the Einstein-Weyl equations \eqref{EinsteinWeylEquations0} together with
the Bianchi identities expressed in the form \eqref{BianchiAuxWeyl2} imply that 
\begin{equation}
 \C_{B}{}^{A'}(\mr{\Omega}^{-1}\tilde{\Psi}_{A'B'C'D'})=0. \label{BianchiOnShell}
\end{equation}
Contracting with $o^{B}$, this is $\tC^{A'}(\mr{\Omega}^{-1}\tilde{\Psi}_{A'B'C'D'})=0$, 
therefore, there exists (locally) a scalar field $\Phi$, with weights $p(\Phi)=-4$ and $w(\Phi)=-1$, such that 
\eqref{PotentializationWeyl} is true.
Contracting \eqref{BianchiOnShell} with $\iota^{B}$, we get $\C^{A'}(\mr\Omega^{-1}\tilde{\Psi}_{A'B'C'D'})=0$. 
Replacing \eqref{PotentializationWeyl}:
\begin{equation}
 0=\C^{A'}\tC_{A'}\tC_{B'}\tC_{C'}\tC_{D'}\Phi 
 = \tC_{B'}\tC_{C'}\tC_{D'}\C_{A'}\tC^{A'}\Phi +[\C_{A'}, \tC_{B'}\tC_{C'}\tC_{D'}]\tC^{A'}\Phi. \label{EqWithCommutator3}
\end{equation}
We see that, as in the Maxwell and Yang-Mills cases (equations \eqref{EqWithCommutator1} and \eqref{EqWithCommutator2} 
respectively), the derivation of equation \eqref{conformalHHeq} boils down to 
the computation of the commutator in the right hand side of \eqref{EqWithCommutator3}. 
The Maxwell case was almost trivial, the Yang-Mills case required a little bit more work, 
and now the calculation is in turn a little harder, but, using the identities given in appendix \ref{App:SFR}, 
the required computation is tedious but straightforward. We leave the details to appendix \ref{App:HH}.

An important point that is needed in the intermediate computations is the following:
\begin{proposition}
Suppose that the Einstein-Weyl equations \eqref{EinsteinWeylEquations0} are satisfied. 
Then the scalar field $\mr\Omega$ satisfies 
\begin{equation}
 \tC_{A'}\tC_{B'}\mr{\Omega}=0. \label{CtCtOmega}
\end{equation}
\end{proposition}

\begin{proof}
This follows from \eqref{ooPhi0} and \eqref{ooPhiw}, 
after imposing the field equations \eqref{EinsteinWeylEquations0}.
\end{proof}

From the discussion in section \ref{Sec:ConstantsOfIntegration}, we can express the solution to 
\eqref{CtCtOmega} in terms of $\pi^{A'}$ and constants of integration:
\begin{equation}
 \mr\Omega = K^{1}_{A'}\pi^{A'} + K^{0},
\end{equation}
where $\tC_{A'}K^{1}_{B'}=0=\tC_{A'}K^{0}$.

\subsection{The ordinary Einstein equations}\label{Sec:OrdinaryEE}

The Einstein equations $R_{ab}=\lambda g_{ab}$ are not conformally invariant. 
In particular, $R_{ab}$ is not a conformal density;
so the conformal machinery associated to $\C_{a}$ seems, in principle, inapplicable to this case.
However, there is a simple trick we can use to study this system as well.
Recalling expression \eqref{DefOmega0} for the Weyl 1-form, we see that if we set $\mr\Omega \equiv 1$, then
\begin{equation}
 \nabla^{\rm w}_{a}|_{\mr\Omega=1} = \nabla_{a}, \label{BreakingCI}
\end{equation}
that is, the Weyl connection $\nabla^{\rm w}_{a}$ reduces to a Levi-Civita connection $\nabla_{a}$. 
All the formulas we obtained above for the Einstein-Weyl system remain valid for the particular value $\mr\Omega=1$, 
the only thing that happens is that we break conformal invariance.
Correspondingly, the Einstein-Weyl equations reduce to the ordinary Einstein equations, 
and the potentialization result applies to them.
In particular, putting $\mr\Omega \equiv 1$, the SD Weyl spinor is
\begin{equation}
 \tilde{\Psi}_{A'B'C'D'} = \tC_{A'}\tC_{B'}\tC_{C'}\tC_{D'}\Phi, \label{NonConfSDWeylSpinor}
\end{equation}
and the conformal HH equation \eqref{conformalHHeq} becomes
\begin{equation}
 (\Box^{\C} + 3\mathcal{R}/2)\Phi +(\tC_{A'}\tC_{B'}\Phi)(\tC^{A'}\tC^{B'}\Phi) -4\tilde{f}^{A'}(\tC^{B'}\Phi)(\tC_{A'}\tC_{B'}\Phi) = 2K,
\label{NonConfHHeq}
\end{equation}
where $\tilde{f}_{A'}=o^{A}f_{AA'}$. Conformal invariance is of course broken in 
\eqref{NonConfSDWeylSpinor} and \eqref{NonConfHHeq}. 
Choosing a special set of coordinates (see section \ref{Sec:Coordinates} below), 
\eqref{NonConfHHeq} should reduce to the hyper-heavenly equation of 
Pleba\'nski and Robinson \cite{PlebanskiRobinson} and Pleba\'nski and Finley \cite[Eq. (3.14)]{PlebanskiFinley} 
(but we have not attempted to show this).

\section{Reconstruction of the conformal structure}\label{Sec:Reconstruction}

After deducing the conformal HH equation \eqref{conformalHHeq} for the 
Einstein-Weyl system \eqref{EinsteinWeylEquations0},
we will now see how to reconstruct the full conformal curvature and metric from a solution to \eqref{conformalHHeq}.
The procedure for the reconstruction of the curvature is analogous to 
the one used for the Maxwell and Yang-Mills cases, seen in sections \ref{Sec:ParallelFramesMaxwell} 
and \ref{Sec:YMASD} respectively.
That is: our method is based on the consideration of suitable parallel frames.
An important difference is that we now have {\em two} types of parallel frames, namely primed and unprimed spin frames.

On the other hand, the reconstruction of the conformal metric requires a different technique, 
which uses as an auxiliary intermediate step the special coordinates associated to twistor surfaces. 
The final expression we get is, however, in terms of abstract spinor fields.

\medskip
We assume that the conditions \eqref{SFR}, \eqref{rPND} and \eqref{fgradient} are satisfied.

\subsection{Curvature}\label{Sec:RecCurvature}

The curvature of $\mathcal{C}_{a}$ is described in Appendix \ref{App:curvatureC}. 
From formulas \eqref{BoxCSpecial1}--\eqref{BoxCSpecial4} we see that
the action of $\Box^{\mathcal{C}}_{A'B'}$ and $\Box^{\mathcal{C}}_{AB}$ on primed and unprimed spinor fields
with $p=0$ involves only one of the curvature spinors $\tilde{X}^{\rm f}_{A'B'C'D'}$, $\Phi^{\rm f}_{ABC'D'}$, $X^{\rm f}_{ABCD}$.
Considering then a primed spin frame $\tilde{\alpha}^{A'}_{\bf A'}=(\tilde{\alpha}^{A'}_{\bf 0'}, \tilde{\alpha}^{A'}_{\bf 1'})$
with  $p(\tilde{\alpha}^{A'}_{\bf A'})=0$, 
and an unprimed spin frame $\alpha^{A}_{\bf A}=(\alpha^{A}_{\bf 0}, \alpha^{A}_{\bf 1})$
with $p(\alpha^{A}_{\bf A})=0$, we get the following expressions for the curvature spinors:
\begin{subequations}
\begin{align}
 \tilde{X}^{\rm f}_{A'B'C'D'} ={}& \tilde\alpha^{\bf E'}_{C'} \Box^{\C}_{A'B'} \tilde\alpha_{{\bf E'} D'}
  \label{XCtilde} \\
  \Phi^{\rm f}_{ABC'D'} ={}& \tilde\alpha^{\bf E'}_{C'} \Box^{\C}_{AB} \tilde\alpha_{{\bf E'}D'}
  \label{PhiCtilde} \\
 X^{\rm f}_{ABCD} ={}& \alpha^{\bf E}_{C} \Box^{\C}_{AB} \alpha_{{\bf E}D}, 
  \label{XC}\\
 \Phi^{\rm f}_{A'B'CD} ={}& \alpha^{\bf E}_{C} \Box^{\C}_{A'B'} \alpha_{{\bf E}D}.
 \label{PhiC} 
\end{align}
\end{subequations}
These expressions are analogous to the Maxwell case, eqs. \eqref{SDBoxCMaxwell}-\eqref{ASDBoxCMaxwell}, 
and to the Yang-Mills case, eqs. \eqref{SDYM1}-\eqref{ASDYM1}.

\begin{remark}
A crucial point for the analysis that follows is the equality between equations \eqref{PhiCtilde} and \eqref{PhiC}, 
which follows from assuming that the Lee form $f_{a}$ is a gradient \eqref{fgradient}. 
If only the weaker condition \eqref{LeeForm2} is used, it is still possible to get formulas like \eqref{XCtilde}--\eqref{PhiC}
provided that one chooses specific conformal weights for the spin frames. The analysis is however quite more complicated.
\end{remark}

We now choose the frames $\tilde\alpha^{A'}_{\bf A'}$ and $\alpha^{A}_{\bf A}$ to be parallel under $\tC_{A'}$:
\begin{equation}
 \tC_{A'}\tilde{\alpha}^{B'}_{\bf B'}=0,   \qquad    \tC_{A'}\alpha^{B}_{\bf B}=0. \label{parallelframes}
\end{equation}
Define
\begin{equation}\label{gammaGravity0}
 \tilde{\gamma}_{A'B'}{}^{C'} := \tilde{\alpha}^{\bf D'}_{B'}\C_{A'}\tilde{\alpha}^{C'}_{\bf D'}, 
 \qquad 
 \gamma_{A'B}{}^{C} := \alpha^{\bf D}_{B}\C_{A'}\alpha^{C}_{\bf D}.
\end{equation}
These objects are the generalization of \eqref{gammaMaxwell0} and \eqref{gammaYM0}.
By using identities \eqref{IdentitySDCC}-\eqref{IdentityASDCC} together with \eqref{parallelframes}, 
a short calculation shows that
\begin{subequations}
\begin{align}
 \tilde{X}^{\rm f}_{A'B'C'D'} ={}& \tC_{(A'}\tilde{\gamma}_{B')C'D'} \label{curvatureC1} \\
 \Phi^{\rm f}_{ABC'D'} ={}& o_{A}o_{B} (\C_{A'}\tilde{\gamma}^{A'}{}_{C'D'} 
  - \tilde{\gamma}_{A'B'C'}\tilde{\gamma}^{A'B'}{}_{D'}) - o_{(A}\iota_{B)}\tC_{A'}\tilde{\gamma}^{A'}{}_{C'D'}, 
  \label{PhiCtilde1} \\
 \Lambda^{\rm f} ={}& \tfrac{1}{6}\tC_{A'}\tilde\gamma_{B'}{}^{A'B'} \label{LambdaCtilde1} \\
 X^{\rm f}_{ABCD} ={}& o_{A}o_{B}(\C_{A'}\gamma^{A'}{}_{CD}-\gamma_{A'BC}\gamma^{A'B}{}_{D})
  - o_{(A}\iota_{B)}\tC_{A'}\gamma^{A'}{}_{CD}, \label{XC1} \\
 \Phi^{\rm f}_{A'B'CD} ={}& \tC_{(A'}\gamma_{B')CD}, \label{PhiC1} \\
 \Lambda^{\rm f} ={}& -\tfrac{1}{6}\tC_{A'}(\gamma^{A'}{}_{CD}o^{C}\iota^{D}). \label{LambdaC1}
\end{align}
\end{subequations}
These formulas, in turn, are the generalization of \eqref{CSDMaxwell1}-\eqref{CASDMaxwell1} for Maxwell
and \eqref{SDYM}-\eqref{ASDYM} for Yang-Mills.

\subsubsection{The self-dual part}

From identities \eqref{curvatureC1}--\eqref{LambdaC1} we deduce, after straightforward calculations, 
some further formulas that we will need below:
\begin{subequations}
\begin{align}
 \tilde{X}^{\rm f}_{A'B'C'D'} +\epsilon_{A'B'}o^{A}\iota^{B}\Phi^{\rm f}_{ABC'D'} 
 ={}& \tC_{A'}\tilde\gamma_{B'C'D'}, \label{IdCurvC1} \\
 \Phi^{\rm f}_{A'B'CD} + \epsilon_{A'B'}o^{A}\iota^{B}X^{\rm f}_{ABCD} ={}& \tC_{A'}\gamma_{B'CD}, 
 \label{IdCurvC2} \\
 o^{C}\iota^{D}\Phi^{\rm f}_{A'B'CD}+3\epsilon_{A'B'}\Lambda^{\rm f} ={}& \tC_{B'}\gamma_{A'01} 
 \label{IdCurvC3}
\end{align}
\end{subequations}
where $\gamma_{A'01}=\gamma_{A'BC}o^{B}\iota^{C}$.

The basic variable that allows us to reconstruct the conformal structure (by which we mean the 
full curvature and the conformal metric) is a spinor field $\theta_{A'B'}$ defined by the following:

\begin{proposition}\label{Prop:Existencetheta}
There exists a symmetric spinor field $\theta_{A'B'}=\theta_{(A'B')}$, 
with weights $p(\theta_{A'B'})=-2$ and $w(\theta_{A'B'})=0$, such that 
\begin{equation}
 \tilde{\gamma}_{A'B'C'}+\epsilon_{B'C'}\gamma_{A'01} = 2\tilde{\mathcal{C}}_{B'}\theta_{C'A'}.  \label{theta0}
\end{equation}
\end{proposition}

\begin{remark}\label{Remark:Freedomtheta}
The spinor field $\theta_{A'B'}$ is defined only up to
\begin{equation}
 \theta_{A'B'} \to \theta_{A'B'} + k_{A'B'}, \label{Freedomtheta}
\end{equation}
where $\tC_{A'}k_{B'C'}=0$. We will use this freedom later (see the proof of proposition \ref{Prop:RelationwP}).
\end{remark}

\begin{proof}[Proof of Proposition \ref{Prop:Existencetheta}]
Contracting \eqref{IdCurvC1} with $\epsilon^{A'C'}$ and renaming indices,
we get $\tC^{B'}\tilde{\gamma}_{A'B'C'} = -\tC_{C'}\gamma_{A'01}$, or equivalently:
\begin{equation*}
 \tC^{B'}(\tilde{\gamma}_{A'B'C'}+\epsilon_{B'C'}\gamma_{A'01})=0.
\end{equation*}
Therefore, from lemma \ref{Lemma:Potentials}, there exists (locally) a spinor field $\theta_{C'A'}$ such that \eqref{theta0} holds.
It remains to show that $\theta_{A'B'}$ can be chosen to be symmetric.
Contraction of \eqref{theta0} with $\epsilon^{A'C'}$ shows that 
$\tilde{\gamma}_{A'B'}{}^{A'}+\gamma_{B'01}=2\tilde{\mathcal{C}}_{B'}\theta^{C'}{}_{C'}$, so taking an 
additional derivative and renaming indices: 
$\tC_{A'}\tilde{\gamma}_{B'C'}{}^{B'}+\tC_{A'}\gamma_{C'01} = 2 \tC_{A'} \tC_{C'}\theta^{B'}{}_{B'}$.
On the other hand, contracting \eqref{IdCurvC1} with $\epsilon^{B'D'}$, we have 
$\tC_{A'}\tilde{\gamma}_{B'C'}{}^{B'} = - \tC_{A'}\gamma_{C'01}$; thefefore $\tC_{A'} \tC_{C'}\theta^{B'}{}_{B'}=0$.
Now, when we replace $\tilde{\gamma}_{A'B'C'}$ in \eqref{curvatureC1}, \eqref{PhiCtilde1}, \eqref{LambdaCtilde1}, 
all formulas involve two $\tC_{A'}$-derivatives of $\theta_{A'B'}$. Decomposing $\theta_{A'B'}$ into its 
symmetric and antisymmetric parts, the fact that $\tC_{A'} \tC_{C'}\theta^{B'}{}_{B'}=0$ implies that the curvature depends only on 
the symmetric part of $\theta_{A'B'}$. Thus we can take $\theta^{B'}{}_{B'} \equiv 0$, i.e. $\theta_{A'B'}=\theta_{(A'B')}$.
\end{proof}

Taking the symmetric and antisymmetric parts in $B'C'$ in \eqref{theta0}, we get
\begin{align}
 \tilde{\gamma}_{A'B'C'} ={}& 2\tC_{(B'}\theta_{C')A'}, \label{gammatilde} \\
 \gamma_{A'01} ={}& \tC_{B'}\theta^{B'}{}_{A'}. \label{gamma01}
\end{align}
From formulas \eqref{curvatureC1}, \eqref{PhiCtilde1}, \eqref{LambdaCtilde1} we see that
the SD curvature $\tilde{X}^{\rm f}_{A'B'C'D'}$ together with $\Phi^{\rm f}_{ABC'D'}$, $\Lambda^{\rm f}$ 
depend only on $\tilde{\gamma}_{A'B'C'}$, therefore, \eqref{gammatilde} allows us to express 
this part of the curvature in terms of the spinor field $\theta_{A'B'}$.
On the other hand, from \eqref{XC1} we see that $X^{\rm f}_{ABCD}$ involves $\gamma_{A'BC}$, 
and \eqref{gamma01} is only the component $\gamma_{A'01}$.
Thus, in order to obtain the full curvature in terms of $\theta_{A'B'}$, we need to determine 
the whole $\gamma_{A'BC}$. We will do this in section \ref{Sec:RecASDcurvature} below.

\medskip
Imposing the field equations \eqref{EinsteinWeylEquations0}, we get the following:

\begin{proposition}\label{Prop:RelationSpinorAndScalarPotentials}
Suppose that the Einstein-Weyl equations \eqref{EinsteinWeylEquations0} are satisfied. 
Let $\Phi$ be the potential for the Weyl spinor, eq. \eqref{PotentializationWeyl}. Then the spinor field 
$\theta_{A'B'}$ defined by \eqref{theta0} is given by
\begin{equation}
 \theta_{A'B'} = \mr{\Omega}\tC_{A'}\tC_{B'}\Phi -2(\tC_{(A'}\mr{\Omega})(\tC_{B')}\Phi ) + \tau_{A'B'}, \label{Potentialtheta}
\end{equation}
where $\tau_{A'B'}$ is such that $\tC_{(A'}\tC_{B'}\tau_{C'D')}=0$.
\end{proposition}

\begin{remark}
As discussed in section \ref{Sec:ConstantsOfIntegration}, the spinor field $\tau_{A'B'}$ in \eqref{Potentialtheta}
can be expressed in terms of $\pi^{A'}$ (defined by \eqref{TwistorLikeEq}) and constants of integration, 
see \eqref{DoubleTwistorLikeEq}.
\end{remark}

\begin{proof}[Proof of proposition \ref{Prop:RelationSpinorAndScalarPotentials}]
Since the primed Weyl curvature spinor is $\tilde{\Psi}_{A'B'C'D'}=\tilde{X}^{\rm f}_{(A'B'C'D')}$,
using \eqref{curvatureC1} we get
\begin{equation}
 \tilde{\mathcal{C}}_{(A'}\tilde{\gamma}_{B'C'D')}=
  \mathring{\Omega}\tilde{\mathcal{C}}_{A'}\tilde{\mathcal{C}}_{B'}\tilde{\mathcal{C}}_{C'}\tilde{\mathcal{C}}_{D'}\Phi.
\end{equation}
Notice that this equation involves only the totally symmetric part $\tilde{\gamma}_{(A'B'C')}$.
Using equation \eqref{CtCtOmega}, it is straightforward to check that the solution is 
\begin{equation}
 \tilde{\gamma}_{(A'B'C')} = \mr{\Omega}\tC_{A'}\tC_{B'}\tC_{C'}\Phi - (\tC_{(A'}\mr{\Omega})(\tC_{B'}\tC_{C')}\Phi ) + \rho_{A'B'C'},
  \label{SolGammaTilde1}
\end{equation}
where $\rho_{A'B'C'}=\rho_{(A'B'C')}$ is such that $\tC_{(A'}\rho_{B'C'D')}=0$.
Recalling now formula \eqref{theta0} for $\tilde{\gamma}_{A'B'C'}$, we have
\begin{equation}
 \tC_{(A'}\theta_{B'C')} = \mr{\Omega}\tC_{A'}\tC_{B'}\tC_{C'}\Phi 
 - (\tC_{(A'}\mr{\Omega})(\tC_{B'}\tC_{C')}\Phi ) + \rho_{A'B'C'}.
\end{equation}
The solution to this equation can be easily found as:
\begin{equation}
 \theta_{A'B'} = \mathring{\Omega}\tilde{\mathcal{C}}_{A'}\tilde{\mathcal{C}}_{B'}\Phi 
 -2(\tilde{\mathcal{C}}_{(A'}\mathring{\Omega})(\tilde{\mathcal{C}}_{B')}\Phi ) + \tau_{A'B'},
\end{equation}
where $\tau_{A'B'}$ is such that $\tC_{(A'}\tau_{B'C')}=\rho_{A'B'C'}$. 
(The existence of such a spinor follows from the condition $\tC^{A'}\tC^{B'}\tC^{C'}\rho_{A'B'C'}=0$, which in turn 
follows from applying $\tC^{A'}\tC^{B'}\tC^{C'}$ to \eqref{SolGammaTilde1}.)
\end{proof}

\subsubsection{The anti-self-dual part}\label{Sec:RecASDcurvature}

We observed before that, in order to determine the ASD curvature $X^{\rm f}_{ABCD}$, 
we need to determine $\gamma_{A'BC}$, and \eqref{theta0} only gives us $\gamma_{A'01}$.
To determine the rest, we notice that we can choose the unprimed spin frame $\alpha^{A}_{\bf A}$ 
to be given by some rescaling of $o^{A},\iota^{A}$ so as to make $p(\alpha^{A}_{\bf A})=0$. 
That is, we set $\alpha^{A}_{0} \equiv \mr{\lambda} o^{A}$, $\alpha^{A}_{1} \equiv \mr{\lambda}^{-1}\iota^{A}$, 
where $\mr{\lambda}$ is a weighted scalar field such that 
\begin{equation}
 p(\mr{\lambda}) = -1, \qquad \tC_{A'}\mr{\lambda} = 0.  \label{lambdao}
\end{equation}
The role of this field $\mr\lambda$ at this point is only auxiliary; however, 
we will see in section \ref{Sec:Coordinates} below that there is a convenient specific choice of $\mr\lambda$, 
see around eq. \eqref{Choicelambdao}.
The resulting spin frame $\alpha^{A}_{\bf A}$ satisfies then $\tC_{A'}\alpha^{B}_{\bf B}=0$ and 
$p(\alpha^{A}_{\bf A})=0$, so all the above formulas for the curvature apply. 
Furthermore, we get:
\begin{equation}
 \gamma_{A'BC} = o_{B}o_{C}\sigma^{1}_{A'} - 2 o_{(B}\iota_{C)} \mr\lambda^{-1}\C_{A'}\mr{\lambda}, \label{gammaGravity1}
\end{equation}
where $\sigma^{1}_{A'}$ is defined in \eqref{Ci}, and it represents the `shear' of $\iota^{A}$.
We see that $\mr\lambda^{-1}\C_{A'}\mr{\lambda}=\gamma_{A'01}=\tC_{B'}\theta^{B'}{}_{A'}$, so $\mr\lambda$ 
is only auxiliary as mentioned. 
The part of $\gamma_{A'BC}$ that remains to be determined then is $\sigma^{1}_{A'}$.

\begin{proposition}
The relation between the spinor field $\sigma^{1}_{A'}$ in \eqref{gammaGravity1} (defined by \eqref{Ci})
and $\theta_{A'B'}$ is:
\begin{equation}
 \sigma^{1}_{A'} = 2(\C_{B'}\theta_{A'}{}^{B'}-\tilde\gamma_{A'B'C'}\theta^{B'C'}) + k_{A'}, \label{sigma1}
\end{equation}
where $k_{A'}$ satisfies $\tC_{A'}\tC_{B'}k_{C'}=0$.
\end{proposition}

\begin{proof}
The Bianchi identity \eqref{BianchiXwt1} can be written as 
$\C^{A'}(\tilde{X}^{\rm f}_{B'A'C'D'}+\epsilon_{B'A'}\Phi^{\rm f}_{C'D'01})=-\tC_{B'}\Phi^{\rm f}_{C'D'11}$.
Replacing \eqref{IdCurvC1}, we get:
\begin{align*}
 \C^{A'}\tC_{B'}\tilde\gamma_{A'C'D'} = -\tC_{B'}\Phi^{\rm f}_{C'D'11}.
\end{align*}
In addition, \eqref{ContractedBianchiw1} can be expressed as 
$\C^{C'}(\Phi^{\rm f}_{C'D'01}+3\Lambda^{\rm f}\epsilon_{C'D'})-\tC^{C'}\Phi^{\rm f}_{C'D'11}=0$,
which, when combined with \eqref{IdCurvC3}, gives
\begin{align*}
 \C^{A'}\tC_{B'}\gamma_{A'01} = \tC^{A'}\Phi^{\rm f}_{A'B'11}.
\end{align*}
Adding the two previous equations and using \eqref{theta0} and \eqref{Phiwii}, we get
\begin{align*}
 2\C^{A'}\tC_{B'}\tC_{C'}\theta_{D'A'} = -\tC_{B'}\tC_{C'}\sigma^{1}_{D'}.
\end{align*}
We then see that, in the left hand side of this equation, we have to commute the operators $\C_{A'}$ and $\tC_{A'}$.
This can be done with the help of the general identities that we give for this commutator in appendix \ref{App:SFR}.
For the sake of brevity, we defer this calculation to appendix \ref{App:Sigma1}, where we prove that
\begin{equation}
 \C_{A'}\tC_{B'}\tC_{C'}\theta_{D'}{}^{A'} = \tC_{B'}\tC_{C'}(\C_{A'}\theta_{D'}{}^{A'}-\tilde\gamma_{D'E'F'}\theta^{E'F'}),
 \label{CommutatorForSigma1}
\end{equation}
and therefore
\begin{align*}
\tC_{B'}\tC_{C'}\left(\sigma^{1}_{D'}-2(\C_{A'}\theta_{D'}{}^{A'}-\tilde\gamma_{D'E'F'}\theta^{E'F'}) \right)=0.
\end{align*}
Thus the result \eqref{sigma1} follows.
\end{proof}

The explicit expression for $\gamma_{A'BC}$ in terms of $\theta_{A'B'}$ is thus:
\begin{equation}
 \gamma_{A'BC} = 2o_{B}o_{C} \left( \C_{B'}\theta_{A'}{}^{B'}-2(\tC_{B'}\theta_{C'A'})\theta^{B'C'} + k_{A'} \right)
 -2o_{(B}\iota_{C)}\tC_{B'}\theta^{B'}{}_{A'}.
\end{equation}
Imposing the Einstein-Weyl equations, we can replace $\theta_{A'B'}$ by its expression \eqref{Potentialtheta} 
in terms of the solution $\Phi$ to the conformal HH equation \eqref{conformalHHeq} and the field $\tau_{A'B'}$, 
which consists of terms with $\pi^{A'}$ and constants of integration. 
Inserting then the resulting expression of $\gamma_{A'BC}$ into the ASD curvature, we obtain 
$X^{\rm f}_{ABCD}$ as a function of $\Phi$ and $\tau_{A'B'}$.

\medskip
In summary, we see that the SD Weyl spinor $\tilde{\Psi}_{A'B'C'D'}$ depends only on the solution $\Phi$ to \eqref{conformalHHeq}, 
and the ASD Weyl spinor $\Psi_{ABCD}$ depends on $\Phi$ but also on constants of integration. 
This is analogous to the results we obtained for Maxwell and Yang-Mills fields. 
More explicitly, the Weyl scalars $\Psi_{3}$ and $\Psi_{4}$ can be computed in terms of $\sigma^{1}_{A'}$ 
as given by formulas \eqref{Psi3} and \eqref{Psi4},
and for $\Psi_{2}$ we use \eqref{Psi2Lambdaw} and \eqref{LambdaCtilde1}, which give 
\begin{subequations}
\begin{align}
 \Psi_{2} ={}& -\tfrac{1}{3}\tC_{A'}\tC_{B'}\theta^{A'B'}, \\
 \Psi_{3} ={}& \tC_{A'}(\C_{B'}\theta^{A'B'}-\tilde\gamma^{A'B'C'}\theta_{B'C'}) + \tfrac{1}{2}\tC_{A'}k^{A'}, \\
 \Psi_{4} ={}& 2\C_{A'}(\C_{B'}\theta^{A'B'}-\tilde\gamma^{A'B'C'}\theta_{B'C'}) + \C_{A'}k^{A'}.
\end{align}
\end{subequations}

\subsection{Twistor surfaces and Pleba\'nski coordinates}\label{Sec:Coordinates}

In the previous section we reconstructed the full conformal curvature in a covariant (coordinate-free) form, 
from a solution to the conformal HH equation \eqref{conformalHHeq} together with constants of integration.
In this section we will see how to obtain, in a coordinate-dependent manner, the general structure of an 
arbitrary metric in the conformal class, based on the original insights of Pleba\'nski and collaborators 
in the hyper-heavenly construction \cite{PlebanskiRobinson, PlebanskiFinley, Boyer}. 
This general structure does not depend on any field equations.
In the next section we will see the relation between 
this coordinate-dependent expression and the abstact objects we considered in previous sections.
In what follows, besides the usual abstract spinor indices that we have been using,
we will also use {\em concrete numerical indices} ${\bf A,B,}...=0,1$ and ${\bf A',B',}...=0',1'$.

\medskip
Our construction so far is based on the use of the preferred unprimed spin frame $(o^{A},\iota^{A})$, 
but we did not need to specify anything about primed spinors.
It turns out that the special structure of the geometries we consider is sufficiently rich so as to also single out 
a {\em primed} spin frame, and this is needed for our approach to the reconstruction of the metric.

\medskip
In the hyper-heavenly construction, Pleba\'nski {\it et al} use four special complex coordinates 
(denoted by $(u,v,x,y)$ in \cite{PlebanskiRobinson}), which arise roughly as follows.
If the manifold is complex, or real-analytic with complexification $\mathbb{C}M$, 
then the involutivity of the distribution $\tilde{L}$ implies the existence of a foliation of $\mathbb{C}M$ by
2-dimensional complex surfaces $\Sigma\subset\mathbb{C}M$ such that $T\Sigma = \tilde{L}|_{\Sigma}$.
These `totally null surfaces' may be considered as the basic object of twistor theory, 
and we call them {\em twistor surfaces}\footnote{In the 
terminology of Pleba\'nski and collaborators, these are called {\em null strings}. In twistor terminology, we can also call them 
{\em $\beta$-surfaces} (since, according to our conventions, the surface element is anti-self-dual), see \cite{PR2}.}.
Then the four special complex coordinates for $\mathbb{C}M$ are given by: 
two coordinates {\em constant} on the surfaces, and two coordinates {\em along} the surfaces.
The original HH construction is entirely formulated and expressed in these coordinates 
\cite{PlebanskiRobinson, PlebanskiFinley, Boyer}.

\subsubsection*{Coordinates constant on twistor surfaces}

At the level of the cotangent bundle, the almost-complex structure $J$ given by \eqref{JGR}
produces the decomposition $T^{*}M\otimes\mathbb{C}= L^{*}\oplus\tilde{L}^{*}$, 
where $L^{*}$ and $\tilde{L}^{*}$ are the $(+i)$- and $(-i)$-eigenbundles 
(and also the duals of $L$ and $\tilde{L}$). The $(+i)$-subbundle is 
$L^{*}=\{\omega_{a}\in T^{*}\otimes\mathbb{C} \;|\; \omega_{a}=o_{A}\mu_{A'} \}$,
for fixed $o_{A}$ and varying $\mu_{A'}$.
The integrability condition implies the existence of two coordinates $z^{\bf A'}=(z^{0'},z^{1'})$ 
(denoted by $(u,v)$ in \cite{PlebanskiRobinson})
such that ${\rm d}z^{\bf A'}\in L^{*}$, where the twistor surfaces are given by $z^{\bf A'}={\rm constant}$.
We then use ${\rm d}z^{\bf A'}$ as two elements of a coframe $e^{\bf AA'}=(e^{0{\bf A'}}, e^{1{\bf A'}})$:
\begin{equation}
 e^{0{\bf A'}} \equiv {\rm d}z^{\bf A'}. \label{Special1Forms}
\end{equation}
Since ${\rm d}z^{\bf A'}\in L^{*}$, there must exist two spinors $\zeta^{0'}_{A'}, \zeta^{1'}_{A'}$ 
with $\epsilon^{A'B'}\zeta^{0'}_{A'}\zeta^{1'}_{B'}\neq0$, such that, in abstract indices\footnote{We could also use a notation like 
$\omega=\omega_{BB'}{\rm d}x^{BB'}$; however, while the identification $\omega_{BB'}\equiv \omega_{b}$
is well-established, the object ``${\rm d}x^{BB'}$'' might be misleading here since, on curved spacetimes, 
it is not really the exterior derivative of four scalar fields.},
\begin{equation}
 e^{0{\bf A'}}_{b}=o_{B}\zeta^{\bf A'}_{B'}. \label{SpecialPSF}
\end{equation}
We require \eqref{Special1Forms} (or \eqref{SpecialPSF}) to be ordinary 1-forms, 
in the sense that they must have vanishing $p$- and $w$-weights.
This implies that the spinors $\zeta^{\bf A'}_{A'}$ must have $p(\zeta^{\bf A'}_{A'})=-1$, $w(\zeta^{\bf A'}_{A'})=-1$,
which in turn implies that $\epsilon^{A'B'}\zeta^{0'}_{A'}\zeta^{1'}_{B'}\neq1$, so $\zeta^{\bf A'}_{A'}$ is a non-normalized spin frame. 
We can express this as
\begin{equation}
 \epsilon^{A'B'}\zeta^{\bf A'}_{A'}\zeta^{\bf B'}_{B'}=N\epsilon^{\bf A'B'} \label{NormalizationPSF}
\end{equation}
for some function $N\neq 0$ with weights $p(N)=-2$, $w(N)=-3$,  
and where $\epsilon^{\bf A'B'}=-\epsilon^{\bf B'A'}$, $\epsilon^{0'1'}=1$.
Thus, we can take the other two elements of the coframe, that is $e^{1{\bf A'}}=e^{1{\bf A'}}_{b}{\rm d}x^{b}$, as
\begin{equation}
 e^{1{\bf A'}}_{b} = N^{-1}\iota_{B}\zeta^{\bf A'}_{B'}. \label{coframe1}
\end{equation}
One can then check that the usual abstract index expression for the metric is recovered, 
namely $g_{ab}=2\epsilon_{\bf A'B'}e^{0{\bf A'}}_{(a}e^{1{\bf B'}}_{b)}=\epsilon_{AB}\epsilon_{A'B'}$.

\begin{proposition}
The spinor fields $\zeta^{\bf A'}_{A'}$ defined in \eqref{SpecialPSF} satisfy the following:
\begin{equation}
 \tC_{A'}\zeta^{\bf B'}_{B'}=0, \qquad \C_{A'}\zeta^{{\bf B'} A'}=0. \label{PropertiesSpecialPSF}
\end{equation}
\end{proposition}

\begin{proof}
We have ${\rm d}^{2}z^{\bf A'}=0$, so, from \eqref{SpecialPSF}, this implies $\partial_{[a}e^{1{\bf A'}}_{b]}=0$.
Since $e^{1{\bf A'}}_{a}$ has vanishing weights, it holds $\partial_{[a}e^{1{\bf A'}}_{b]}=\C_{[a}e^{1{\bf A'}}_{b]}$. 
Replacing \eqref{SpecialPSF}, the result follows.
\end{proof}

We can use the object $\epsilon^{\bf A'B'}$ and its inverse $\epsilon_{\bf A'B'}$ to raise and lower concrete 
indices, analogously to abstract spinor indices. 
For example, we define $\zeta^{A'}_{\bf A'}:=\epsilon^{A'B'}\epsilon_{\bf B'A'}\zeta^{\bf B'}_{B'}$. 
Notice that with this definition, we have $\zeta^{\bf A'}_{A'}\zeta^{A'}_{\bf B'}=-N\delta^{\bf A'}_{\bf B'}$.

\smallskip
From \eqref{PropertiesSpecialPSF}, we see that the existence of twistor surfaces gives origin to 
a primed, non-normalized spin frame that is parallel under $\tC_{A'}$.
In particular, then, we can use $\zeta^{\bf A'}_{A'}$ to define a parallel, normalized spin frame 
by $\tilde\alpha^{\bf A'}_{A'} = iN^{-1/2}\zeta^{\bf A'}_{A'}$, whose weights are 
$p(\tilde\alpha^{\bf A'}_{A'})=0$ and $w(\tilde\alpha^{\bf A'}_{A'})=1/2$. 
The factor `$i$' in $\tilde\alpha^{\bf A'}_{A'} = iN^{-1/2}\zeta^{\bf A'}_{A'}$ ensures that the 
`dual' of $\tilde\alpha^{\bf A'}_{A'}$ is $\tilde\alpha^{A'}_{\bf A'}=\epsilon^{A'B'}\epsilon_{\bf B'A'}\tilde\alpha^{\bf B'}_{B'}$, 
namely $\tilde\alpha^{\bf A'}_{A'}\tilde\alpha^{A'}_{\bf B'}=\delta^{\bf A'}_{\bf B'}$.
All the formulas for curvature, connections, etc. given in section \ref{Sec:RecCurvature} are, therefore, 
valid if we use this particular spin frame. 
In addition, since $p(N^{1/2})=-1$ and $\tC_{A'}N^{1/2}=0$, we see that $N$ gives us 
a specific choice for the scalar field $\mr\lambda$ introduced in \eqref{lambdao}, that is, we can set:
\begin{equation}
 \mr\lambda \equiv N^{1/2}. \label{Choicelambdao}
\end{equation}

\begin{proposition}
With the choice \eqref{Choicelambdao}, the relation between the spinor field $\theta_{A'B'}$ introduced in 
\eqref{theta0} and the primed, weighted spin frame $\zeta^{\bf A'}_{A'}$ defined by \eqref{SpecialPSF} is:
\begin{equation}
 2\tC_{B'}\theta_{C'A'} = -N^{-1}\zeta^{\bf D'}_{B'}\C_{A'}\zeta_{{\bf D'}C'}.\label{theta-PSF}
\end{equation}
\end{proposition}

\begin{proof}
Using the spin frame $\tilde\alpha^{\bf A'}_{A'}=iN^{-1/2}\zeta^{\bf A'}_{A'}$ and its dual, 
we can define the `connection form' $\tilde\gamma_{A'B'C'}$ as in \eqref{gammaGravity0}. 
Defining $\mr\lambda$ by \eqref{Choicelambdao}, 
recalling that $\gamma_{A'01}=\mr\lambda^{-1}\C_{A'}\mr\lambda$, 
and expressing $\tilde\gamma_{A'B'C'}$ in terms of $\mr\lambda=N^{1/2}$ and $\zeta^{\bf A'}_{A'}$, 
a short calculation gives
\begin{align*}
 \tilde\gamma_{A'B'C'}+\epsilon_{B'C'}\gamma_{A'01} = -N^{-1}\zeta^{\bf D'}_{B'}\C_{A'}\zeta_{{\bf D'}C'}.
\end{align*}
Recalling now formula \eqref{theta0}, the relation \eqref{theta-PSF} follows.
\end{proof}

The result \eqref{theta-PSF} will be key below to connect the expression for the metric and the solution 
to the conformal HH equation.

\subsubsection*{Coordinates along twistor surfaces}

In order to get an expression for the metric, we first need to understand how
the second pair of complex coordinates arises. 
The frame of vector fields dual to the coframe $e^{\bf AA'}$ will be denoted by $e_{\bf AA'}=e^{b}_{\bf AA'}\partial_{b}$.
It is defined by the condition $e^{\bf AA'}(e_{\bf BB'})=\delta^{\bf A}_{\bf B}\delta^{\bf A'}_{\bf B'}$.
One can check that
\begin{subequations}
\begin{align}
 e^{b}_{0{\bf A'}} ={}& -N^{-1}\iota^{B}\zeta^{B'}_{\bf A'}, \label{SpecialFrame0} \\
 e^{b}_{1{\bf A'}} ={}& o^{B}\zeta^{B'}_{\bf A'}. \label{SpecialFrame1}
\end{align}
\end{subequations}
Notice that $p(e_{0{\bf A'}})=0$, $w(e_{0{\bf A'}})=0$ and $p(e_{1{\bf A'}})=0$, $w(e_{1{\bf A'}})=-2$.

\begin{proposition}
Let $\phi$ be the scalar field defined by the Lee form \eqref{fgradient}, and let $e_{1{\bf A'}}$ be the 
two vector fields \eqref{SpecialFrame1}.
Local coordinates $\tilde{z}^{\bf A'}=(\tilde{z}^{\bf 0'}, \tilde{z}^{\bf 1'})$ along the twistor surfaces are given by
\begin{equation}
 \frac{\partial}{\partial\tilde{z}^{\bf A'}} = \phi^{-2}e_{1{\bf A'}}. \label{Coordinates2}
\end{equation}
\end{proposition}

\begin{proof}
Using \eqref{PropertiesSpecialPSF}, a short calculation shows that the vector fields
defined by the right hand side of \eqref{Coordinates2} commute; therefore, they define local coordinates. 
Notice that the factor $\phi^{-2}$ is essential to ensure the commutation property, and it also 
ensures that the vector fields $\partial_{\tilde{z}^{\bf A'}}$ have conformal weight zero.
\end{proof}

We recall that, in terms of a null coframe $e^{\bf AA'}$, the metric can be written as
\begin{equation}
 g = e^{00'}\otimes e^{11'} + e^{11'}\otimes e^{00'} - e^{01'}\otimes e^{10'} - e^{10'}\otimes e^{01'}.
\end{equation}
In abstract indices, this is $g_{ab}=\epsilon_{\bf AB}\epsilon_{\bf A'B'}e^{\bf AA'}_{(a}e^{\bf BB'}_{b)}$.

\begin{proposition}[See \cite{PlebanskiFinley}]
In terms of the complex coordinates $(z^{\bf A'},\tilde{z}^{\bf A'})$,
the metric has the expression\footnote{Our convention for the symmetric tensor product is 
$a\odot b := \frac{1}{2}(a\otimes b + b\otimes a) $.}
\begin{equation}\label{HHmetricCoordinates}
 g = 2\phi^{-2}\epsilon_{\bf A'B'}{\rm d}z^{\bf A'}\odot{\rm d}\tilde{z}^{\bf B'}+2Q_{\bf A'B'}{\rm d}z^{\bf A'}\odot{\rm d}z^{\bf B'},
\end{equation}
where we define
\begin{equation}
 Q_{\bf A'}{}^{\bf B'} := \phi^{-2}e_{0{\bf A'}}(\tilde{z}^{\bf B'}). \label{QPlebanski}
\end{equation}
\end{proposition}

\begin{proof}
We have 
\begin{equation*}
g = \epsilon_{\bf AB}\epsilon_{\bf A'B'}e^{\bf AA'}\otimes e^{\bf BB'} = 2\epsilon_{\bf A'B'} e^{0{\bf A'}}\odot e^{1{\bf B'}} 
= 2\epsilon_{\bf A'B'} {\rm d}z^{\bf A'}\odot e^{1{\bf B'}}, 
\end{equation*}
so we need to find an expression for $e^{1{\bf B'}}$ in terms of $\tilde{z}^{\bf A'}, z^{\bf A'}$. 
To do this, we notice that
\begin{align*}
 {\rm d}\tilde{z}^{\bf A'} ={}& e_{\bf BB'}(\tilde{z}^{\bf A'}) e^{\bf BB'}  \\
 ={}& e_{0{\bf B'}}(\tilde{z}^{\bf A'}) e^{0{\bf B'}} + e_{1{\bf B'}}(\tilde{z}^{\bf A'}) e^{1{\bf B'}}  \\
 ={}& e_{0{\bf B'}}(\tilde{z}^{\bf A'}){\rm d}z^{\bf B'} + \phi^{2}e^{1{\bf A'}}
\end{align*}
where in the last line we used \eqref{Special1Forms} and \eqref{Coordinates2}; the latter implying that
$e_{1{\bf B'}}(\tilde{z}^{\bf A'})=\phi^{2}\delta^{\bf A'}_{\bf B'}$. Therefore:
\begin{equation}
 e^{1{\bf B'}} =  \phi^{-2}{\rm d}\tilde{z}^{\bf B'}  - \phi^{-2}e_{0{\bf C'}}(\tilde{z}^{\bf B'}){\rm d}z^{\bf C'}. \label{coframe1Coord}
\end{equation}
Defining then $Q_{\bf A'}{}^{\bf B'}$ via \eqref{QPlebanski}, the result follows.
\end{proof}

\begin{remark}\label{Remark:FlatMetric}
As noticed by Pleba\'nski and collaborators \cite{PlebanskiRobinson, PlebanskiFinley}, the tensor field
\begin{equation}
 \eta:=2\epsilon_{\bf A'B'}{\rm d}z^{\bf A'}\odot{\rm d}\tilde{z}^{\bf B'} 
 = 2({\rm d}z^{0'}{\rm d}\tilde{z}^{1'}-{\rm d}z^{1'}{\rm d}\tilde{z}^{0'}), \label{FlatMetric}
\end{equation}
is simply a flat metric expressed in (complex) double null coordinates (see e.g. \cite[Chapter 2]{MW}).
Therefore, the metric \eqref{HHmetricCoordinates} consists of a conformally flat part $\phi^{-2}\eta$
plus a $Q_{\bf A'B'}$-dependent part that codifies the non-trivial curvature of the conformal structure.
\end{remark}

\begin{remark}\label{Remark:pi}
Using $\tC_{A'}=\delta^{B'}_{A'}o^{B}\C_{BB'}$, the identity $\delta^{B'}_{A'}=-N^{-1}\zeta^{\bf A'}_{A'}\zeta^{B'}_{\bf A'}$, 
the definition $e^{b}_{1{\bf C'}}=o^{B}\zeta^{B'}_{\bf C'}$, and the result \eqref{Coordinates2}, 
a short calculation gives
\begin{equation}
 \tC_{A'}\tilde{z}^{\bf D'} = -N^{-1}\phi^{2}\zeta^{\bf D'}_{A'}. \label{tCtz}
\end{equation}
(Recall that $\tilde{z}^{\bf D'}$ are ordinary scalar fields.)
Multiplying by $\phi^{-2}\zeta^{B'}_{\bf D'}$, it follows that
\begin{equation}
 \tC_{A'}\pi^{B'} = \delta^{B'}_{A'}, \qquad \text{where } \qquad
 \pi^{B'} := \phi^{-2}\zeta^{B'}_{\bf D'}\tilde{z}^{\bf D'}. \label{RelationPiZetaTilde}
\end{equation}
This equation gives us the explicit relation between coordinates on twistor surfaces and 
a solution to the twistor-like equation \eqref{TwistorLikeEq}.
\end{remark}

\subsection{The conformal metric and the almost-complex structure}\label{Sec:AbstractMetric}

Having determined the general structure \eqref{HHmetricCoordinates} of the metric (for which we 
followed Pleba\'nski {\it et al} \cite{PlebanskiRobinson, PlebanskiFinley, Boyer}) together with some 
relations to abstract objects we used in previous sections (such as \eqref{Choicelambdao}-\eqref{theta-PSF}), 
we now want to understand how, in case the Einstein-Weyl equations \eqref{EinsteinWeylEquations0} are satisfied,
the solution to the conformal HH equation \eqref{conformalHHeq} enters the metric.

The first step to do this is to convert \eqref{HHmetricCoordinates} to an abstract expression.
Using \eqref{FlatMetric} and \eqref{Special1Forms}, a short calculation shows that  
\eqref{HHmetricCoordinates} can be written in abstract indices as
\begin{equation}
 g_{ab} = \phi^{-2}\eta_{ab} + 2o_{A}o_{B}Q_{A'B'} \label{AbstractMetric0}
\end{equation}
where we defined the spinor field
\begin{equation}
 Q_{A'B'} := \zeta^{\bf A'}_{A'}\zeta^{\bf B'}_{B'}Q_{\bf A'B'}. \label{AbstractQPlebanski}
\end{equation}
Notice that $p(Q_{A'B'})=-2$ and $w(Q_{A'B'})=0$.

\begin{proposition}\label{Prop:RelationwP}
The relation between the spinor field $Q_{A'B'}$ defined in \eqref{AbstractQPlebanski} and the 
spinor field $\theta_{A'B'}$ given by \eqref{theta0} is
\begin{equation}
 Q_{A'B'} = -2\theta_{A'B'}. \label{RelationwP}
\end{equation}
\end{proposition}

\begin{proof}
By definition \eqref{QPlebanski}, we have 
$Q_{\bf A'}{}^{\bf B'} = \phi^{-2}e^{c}_{0{\bf A'}}\partial_{c}\tilde{z}^{\bf B'}$.
Since $\tilde{z}^{\bf B'}$ are ordinary scalars, it holds $\partial_{c}\tilde{z}^{\bf B'}=\C_{c}\tilde{z}^{\bf B'}$.
Using then \eqref{SpecialFrame0}: $Q_{\bf A'}{}^{\bf B'}=-\phi^{-2}N^{-1}\zeta^{C'}_{\bf A'}\C_{C'}\tilde{z}^{\bf B'}$.
In terms of the spinor field \eqref{AbstractQPlebanski},  
this is $Q_{A'}{}^{B'} = \phi^{-2}\zeta^{B'}_{\bf D'}\mathcal{C}_{A'} p^{\bf D'}$,
where we used the identity $\zeta^{\bf A'}_{A'}\zeta^{C'}_{\bf A'}=-N\delta^{C'}_{A'}$.

\smallskip
Lowering an index, we have $Q_{B'C'}=\phi^{-2}\zeta_{{\bf D'} C'}\C_{B'} \tilde{z}^{\bf D'}$, and
since $\tC_{A'}\zeta_{{\bf D'} C'} = 0 = \tC_{A'}\phi$, and $\tC_{A'}$ and $\C_{A'}$ commute on ordinary scalars, we get
$\tC_{A'}Q_{B'C'} = \phi^{-2}\zeta_{{\bf D'} C'}\C_{B'}\tC_{A'}\tilde{z}^{\bf D'}$.
Using \eqref{tCtz} and $\C_{B'}\phi=0$, we have $\tC_{A'}Q_{B'C'} = -\zeta_{{\bf D'} C'}\C_{B'}(N^{-1}\zeta^{\bf D'}_{A'})$,
which can be rewritten as
\begin{align*}
 \tC_{A'}Q_{B'C'} = N^{-1}\zeta^{\bf D'}_{A'}\C_{B'}\zeta_{{\bf D'}C'},
\end{align*}
where we used $N^{-1}\zeta_{{\bf D'}C'}\zeta^{\bf D'}_{A'}=\epsilon_{C'A'}$.
Recalling now the result \eqref{theta-PSF}, we see that 
\begin{align*}
 \tC_{A'}(Q_{B'C'}+2\theta_{B'C'}) = 0,
\end{align*}
from where we deduce that $Q_{A'B'} = -2(\theta_{A'B'} + k_{A'B'})$ for some $k_{A'B'}$ with $\tC_{A'}k_{B'C'}=0$. 
But we noticed in remark \ref{Remark:Freedomtheta}
that for $\theta_{A'B'}$ we have the freedom \eqref{Freedomtheta}, which means that 
$k_{A'B'}$ can be absorbed in the definition of $\theta_{A'B'}$. Thus, the relation \eqref{RelationwP} follows.
\end{proof}

Imposing now the Einstein-Weyl equations \eqref{EinsteinWeylEquations0},
the constructions of previous sections apply, so there exists a scalar field $\Phi$ that solves the 
conformally invariant HH equation \eqref{conformalHHeq} and such that the SD Weyl spinor is \eqref{PotentializationWeyl}, 
and the spinor field $\theta_{A'B'}$ is given in terms of $\Phi$ and constants of integrations by formula \eqref{Potentialtheta}.
Replacing these expressions, we arrive at:

\begin{proposition}\label{Prop:CovariantMetric}
Assume the Einstein-Weyl equations \eqref{EinsteinWeylEquations0} hold.
In terms of the scalar field $\Phi$ that solves the conformal HH equation \eqref{conformalHHeq}, 
any metric in the conformal class can be expressed as
\begin{equation}\label{CovariantMetric}
 g_{ab} = \phi^{-2}\eta_{ab} + c_{ab} 
 -4 \mr\Omega^{3}\nabla_{c}[\mr\Omega^{-2}\phi^{-4}\nabla_{d}(P_{(a}{}^{c}{}_{b)}{}^{d}\phi^{4}\Phi)]
\end{equation}
where $\nabla_{a}$ is the Levi-Civita connection of $g_{ab}$,
$c_{ab} \equiv -4 o_{A}o_{B}\tau_{A'B'}$ represents constants of integration, 
and the tensor field $P_{abcd}$ is defined by $P_{abcd} = 4o_{A}o_{B}o_{C}o_{D}\epsilon_{A'B'}\epsilon_{C'D'}$.
\end{proposition}

Regarding the almost-complex structure \eqref{JGR}, we have: 
\begin{proposition}
In the coordinates $(z^{\bf A'},\tilde{z}^{\bf A'})$, the almost-complex structure \eqref{JGR} is
\begin{equation}
 J = i\frac{\partial}{\partial z^{\bf A'}}\otimes{\rm d}z^{\bf A'} - i\frac{\partial}{\partial\tilde{z}^{\bf A'}}\otimes{\rm d}\tilde{z}^{\bf A'} 
  +2i\phi^{2}Q_{\bf A'}{}^{\bf B'}\frac{\partial}{\partial\tilde{z}^{\bf B'}}\otimes{\rm d}z^{\bf A'}. \label{JinCoordinates}
\end{equation}
\end{proposition}

\begin{proof}
Using the original expression \eqref{JGR} for $J$, together with
\eqref{SpecialPSF}, \eqref{coframe1}, \eqref{SpecialFrame0}, \eqref{SpecialFrame1}, and 
the identity $\delta^{A'}_{B'}=-N^{-1}\zeta^{\bf A'}_{B'}\zeta^{A'}_{\bf A'}$,
a short calculation shows that $J=ie_{0{\bf A'}}\otimes e^{0{\bf A'}}-ie_{1{\bf A'}}\otimes e^{1{\bf A'}}$.
Using now \eqref{Special1Forms}, \eqref{Coordinates2}, \eqref{coframe1Coord}, and 
\begin{align}
 e_{0{\bf A'}} = \frac{\partial}{\partial z^{\bf A'}} + \phi^{2}Q_{\bf A'}{}^{\bf B'}\frac{\partial}{\partial\tilde{z}^{\bf B'}}, 
 \label{frame0Coord}
\end{align}
we get \eqref{JinCoordinates}. Equation \eqref{frame0Coord} is obtained by noticing that we must have 
$e_{0{\bf A'}}(z^{\bf B'})=\delta^{\bf B'}_{\bf A'}$ and $e_{0{\bf A'}}(\tilde{z}^{\bf B'})=\phi^{2}Q_{\bf A'}{}^{\bf B'}$.
\end{proof}

From \eqref{JinCoordinates} we see that while $\partial/\partial\tilde{z}^{\bf A'}$ are anti-holomorphic vector fields for $J$,
the presence of the $Q_{\bf A'}{}^{\bf B'}$-term implies that the vector fields $\partial/\partial z^{\bf A'}$ 
are not holomorphic w.r.t. $J$.
Notice that the first two terms in \eqref{JinCoordinates}, namely 
\begin{equation}
 J_{0} \equiv i\frac{\partial}{\partial z^{\bf A'}}\otimes{\rm d}z^{\bf A'} - i\frac{\partial}{\partial\tilde{z}^{\bf A'}}\otimes{\rm d}\tilde{z}^{\bf A'},
\end{equation}
give an orthogonal complex structure for the conformally flat metric $\phi^{-2}\eta$, that is $(J_{0})^{2}=-\mathbb{I}$, 
$\phi^{-2}\eta(J_{0}\cdot, J_{0}\cdot)=\phi^{-2}\eta(\cdot,\cdot)$.

\subsection{Discussion}\label{Sec:Discussion}

From proposition \ref{Prop:CovariantMetric}, we see that the conformal metric can be expressed as 
$g_{ab}=\phi^{-2}\eta_{ab} + c_{ab} + h_{ab}$, where the tensor field $h_{ab}$ is defined by
\begin{equation}\label{HertzPotential}
 h_{ab} =  -4\mr\Omega^{3}\nabla_{c}[\mr\Omega^{-2}\phi^{-4}\nabla_{d}(P_{(a}{}^{c}{}_{b)}{}^{d}\phi^{4}\Phi)].
\end{equation}
We can then distinguish three different contributions to $g_{ab}$. 

First, as noticed in remark \ref{Remark:FlatMetric}, the tensor field $\phi^{-2}\eta_{ab}$ is a conformally flat metric; 
in other words, this gives a {\em flat conformal structure}.
It corresponds to setting $\Phi\equiv 0$ and all `constants of integrations' to zero.
Both SD and ASD Weyl spinors vanish in this case: $\tilde{\Psi}_{A'B'C'D'}\equiv0$ and $\Psi_{ABCD}\equiv0$.

Second, the part $c_{ab}$ is only associated to `integration constants', since it only involves 
$\tau_{A'B'}$, whose explicit expression is \eqref{DoubleTwistorLikeEq}. 
It corresponds to setting $\Phi\equiv 0$, which leaves a {\em half-flat conformal structure}
since the SD Weyl spinor then vanishes, $\tilde\Psi_{A'B'C'D'}\equiv 0$. 
In other words, this corresponds to a (possibly complex) conformally hyperk\"ahler structure.
The ASD Weyl spinor, in turn, only depends on constants of integration.

Finally, the third contribution to $g_{ab}$ is given by the tensor field \eqref{HertzPotential}, 
and it encodes the full curvature of the conformal structure.
We find this term particularly interesting for the following reason.
Setting $\mr\Omega\equiv1$ (i.e., breaking conformal invariance and restricting to an Einstein metric in the conformal class, 
as in section \ref{Sec:OrdinaryEE}), \eqref{HertzPotential} coincides exactly with 
the expression for a ``Hertz potential'' in {\em linearized} gravity (on curved backgrounds). 
In perturbation theory, the (ordinary) linearized Einstein equations are reduced to a scalar, linear, wave-like equation 
known as `Debye equation' or `adjoint Teukolsky equation',
and a Hertz potential is the reconstruction of the linearized metric in terms of a solution to this equation.
Here, however, all of our results are valid for the full, non-linear Einstein(-Weyl) equations, and there 
are no linearizations involved whatsoever.

Remarkably enough, the {\em full non-linear} curvature of the conformal structure is still coming from a ``Hertz potential'' 
(more precisely, its conformally covariant version \eqref{HertzPotential}), 
but now the scalar field $\Phi$ is required to satisfy the non-linear conformal HH equation \eqref{conformalHHeq}.

\medskip
Notice that, at the linear level, i.e. keeping only linear terms in $\Phi$, and setting the integration constant $K$ to zero,
the conformal HH equation \eqref{conformalHHeq} is $(\Box^{\C}+3\mathcal{R}/2)\Phi = 0$.
We can express this in GHP notation by using identity \eqref{boxC2} 
and recalling that $p(\Phi)=-4$, $w(\Phi)=-1$ and that $\mathcal{R}=-12\Psi_{2}$ (eq. \eqref{ScalarCurvatureCSpecial}). 
After a short calculation (using also \cite[Eq. (4.12.32)(f)]{PR1}) we get
\begin{align}
 (\Box^{\C}+3\mathcal{R}/2)\Phi =2[(\tho'-\tilde\rho')(\tho+3\rho)-(\edt'-\tilde\t)(\edt+3\t)-3\Psi_{2}]\Phi = 0.
 \label{DebyeEqGravGHP}
\end{align}
This is exactly the Debye (or adjoint Teukolsky) equation for gravitational perturbations.
Thus, from formulas \eqref{HertzPotential} (for $\mr\Omega=1$) and \eqref{DebyeEqGravGHP}, 
we see explicitly that the metric reconstruction and Teukolsky equation from perturbation theory 
in 4d general relativity are the linearized version of the conformally invariant HH construction obtained in this paper. 
This answers in the affirmative a question posed in \cite[Remark 3.14]{Araneda20}.

\section{Summary and conclusions}\label{Sec:Conclusions}

Inspired by the connections between integrability in complex geometry and in differential equations
that are derived from twistor theory and the heavenly equations in the hyperk\"ahler setting, and motivated by 
the observation that the integrability of an orthogonal almost-complex structure is conformally invariant,
in this work we have studied the potentialization of 4d closed Einstein-Weyl structures with 
a half-algebraically special Weyl tensor and a half-integrable almost-complex structure 
(equivalently, a shear-free spinor field), 
as well as the potentialization of the Dirac-Weyl, Maxwell, and Yang-Mills systems.
Our method was based on the construction of a conformally invariant connection, $\C_{AA'}$, associated 
to a choice of almost-complex structure $J$, that, when the half-integrability properties of $J$
are taken into account, produces a flat `connection' $\tC_{A'}=o^{A}\C_{AA'}$, 
that is at the heart of the integration machinery, 
since it gives origin to a de Rham complex and to parallel frames for all the systems considered.

\smallskip
As with any integration process, the general solution to the differential equations we studied 
depends on a sort of `constants of integration', that in our approach take the form of 
spinor fields in the kernel of $\tC_{A'}$.
The equations can be solved formally in terms of these constants of integration 
and a spinor field $\pi^{A'}$ that satisfies the twistor-like equation \eqref{TwistorLikeEq}, as
discussed in section \ref{Sec:ConstantsOfIntegration}.
These two objects correspond to an abstract, conformally invariant version of the coordinates 
$(\tilde{z}^{\bf A'},z^{\bf A'})$
adapted to twistor surfaces in which the original HH construction is based: 
our `constants of integration' correspond to functions which only depend on $z^{\bf A'}$, 
and the spinor field $\pi^{A'}$ is an abstract version of $\tilde{z}^{\bf A'}$ 
($z^{\bf A'}$ and $\tilde{z}^{\bf A'}$ are, respectively, coordinates {\em constant} and {\em along} 
twistor surfaces, see section \ref{Sec:Coordinates}).
The different dependence of each part of the reconstructed fields on integration constants
and potentials is subtle, see below.

\smallskip
The concept of ``potentialization'' is here understood as follows. The field equations lead to two main consequences.
First, {\em primed} spinor fields can be entirely expressed in terms of a scalar 
field\footnote{Here we mean of course {\em spacetime scalar}, since 
for the Yang-Mills case the scalar has internal indices.}: 
a `potential', which is the generalization of the Pleba\'nski potential in the second heavenly equation.
Second, the potential must satisfy a scalar, conformally invariant wave-like equation, which is linear or non-linear 
depending on whether the corresponding original system is itself linear or non-linear. 
Recall that the ordinary scalar conformal wave equation is $(\Box+R/6)\Phi=0$.
For linear systems, the equations we obtained are a generalization of this: 
\begin{equation}
 (\Box^{\C}+n\mathcal{R})\Phi=0, \label{GeneralizedCWE}
\end{equation}
where $n$ is some number depending on the system considered,
and $\Box^{\C}$ and $\mathcal{R}$ are respectively the wave operator and scalar curvature of $\C_{a}$.
Equation \eqref{GeneralizedCWE} includes the Teukolsky equations from 
perturbation theory, see eqs. \eqref{DebyeEqDiracGHP}, \eqref{DebyeEqMaxwellGHP}, \eqref{DebyeEqGravGHP}.
This connection to perturbation theory 
can be understood as a generalization to the conformal HH setting of the fact that 
linearized solutions to the heavenly equations satisfy the ordinary wave equation, as shown in \cite{DunajskiMason}.
For non-linear systems, eq. \eqref{GeneralizedCWE} 
involves also terms non-linear in the potential, as well as a constant of integration in the right hand side.

\smallskip
On the other hand, {\em unprimed} spinor fields are not directly involved in the potentialization,
but they can be related to the potential when one considers the primed and unprimed fields as 
the SD and ASD parts of the curvature of a connection (a fact that is itself independent of the field equations).
Thus, even though only ``one side'' of the geometry is special, both sides are related via the connection form.

\smallskip
A signature issue must be mentioned here: 
in Lorentz signature, reality conditions can be imposed so that the SD and ASD parts are complex-conjugated of each other.
In the Maxwell case, we saw that this can be exploited to eliminate the integration constants.
For the Einstein-Weyl system, the implementation of different reality conditions is more subtle and was not analysed here.
Naturally, each signature has its own peculiarities that deserve special attention; 
for example, in Riemann and split signature, our assumptions \eqref{SFR} and \eqref{rPND} lead to 
a half-type D Weyl tensor together with an (ordinary) integrable complex structure, so several formulas 
and identities get simplified. 
In addition, split signature allows the interesting possibility of real twistor surfaces.
We also mention that the integrability issue for the hyper-heavenly equation is 
subtle and it might also be related to signature, see \cite{Tod} and \cite{DP}.

\smallskip
The Einstein-Weyl system has the additional feature that one is also interested in reconstructing the conformal metric. 
When the field equations are satisfied, three different contributions to the structure of the conformal metric can be distinguished:
a conformally flat part $\phi^{-2}\eta_{ab}$, a conformally half-flat part $c_{ab}$ associated to constants of integration, 
and a potential part $h_{ab}$ \eqref{HertzPotential}, which encodes the full non-linear curvature 
(see sections \ref{Sec:AbstractMetric} and \ref{Sec:Discussion}).
When restricted to an Einstein metric in the conformal class, 
$h_{ab}$ has exactly the structure of a so-called Hertz potential, which is a tensor field 
arising in {\em linearized} Einstein theory (on curved backgrounds) that 
reconstructs the linearized metric in terms of a solution to the Teukolsky equation \eqref{GeneralizedCWE}.
In our case, the same tensor field (or, rather, a conformally covariant version)
arises in the {\em full, non-linear} Einstein(-Weyl) theory, where it satisfies the conformal HH equation \eqref{conformalHHeq}.

\smallskip
Finally, one of our main goals in this work has been to show that conformal invariance 
features prominently in the analysis of the geometric structure and potentialization of special spacetimes, even 
for systems that are not conformally invariant such as the ordinary Einstein equations 
(see section \ref{Sec:OrdinaryEE}, where we show how to deal with such non-conformal systems by 
using a simple trick).
In our approach, this symmetry originates in the fact that the integrability of an orthogonal almost-complex structure $J$
is conformally invariant. 
The connection $\C_{AA'}$ incorporates, by construction, these structures, 
and it also encodes integrability of $J$ in parallel spinors \eqref{ParallelSpinor}.
(Alternatively, in view of \eqref{HalfParallelSpinors}, eq. \eqref{ParallelSpinor} is equivalent to $\tC_{A'}o^{B}=0$, 
which we can interpret as saying that $o^{A}$ is {\em holomorphic}, 
in agreement with \cite[Chapter IV, Theorem 9.11]{Lawson}.)

\bigskip
\noindent
{\bf Acknowledgments.}
The possibility of relations among some of the structures considered in this work 
and Pleba\'nski's hyper-heavenly construction was suggested to me by Maciej Dunajski and Lionel Mason in 2019. 
I am very grateful to them for this suggestion and for conversations at that time, which led to
the development of this work, and also for helpful comments on the current manuscript.
I would also like to thank Steffen Aksteiner and Lars Andersson for discussions on related topics.
This work is supported by a postdoctoral fellowship from the Alexander von Humboldt Foundation.

%%========================================================================

\appendix

\section{Spinor decomposition for Weyl connections}\label{App:WeylConnections}

Let $\nabla^{\rm w}_{a}$ be an arbitrary Weyl connection, with Weyl 1-form ${\rm w}_{a}$. 
The curvature tensor is given by $[\nabla^{\rm w}_{a}, \nabla^{\rm w}_{b}]v^{d} = R^{\rm w}_{abc}{}^{d}v^{c}$.
If $R_{abc}{}^{d}$ is the Riemann curvature of an arbitrary metric in the conformal class, with Levi-Civita connection $\nabla_{a}$,
we have the relation
\begin{equation}
 R^{\rm w}_{abc}{}^{d} = R_{abc}{}^{d}+2\nabla_{[a}K_{b]c}{}^{d} +2K_{[a|e|}{}^{d}K_{b]c}{}^{e},
 \label{CurvatureTensorWeylC}
\end{equation}
where $K_{ab}{}^{c} = {\rm w}_{a}\delta_{b}^{c} + {\rm w}_{b}\delta_{a}^{c} - {\rm w}_{e}g^{ce}g_{ab}$.
We notice that the scalar curvature is
\begin{equation}
 R^{\rm w}:=g^{ac}\delta^{b}_{d}R^{\rm w}_{abc}{}^{d} = R + 6(\nabla_{a}{\rm w}^{a}+{\rm w}_{a}{\rm w}^{a}).
 \label{ScalarCurvatureWeylC}
\end{equation}

By definition, we have the symmetry $R^{\rm w}_{abc}{}^{d} = R^{\rm w}_{[ab]c}{}^{d}$.
In addition, the torsion-free property implies that $R^{\rm w}_{[abc]}{}^{d}=0$. 
From this it follows by a standard argument that the Bianchi identity is satisfied, $\nabla^{\rm w}_{[e}R^{\rm w}_{ab]c}{}^{d}=0$.
Defining $R^{\rm w}_{abcd} = R^{\rm w}_{abc}{}^{e}g_{ed}$, applying another $\nabla^{\rm w}_{a}$ derivative to 
$\nabla^{\rm w}_{a}g_{bc}=-2{\rm w}_{a}g_{bc}$ and taking the commutator, one finds
\begin{equation}
 R^{\rm w}_{abcd} = R^{\rm w}_{ab[cd]} + g_{cd}\partial_{[a}{\rm w}_{b]}.
\end{equation}
We thus see that $R^{\rm w}_{abcd}$ is antisymmetric in the last two indices 
if and only if ${\rm w}_{a}$ is closed, $\partial_{[a}{\rm w}_{b]}=0$.
{\em From now on we will assume that this is the case}. 
The condition $R^{\rm w}_{abcd}=R^{\rm w}_{ab[cd]}$, together with the other symmetries, imply the additional symmetry
$R^{\rm w}_{abcd} = R^{\rm w}_{cdab}$.

\medskip
The algebraic symmetries of $R^{\rm w}_{abcd}$ mentioned above lead to the following spinor decomposition 
(the proof is completely analogous to the discussion in \cite[Section 4.6]{PR1} for the ordinary Riemann tensor):
\begin{equation}
 R^{\rm w}_{abcd} = \tilde{X}^{\rm w}_{A'B'C'D'}\epsilon_{AB}\epsilon_{CD} + \tilde{\Phi}^{\rm w}_{ABC'D'}\epsilon_{A'B'}\epsilon_{CD}
 +\Phi^{\rm w}_{A'B'CD}\epsilon_{AB}\epsilon_{C'D'} + X^{\rm w}_{ABCD}\epsilon_{A'B'}\epsilon_{C'D'}, 
 \label{SpinorDecCurvatureWeylC}
\end{equation}
where $\tilde{X}^{\rm w}_{A'B'C'D'}=\frac{1}{4}\epsilon^{AB}\epsilon^{CD}R^{\rm w}_{abcd}$, etc.
The spinors $\tilde{X}^{\rm w}_{A'B'C'D'}$, $X^{\rm w}_{ABCD}$ can be furthermore decomposed into irreducible pieces as
\begin{subequations}
\begin{align}
 \tilde{X}^{\rm w}_{A'B'C'D'} ={}& \tilde{\Psi}^{\rm w}_{A'B'C'D'} 
 + \tilde{\Lambda}^{\rm w} (\epsilon_{A'C'}\epsilon_{B'D'}+\epsilon_{A'D'}\epsilon_{B'C'}),  \\
  X^{\rm w}_{ABCD} ={}& \Psi^{\rm w}_{ABCD} + \Lambda^{\rm w} (\epsilon_{AC}\epsilon_{BD}+\epsilon_{AD}\epsilon_{BC}),
\end{align}
\end{subequations}
where $\tilde{\Psi}^{\rm w}_{A'B'C'D'} = \tilde{X}^{\rm w}_{(A'B'C'D')}$, $\Psi^{\rm w}_{ABCD} = X^{\rm w}_{(ABCD)}$,
$\tilde{\Lambda}^{\rm w} = \tfrac{1}{6}\tilde{X}^{\rm w}_{A'B'}{}^{A'B'}$, $\Lambda^{\rm w} = \tfrac{1}{6}X^{\rm w}_{AB}{}^{AB}$.

\medskip
From the symmetries $R^{\rm w}_{abcd}=R^{\rm w}_{cdab}$ and $R^{\rm w}_{[abc]d} = 0$ it follows, respectively, that 
\begin{equation}
 \tilde{\Phi}^{\rm w}_{ABC'D'} = \Phi^{\rm w}_{C'D'AB}, \qquad \tilde{\Lambda}^{\rm w} =  \Lambda^{\rm w}.
 \label{CrucialRelationsCurvatureWeylC}
\end{equation}
We notice that, if ${\rm w}_{a}$ is not closed, the first of this relations does not hold.
We also notice that the relation between $\Lambda^{\rm w}$ and the scalar curvature \eqref{ScalarCurvatureWeylC} is
\begin{equation}
 \Lambda^{\rm w} = R^{\rm w} / 24.
\end{equation}

The relation between the curvature spinors of $\nabla^{\rm w}_{a}$ and those of an arbitrary Levi-Civita connection 
in the conformal class are as follows:
\begin{subequations}
\begin{align}
 \Lambda^{\rm w} ={}& \Lambda + \tfrac{1}{4}(\nabla_{e}{\rm w}^{e} + {\rm w}_{e}{\rm w}^{e} ), \label{Weyl-LC1}\\
 \Phi^{\rm w}_{ABC'D'} ={}& \Phi_{ABC'D'} - \nabla_{C'(A}{\rm w}_{B)D'}+{\rm w}_{C'(A}{\rm w}_{B)D'} \label{Weyl-LC2} \\
 \tilde{\Psi}^{\rm w}_{A'B'C'D'} ={}& \tilde{\Psi}_{A'B'C'D'}, \label{Weyl-LC3} \\
 \Psi^{\rm w}_{ABCD} ={}& \Psi_{ABCD}. \label{Weyl-LC4}
\end{align}
\end{subequations}
We see that the totally symmetric parts of $\tilde{X}^{\rm w}_{A'B'C'D'}$ and $X^{\rm w}_{ABCD}$ coincide 
with the ordinary Weyl curvature spinors;
so we can drop the superscript ${}^{\rm w}$ on $\tilde{\Psi}^{\rm w}_{A'B'C'D'}$, $\Psi^{\rm w}_{ABCD}$.

\medskip
Although the Weyl connection is fixed in the conformal class of metrics, it does not map general conformal densities 
to conformal densities. This issue can be fixed by considering a simple modification: 
a covariant derivative $\mathcal{D}_{a}$ acting on objects with conformal weight $w$ by
\begin{equation}
 \mathcal{D}_{a} := \nabla^{\rm w}_{a} + w \: {\rm w}_{a}. \label{ConformalDerGeneral}
\end{equation}
It follows that $\mathcal{D}_{a}g_{bc}=0$, so $\mathcal{D}_{a}$ has the advantage that, 
when working with an expression involving it, we can raise and lower indices unambiguously.
Since $\partial_{[a}{\rm w}_{b]}=0$, the curvature tensors of $\mathcal{D}_{a}$ and $\nabla^{\rm w}_{a}$ agree,
\begin{equation}
 [\mathcal{D}_{a}, \mathcal{D}_{b}] = [\nabla^{\rm w}_{a},\nabla^{\rm w}_{b}].
\end{equation}
The Bianchi identities can be written as $\mathcal{D}_{[e}R^{\rm w}_{ab]cd}=0$. In spinor terms:
\begin{subequations}
\begin{align}
 \mathcal{D}_{B}{}^{A'}\tilde{X}^{\rm w}_{A'B'C'D'} ={}& \mathcal{D}_{B'}{}^{A}\Phi^{\rm w}_{C'D'AB}, \label{BianchiWeyl1} \\
 \mathcal{D}_{B'}{}^{A}X^{\rm w}_{ABCD} ={}& \mathcal{D}_{B}{}^{A'}\Phi^{\rm w}_{CDA'B'}.
\end{align}
\end{subequations}
The irreducible pieces are
\begin{subequations}
\begin{align}
 & \mathcal{D}_{B}{}^{A'}\tilde{\Psi}^{\rm w}_{A'B'C'D'} = \mathcal{D}_{(B'}{}^{A}\Phi^{\rm w}_{C'D')AB}, \\
 & \mathcal{D}_{B'}{}^{A}\Psi^{\rm w}_{ABCD} = \mathcal{D}_{(B}{}^{A'}\Phi^{\rm w}_{CD)A'B'}, \\
 & \mathcal{D}^{AC'}\Phi^{\rm w}_{ABC'D'} + 3\mathcal{D}_{DB'}\Lambda^{\rm w} = 0. \label{ContractedBianchiWeyl}
\end{align}
\end{subequations}
Unlike the Bianchi identities for a Levi-Civita connection, each side of the equations in the 
identities above is a conformal density.

\section{Identities for $\C_{a}$}\label{App:identities}

We will often use the following notation:
subscripts $0$ and $1$ mean contraction with $o^{A}$ and $\iota^{A}$ respectively; that is, 
for an arbitrary spinor-indexed object $\phi_{AB...FGH...LA'...S'}$:
\begin{equation}
 \phi_{00...011...1A'...S'} = \phi_{AB...FGH...LA'...S'}o^{A}o^{B}...o^{F}\iota^{G}\iota^{H}...\iota^{L}. \label{Notation01}
\end{equation}

\subsection{Arbitrary spacetimes}\label{App:curvatureC}

The complex-conformal connection $\C_{a}$ was defined in section \ref{Sec:ConnectionC}.
If $o^{A}, \iota^{A}$ are the (projective) spinor fields that define $\C_{a}$, we have the following identities:
\begin{subequations}
\begin{align}
 \mathcal{C}_{AA'}o^{B} ={}& \sigma^{o}_{A'}\iota_{A}\iota^{B}, \qquad  \sigma^{o}_{A'} \equiv o^{A}o^{B}\nabla_{AA'}o_{B}, \label{Co} \\
 \mathcal{C}_{AA'}\iota^{B} ={}& \sigma^{1}_{A'}o_{A}o^{B}, \qquad  \sigma^{1}_{A'} \equiv \iota^{A}\iota^{B}\nabla_{AA'}\iota_{B}. \label{Ci}
\end{align}
\end{subequations}
Let $\tC_{A'}=o^{A}\C_{AA'}$, $\C_{A'}=\iota^{A}\C_{AA'}$. Notice that it always holds
\begin{equation}\label{HalfParallelSpinors}
 \C_{A'}o^{B}=0, \qquad \tC_{A'}\iota^{B} = 0.
\end{equation}
Acting on arbitrary spinor/tensor fields with weights $(p,w)$, the curvature of $\mathcal{C}_{a}$ can be decomposed as
\begin{equation}
 [\C_{a},\C_{b}] = [\nabla^{\rm f}_{a}, \nabla^{\rm f}_{b}] + w2\partial_{[a}f_{b]} + p2\partial_{[a}P_{b]}. 
 \label{DecompositionCurvatureC0}
\end{equation}
In particular, let $e^{a}_{\bf a}$ be a frame of vector fields with $w(e^{a}_{\bf a})=w$ and $p(e^{a}_{\bf a})=p$, 
and let $e^{\bf a}_{a}$ be the dual coframe. We then have:
\begin{equation}
 e^{\bf c}_{c}[\C_{a},\C_{b}] e^{d}_{\bf c} = R^{\rm f}_{abc}{}^{d} + 2(w\partial_{[a}f_{b]} + p\partial_{[a}P_{b]})\delta^{d}_{c},
\end{equation}
where $R^{\rm f}_{abc}{}^{d}$ is the curvature tensor of the Weyl connection $\nabla^{\rm f}_{a}$, see \eqref{CurvatureTensorWeylC}. 
We could define the left hand side as the curvature tensor of $\C_{a}$, ``$R^{\C}_{abc}{}^{d}$'', but this 
depends explicitly on the weights $(p,w)$ of the frame that is used for the definition.
However, the ``scalar curvature'', defined by the contraction with $g^{ac}\delta^{b}_{d}$, 
coincides with the scalar curvature of $R^{\rm f}_{abc}{}^{d}$ and so is independent of the frame:
\begin{equation}
 \mathcal{R} := e^{{\bf c}a}[\C_{a},\C_{b}] e^{b}_{\bf c} = R^{\rm f}. \label{ScalarCurvatureC}
\end{equation}
Recall from \eqref{ScalarCurvatureWeylC} that $R^{\rm f}=R+6(\nabla_{a}f^{a}+f_{a}f^{a})$.

\subsubsection*{The curvature operators $\Box^{\C}_{A'B'}$ and $\Box^{\C}_{AB}$}

The curvature $[\C_{a},\C_{b}]$ can also be decomposed into SD and ASD pieces:
\begin{equation}\label{SplittingCurvatureC}
 [\mathcal{C}_{a},\mathcal{C}_{b}] = \epsilon_{AB} \Box^{\mathcal{C}}_{A'B'} + \epsilon_{A'B'}\Box^{\mathcal{C}}_{AB}
\end{equation}
where $\Box^{\C}_{A'B'} = \tfrac{1}{2}\epsilon^{AB}[\C_{a},\C_{b}] = \C_{A(A'}\C^{A}_{B')}$ 
and $\Box^{\C}_{AB} = \tfrac{1}{2}\epsilon^{A'B'}[\C_{a},\C_{b}] = \C_{A'(A}\C^{A'}_{B)}$.
Performing an analogous SD-ASD decomposition for each term in the right hand side of \eqref{DecompositionCurvatureC0}, 
we get the following formulas:
\begin{subequations}
\begin{align}
 & \Box^{\C}_{A'B'} = \Box^{\rm f}_{A'B'}+w\nabla_{A(A'}f^{A}_{B')}+p\nabla_{A(A'}P^{A}_{B')}, \\
 & \Box^{\C}_{AB} = \Box^{\rm f}_{AB}+w\nabla_{A'(A}f^{A'}_{B)}+p\nabla_{A'(A}P^{A'}_{B)}.
\end{align}
\end{subequations}
Therefore, if $\psi_{C'...M'}$ and $\varphi_{C...M}$ have weights $(p,w)$ each, we have:
\begin{subequations}
\begin{align}
\nonumber  \Box^{\C}_{A'B'}\psi_{C'...M'} ={}& \tilde{X}^{\rm f}_{A'B'C'Q'}\psi^{Q'}{}_{D'...M'} + ... 
 + \tilde{X}^{\rm f}_{A'B'M'Q'}\psi_{C'...L'}{}^{Q'} \\
 & + \left[ w\nabla_{A(A'}f^{A}_{B')}+p\nabla_{A(A'}P^{A}_{B')} \right]\psi_{C'...M'},  \label{CurvatureCGeneral1} \\
\nonumber  \Box^{\C}_{AB}\psi_{C'...M'} ={}& \tilde{\Phi}^{\rm f}_{ABC'Q'}\psi^{Q'}{}_{D'...M'} + ... 
 + \tilde{\Phi}^{\rm f}_{ABM'Q'}\psi_{C'...L'}{}^{Q'} \\
 & + \left[ w\nabla_{A'(A}f^{A'}_{B)}+p\nabla_{A'(A}P^{A'}_{B)} \right]\psi_{C'...M'},  \label{CurvatureCGeneral2} \\
\nonumber  \Box^{\C}_{A'B'}\varphi_{C...M} ={}& \Phi^{\rm f}_{A'B'CQ}\varphi^{Q}{}_{D...M} + ... 
 + \Phi^{\rm f}_{A'B'MQ}\varphi_{C...L}{}^{Q} \\
 & + \left[ w\nabla_{A(A'}f^{A}_{B')}+p\nabla_{A(A'}P^{A}_{B')} \right]\varphi_{C...M},  \label{CurvatureCGeneral3} \\
\nonumber  \Box^{\C}_{AB}\varphi_{C...M} ={}& X^{\rm f}_{ABCQ}\varphi^{Q}{}_{D...M} + ... 
 + X^{\rm f}_{ABMQ}\varphi_{C...L}{}^{Q} \\
 & + \left[ w\nabla_{A'(A}f^{A'}_{B)}+p\nabla_{A'(A}P^{A'}_{B)} \right]\varphi_{C...M},  \label{CurvatureCGeneral4} 
\end{align}
\end{subequations}
where $\tilde{X}^{\rm f}_{A'B'C'D'}$, $\tilde{\Phi}^{\rm f}_{ABC'D'}$, $X^{\rm f}_{ABCD}$ and $\Phi^{\rm f}_{A'B'CD}$
are the curvature spinors of the Weyl connection $\nabla^{\rm f}_{a}$: 
in terms of the curvature of an arbitrary Levi-Civita connection in the conformal class, we have:
\begin{subequations}
\begin{align}
 & \tilde{X}^{\rm f}_{A'B'C'D'} = \tilde{X}_{A'B'C'D'} + \epsilon_{(A'|D'|}\left[ \nabla_{B')C}f^{C}_{C'}+f_{B')C}f^{C}_{C'} \right], \\
 & \tilde{\Phi}^{\rm f}_{ABC'D'} = \Phi_{ABC'D'} -\nabla_{D'(A}f_{B)C'} + f_{D'(A}f_{B)C'}, \\
 & X^{\rm f}_{ABCD} =  X_{ABCD} + \epsilon_{(A|D|}\left[ \nabla_{B)C'}f^{C'}_{C}+f_{B)C'}f^{C'}_{C} \right], \\
 & \Phi^{\rm f}_{A'B'CD} = \Phi_{A'B'CD} -\nabla_{D(A'}f_{B')C} + f_{D(A'}f_{B')C}.
\end{align}
\end{subequations}
Notice that these formulas are more general than the ones seen in appendix \ref{App:WeylConnections}, 
since here we do not assume that the 1-form $f_{a}$ is closed.

\smallskip
\noindent
Another useful expression for $\Box^{\mathcal{C}}_{A'B'}$ and $\Box^{\mathcal{C}}_{AB}$ is
in terms of $\tC_{A'}$, $\C_{A'}$: a short calculation gives
\begin{subequations}
\begin{align}
& \Box^{\mathcal{C}}_{A'B'} = \tilde{\mathcal{C}}_{(A'}\mathcal{C}_{B')} - \mathcal{C}_{(A'}\tilde{\mathcal{C}}_{B')}, 
 \label{IdentitySDCC} \\
& \Box^{\mathcal{C}}_{AB}= o_{A}o_{B}(\mathcal{C}_{A'}\mathcal{C}^{A'}-\sigma^{1}_{A'}\tilde{\mathcal{C}}^{A'}) 
 +\iota_{A}\iota_{B}(\tilde{\mathcal{C}}_{A'}\tilde{\mathcal{C}}^{A'}+\sigma^{o}_{A'}\mathcal{C}^{A'}) 
-o_{(A}\iota_{B)}(\tilde{\mathcal{C}}_{A'}\mathcal{C}^{A'}+\mathcal{C}_{A'}\tilde{\mathcal{C}}^{A'}),
 \label{IdentityASDCC}
\end{align}
\end{subequations}

\smallskip
\noindent
Contractions of $\tC_{A'}$ and $\C_{A'}$ can be expressed in terms of curvature and 
the wave operator $\Box^{\C}=g^{ab}\C_{a}\C_{b}$:
\begin{subequations}
\begin{align}
 & \tilde{\mathcal{C}}^{A'}\tilde{\C}_{A'}=-\Box^{\C}_{00}-\sigma^{oA'}\C_{A'}, \label{tildeC-tildeC}  \\
 & \C^{A'}\tilde{\C}_{A'}=\tfrac{1}{2}\Box^{\C}-\Box^{\C}_{01}, \label{C-tildeC}  \\
 & \tilde\C^{A'}\C_{A'}=-\tfrac{1}{2}\Box^{\C}-\Box^{\C}_{01}, \label{tildeC-C} \\
 & \C^{A'}\C_{A'}=-\Box^{\C}_{11}+\sigma^{1A'}\tilde\C_{A'}. \label{C-C}
\end{align}
\end{subequations}
($\Box^{\C}_{00}$, etc. are defined according to \eqref{Notation01}). 
One can show that for a scalar field $\Phi$ with weights $(p,w)$, it holds:
\begin{equation}
 \C^{A'}\tilde{\C}_{A'}\Phi = \tfrac{1}{2}(\Box^{\C}-p\mathcal{R}/4)\Phi. \label{C-tC-Scalars}
\end{equation}

\smallskip
\noindent
If $\phi_{A'...K'}$ and $\varphi_{A...K}$ are totally symmetric, and with weights $p=0$, $w=-1$ each, then:
\begin{subequations}
\begin{align}
 \nabla_{A}{}^{A'}\phi_{A'...K'} = \mathcal{C}_{A}{}^{A'}\phi_{A'...K'}, \label{identityRHF} \\
 \nabla_{A'}{}^{A}\varphi_{A...K} = \mathcal{C}_{A'}{}^{A}\varphi_{A...K}. \label{identityLHF}
\end{align}
\end{subequations}

\subsubsection*{Yang-Mills}

The connection $\C_{a}$ is generalized to $\Cym_{a}$ when we study Yang-Mills fields, as we describe in 
section \ref{Sec:YangMills}.
The curvature $[\Cym_{a},\Cym_{b}]$ decomposes into SD and ASD parts analogously to \eqref{SplittingCurvatureC}, 
so one has the operators $\Box^{\Cym}_{A'B'}=\Cym_{A(A'}\Cym^{A}_{B')}$ and $\Box^{\Cym}_{AB}=\Cym_{A'(A}\Cym^{A'}_{B)}$.
For arbitrary $\mu_{i}{}^{j}$, we have:
\begin{subequations}
\begin{align}
 & \Box^{\Cym}_{A'B'}\mu_{i}{}^{j} = \Box^{\C}_{A'B'}\mu_{i}{}^{j} - \chi_{A'B'i}{}^{k}\mu_{k}{}^{j} + \chi_{A'B'k}{}^{j}\mu_{i}{}^{k}, 
 \label{SDBoxYM} \\
 & \Box^{\Cym}_{AB}\mu_{i}{}^{j} = \Box^{\C}_{AB}\mu_{i}{}^{j} - \varphi_{ABi}{}^{k}\mu_{k}{}^{j} + \varphi_{ABk}{}^{j}\mu_{i}{}^{k}.
 \label{ASDBoxYM}
\end{align}
\end{subequations}

\smallskip
\noindent
Similarly to \eqref{tildeC-tildeC}--\eqref{C-C}, we have the general identities
\begin{subequations}
\begin{align}
 & \tilde\Cym^{A'}\tilde\Cym_{A'}=-\Box^{\Cym}_{00}-\sigma^{oA'}\Cym_{A'}, \label{tCym-tCym}  \\
 & \Cym^{A'}\tilde{\Cym}_{A'}=\tfrac{1}{2}\Box^{\Cym}-\Box^{\Cym}_{01}, \label{Cym-tCym} \\
 & \tilde\Cym^{A'}\Cym_{A'}=-\tfrac{1}{2}\Box^{\Cym}-\Box^{\Cym}_{01}, \label{tCym-Cym} \\
 & \Cym^{A'}\Cym_{A'}=-\Box^{\Cym}_{11}+\sigma^{1A'}\tilde\Cym_{A'}. \label{Cym-Cym}
\end{align}
\end{subequations}

\subsubsection*{Some GHP expressions}

When acting on {\em scalar} fields with weights $(p,w)$, the wave operator $\Box^{\C}$ can be expressed in GHP notation as:
\begin{align}
\nonumber \Box^{\C} ={}& 2[(\tho-p\rho-\tilde\rho)(\tho'-\rho')-(\edt-p\t-\tilde\t')(\edt'-\t')]\\
\nonumber &+2(w+1)(\rho\tho'+\rho'\tho-\t\edt'-\t'\edt-2\L-\Psi_2) \\
& -2(w+1)(p-w)(\rho\rho'-\t\t')+3p\Psi_2+2(p-1)(\k\k'-\sigma\sigma'). \label{boxC1}
\end{align}
Alternatively, commuting the GHP operators: 
\begin{align}
\nonumber \Box^{\C} ={}& 2[(\tho'+p\rho'-\tilde\rho')(\tho-\rho)-(\edt'+p\t'-\tilde\t)(\edt-\t)]\\
\nonumber &-2(p-(w+1))(\rho'\tho+\rho\tho'-\t'\edt-\t\edt'-2\L-\Psi_2) \\
& +2(p-(w+1)(p-w))(\rho\rho'-\t\t')-3p\Psi_2-2(w+1)(\k\k'-\sigma\sigma'). \label{boxC2} 
\end{align}

\subsection{The case in which $o^{A}$ is shear-free and a repeated principal spinor}\label{App:SFR}

Here we specialize the general identities given in appendix \ref{App:curvatureC} to the case in 
which $o^{A}$ satisfies \eqref{SFR} and \eqref{rPND}.
In addition, the other spinor field $\iota^{A}$ in the definition of $\C_{a}$ is chosen so that 
the Lee form is a gradient (see proposition \ref{Prop:LeeForm}):
\begin{equation}
 f_{a} = \partial_{a}\log\phi. \label{fgradientAppendix}
\end{equation}
In the general curvature identities \eqref{CurvatureCGeneral1}--\eqref{CurvatureCGeneral4}, eq. 
\eqref{fgradientAppendix} implies $\tilde{\Phi}^{\rm f}_{ABC'D'}=\Phi^{\rm f}_{C'D'AB}$, and
\begin{equation}
 \nabla_{A(A'}f^{A}_{B')}=0=\nabla_{A'(A}f^{A'}_{B)}. \label{CurvatureOfLeeForm0}
\end{equation}

\subsubsection*{The operators $\Box^{\C}_{A'B'}$ and $\Box^{\C}_{AB}$}

The fact that $\C_{AA'}o_{B}=0$ allows to describe the curvature of $\mathcal{C}_{AA'}$ entirely in terms of the 
curvature spinors of the Weyl-connection $\nabla^{\rm f}_{a}$, as we now show.
Applying another $\mathcal{C}_{a}$ derivative to $\C_{b}o_{C}=0$ and taking the commutator:
\begin{subequations}
\begin{align}
 0 ={}& \Box^{\C}_{A'B'}o_{C} = \Phi^{\rm f}_{A'B'CD}o^{D} +(\nabla_{A(A'}P^{A}_{B')})o_{C}, \\
 0 ={}& \Box^{\C}_{AB}o_{C} = X^{\rm f}_{ABCD}o^{D} + (\nabla_{A'(A}P^{A'}_{B)})o_{C},
\end{align}
\end{subequations}
from where we deduce:
\begin{subequations}
\begin{align}
 \Phi^{\rm f}_{A'B'01} ={}& - \nabla_{A(A'}P^{A}_{B')}, \label{oiPhiw} \\
 X^{\rm f}_{AB01} ={}& -\nabla_{A'(A}P^{A'}_{B)}, \label{oiXw} \\
 \Phi^{\rm f}_{A'B'00} ={}& 0, \label{ooPhiw} \\
 X^{\rm f}_{AB00} ={}& 0, \label{ooXw}
\end{align}
\end{subequations}
and therefore, from \eqref{oiPhiw}-\eqref{oiXw}:
\begin{equation}
 2\partial_{[a}P_{b]} = -(\epsilon_{A'B'}X^{\rm f}_{AB01}+\epsilon_{AB}\Phi^{\rm f}_{A'B'01}). \label{CurvatureOfP}
\end{equation}
Furthermore, \eqref{ooXw} can be written as
\begin{equation}
 \Psi_{ABCD}o^{C}o^{D} = -2\Lambda^{\rm f}o_{A}o_{B},
\end{equation}
which implies, in particular:
\begin{equation}
 \Psi_{2} = -2\Lambda^{\rm f}. \label{Psi2Lambdaw}
\end{equation}
Thus, recalling that the `scalar curvature' \eqref{ScalarCurvatureC} of $\C_{a}$ 
is $\mathcal{R}=24\Lambda^{\rm f}$, we get:
\begin{equation}
 \mathcal{R} = -12\Psi_{2}. \label{ScalarCurvatureCSpecial}
\end{equation}

\medskip
\noindent
The action of the operators $\Box^{\mathcal{C}}_{A'B'}$ and $\Box^{\mathcal{C}}_{AB}$ is then:
\begin{subequations}
\begin{align}
\Box^{\C}_{A'B'}\psi_{C'...M'} ={}& \tilde{X}^{\rm f}_{A'B'C'Q'}\psi^{Q'}{}_{D'...M'} + ... 
 + \tilde{X}^{\rm f}_{A'B'M'Q'}\psi_{C'...L'}{}^{Q'} -p\Phi^{\rm f}_{A'B' 01}\psi_{C'...M'},  \label{BoxCSpecial1} \\
\Box^{\C}_{AB}\psi_{C'...M'} ={}& \Phi^{\rm f}_{ABC'Q'}\psi^{Q'}{}_{D'...M'} + ... 
 + \Phi^{\rm f}_{ABM'Q'}\psi_{C'...L'}{}^{Q'} -pX^{\rm f}_{AB01}\psi_{C'...M'},  \label{BoxCSpecial2}  \\
\Box^{\C}_{A'B'}\varphi_{C...M} ={}& \Phi^{\rm f}_{A'B'CQ}\varphi^{Q}{}_{D...M} + ... 
 + \Phi^{\rm f}_{A'B'MQ}\varphi_{C...L}{}^{Q} -p\Phi^{\rm f}_{A'B' 01}\varphi_{C...M},  \label{BoxCSpecial3} \\
\Box^{\C}_{AB}\varphi_{C...M} ={}& X^{\rm f}_{ABCQ}\varphi^{Q}{}_{D...M} + ... 
 + X^{\rm f}_{ABMQ}\varphi_{C...L}{}^{Q}  -pX^{\rm f}_{AB01}\varphi_{C...M}.  \label{BoxCSpecial4}
\end{align}
\end{subequations}

\subsubsection*{Commutators}

The commutator $[\tC_{A'},\C_{B'}]$ can be computed by noticing that
\begin{equation}
 [\tC_{A'},\C_{B'}] = o^{A}\iota^{B}[\C_{a},\C_{b}] = \Box^{\C}_{A'B'} + \epsilon_{A'B'}\Box^{\C}_{01}.
\end{equation}
Straightforward calculations then give:
\begin{subequations}
\begin{align}
\nonumber [\tC_{A'},\C_{B'}]\psi_{C'...M'} ={}& (\tilde{X}^{\rm f}_{A'B'C'Q'}+\epsilon_{A'B'}\Phi^{\rm f}_{01C'Q'})\psi^{Q'}{}_{D'...M'} \\
 & +...+(\tilde{X}^{\rm f}_{A'B'M'Q'}+\epsilon_{A'B'}\Phi^{\rm f}_{01M'Q'})\psi_{C'...L'}{}^{M'}
  -p(\Phi^{\rm f}_{B'A'01}+3\Lambda^{\rm f}\epsilon_{B'A'})\psi_{C'...M'}
 \label{CommutatorPS} \\
\nonumber [\tC_{A'},\C_{B'}]\varphi_{C...M} ={}& (\Phi^{\rm f}_{A'B'CQ}+\epsilon_{A'B'}X^{\rm f}_{01CQ})\varphi^{Q}{}_{D...M} \\
 & +...+(\Phi^{\rm f}_{A'B'MQ}+\epsilon_{A'B'}X^{\rm f}_{01MQ})\varphi_{C...L}{}^{M}
  -p(\Phi^{\rm f}_{B'A'01}+3\Lambda^{\rm f}\epsilon_{B'A'})\varphi_{C...M}.
 \label{CommutatorUS}
\end{align}
\end{subequations}
The contraction of \eqref{CommutatorPS} with $\epsilon^{B'M'}$ is also needed in this work:
\begin{align}
\nonumber [\tC_{A'},\C_{B'}]\psi_{C'...L'}{}^{B'} ={}& 
(\tilde{X}^{\rm f}_{A'B'C'Q'}+\epsilon_{A'B'}\Phi^{\rm f}_{01C'Q'})\psi^{Q'}{}_{D'...L'}{}^{B'} \\
\nonumber & +...+(\tilde{X}^{\rm f}_{A'B'L'Q'}+\epsilon_{A'B'}\Phi^{\rm f}_{01L'Q'})\psi_{C'...K'}{}^{Q'}{}_{M'} \\
 & -(p+1)(\Phi^{\rm f}_{B'A'01}+3\Lambda^{\rm f})\psi_{C'...L'}{}^{B'}. \label{ContractedCommutatorPS}
\end{align}
All these formulas can be checked by considering the simplest case of a spinor with only one index.

\medskip
\noindent
Taking into account that $\sigma^{1}_{A'}=\iota^{B}\mathcal{C}_{A'}\iota_{B}$, applying the identities above 
we find:
\begin{subequations}
\begin{align}
 \tC_{(A'}\sigma^{1}_{B')} ={}& \Phi^{\rm f}_{A'B'11}, \label{Phiwii} \\
 \tC_{A'}\sigma^{1 A'}  ={}& 2\Psi_{3}, \label{Psi3} \\
 \C_{A'}\sigma^{1 A'}  ={}& \Psi_{4}. \label{Psi4}
\end{align}
\end{subequations}

\subsubsection*{Bianchi identities}

The Bianchi identities are
\begin{subequations}
\begin{align}
 & \C_{B}{}^{A'}\tilde{X}^{\rm f}_{A'B'C'D'} = \C_{B'}{}^{A}\Phi^{\rm f}_{C'D'AB}, \label{BianchiXwt} \\
 & \C_{B'}{}^{A}X^{\rm f}_{ABCD} = \C_{B}{}^{A'}\Phi^{\rm f}_{CDA'B'}, \label{BianchiXw}
\end{align}
\end{subequations}
or in irreducible components:
\begin{subequations}
\begin{align}
 & \mathcal{C}_{B}{}^{A'}\tilde{\Psi}_{A'B'C'D'} = \mathcal{C}_{(B'}{}^{A}\Phi^{\rm f}_{C'D')AB}, \\
 & \mathcal{C}_{B'}{}^{A}\Psi_{ABCD} = \mathcal{C}_{(B}{}^{A'}\Phi^{\rm f}_{CD)A'B'}, \\
 & \mathcal{C}^{AC'}\Phi^{\rm f}_{ABC'D'} + 3\mathcal{C}_{BD'}\Lambda^{\rm f} = 0. \label{ContractedBianchiw}
\end{align}
\end{subequations}

\medskip
\noindent
Contracting \eqref{BianchiXwt} and \eqref{ContractedBianchiw} with $o^{B}$ and $\iota^{B}$, and taking into account \eqref{ooPhiw}:
\begin{subequations}
\begin{align}
 & \tC^{A'}\tilde{X}^{\rm f}_{A'B'C'D'} = -\tC_{B'}\Phi^{\rm f}_{C'D'01}, \label{BianchiXwto} \\
 & \C^{A'}\tilde{X}^{\rm f}_{A'B'C'D'}  = \C_{B'}\Phi^{\rm f}_{C'D'01} - \tC_{B'}\Phi^{\rm f}_{C'D'11}, \label{BianchiXwt1} \\
 & -\tC^{C'}\Phi^{\rm f}_{C'D'01} + 3\tC_{D'}\Lambda^{\rm f} = 0, \label{ContractedBianchiwo} \\
 & \C^{C'}\Phi^{\rm f}_{C'D'01}-\tC^{C'}\Phi^{\rm f}_{C'D'11}+3\C_{D'}\Lambda^{\rm f} = 0. \label{ContractedBianchiw1}
\end{align}
\end{subequations}
Eqs. \eqref{BianchiXwto} and \eqref{BianchiXwt1} can be expressed in terms of the Weyl spinor 
by taking the parts totally symmetric in $B'C'D'$.

\subsubsection*{Yang-Mills}

The Yang-Mills spinor curvature operators $\Box^{\Cym}_{A'B'}$ and $\Box^{\Cym}_{AB}$ can 
be expressed in terms of $\tilde\Cym_{A'}$, $\Cym_{A'}$ analogously to \eqref{IdentitySDCC}-\eqref{IdentityASDCC}. 
For the case that $\sigma^{o}_{A'}=0$ and $\tilde\Cym_{A'}\tilde\Cym^{A'}=0$, we have:
\begin{subequations}
\begin{align}
 \Box^{\Cym}_{A'B'} ={}& \tilde\Cym_{(A'}\Cym_{B')} - \Cym_{(A'}\tilde\Cym_{B')}, \label{BoxYMSD0} \\
 \Box^{\Cym}_{AB}={}& o_{A}o_{B}(\Cym_{A'}\Cym^{A'}-\sigma^{1}_{A'}\tilde\Cym^{A'}) 
 -o_{(A}\iota_{B)}(\tilde\Cym_{A'}\Cym^{A'}+\Cym_{A'}\tilde\Cym^{A'}), \label{BoxYMASD0}.
\end{align}
\end{subequations}
The explicit action on a field $\mu^{i}$ with weights $(p,w)$ is:
\begin{subequations}
\begin{align}
 \Box^{\Cym}_{A'B'}\mu^{i} ={}& -p\Phi^{\rm f}_{A'B'01}\mu^{i}+\chi_{A'B'j}{}^{i}\mu^{j}, 
 \label{BoxYMSD1}\\
 \Box^{\Cym}_{AB}\mu^{i} ={}& -pX^{\rm f}_{AB01}\mu^{i}+\varphi_{ABj}{}^{i}\mu^{j}.
 \label{BoxYMASD1}
\end{align}
\end{subequations}

\medskip
\noindent
The commutator \eqref{CommutatorPS} acting on a field $\mu_{C'i}{}^{j}$, with arbitrary weights $(p,w)$, is
\begin{align}
\nonumber [\tilde\Cym_{A'},\Cym_{B'}]\mu_{C'i}{}^{j} ={}& [\tC_{A'},\C_{B'}]\mu_{C'i}{}^{j}
 -\chi_{A'B'i}{}^{k}\mu_{C'k}{}^{j}+\chi_{A'B'k}{}^{j}\mu_{C'i}{}^{k} \\
  & +\epsilon_{A'B'} \left( - \varphi_{01i}{}^{k}\mu_{C'k}{}^{j} + \varphi_{01k}{}^{j}\mu_{C'i}{}^{k} \right).
\end{align}
Contracting with $\epsilon^{B'C'}$:
\begin{align}
\nonumber [\tilde\Cym_{A'},\Cym_{B'}]\mu^{B'}{}_{i}{}^{j} ={}& 
 -(p+1)\left( \Phi^{\rm f}_{01A'B'}\mu^{B'}{}_{i}{}^{j} +3\Lambda^{\rm f}\mu_{A'i}{}^{j} \right) \\
 & -\chi_{A'B'i}{}^{k}\mu^{B'}{}_{k}{}^{j} + \chi_{A'B'k}{}^{j}\mu^{B'}{}_{i}{}^{k} 
 + \varphi_{01i}{}^{k}\mu_{A'k}{}^{j} - \varphi_{01k}{}^{j}\mu_{A'i}{}^{k}. \label{IdentityCommutatorYM}
\end{align}

\subsubsection*{Einstein-Weyl spaces}

Recall that the Einstein-Weyl equations are $\Phi^{\rm w}_{ABC'D'}=0$.
Using the Bianchi identities \eqref{BianchiAuxWeylPsi}, this gives $\C_{B'}{}^{A}(\mr\Omega^{-1}\Psi_{ABCD})=0$. 
These can be expressed in components as follows (using Newman-Penrose notation for the Weyl scalars):
\begin{subequations}
\begin{align}
 & \tC_{B'}(\mr\Omega^{-1}\Psi_{2}) =0, \label{BianchiComponents1} \\
 & \C_{B'}(\mr\Omega^{-1}\Psi_{2})-\tC_{B'}(\mr\Omega^{-1}\Psi_{3}) = 0, \\
 & \C_{B'}(\mr\Omega^{-1}\Psi_{3})-3\sigma^{1}_{B'}\mr\Omega^{-1}\Psi_{2}-\tC_{B'}(\mr\Omega^{-1}\Psi_{4}) = 0.
\end{align}
\end{subequations}

\medskip
\noindent
Using \eqref{Psi2Lambdaw} and \eqref{BianchiComponents1}, 
we see that $\tC_{A'}(\mr\Omega^{-1}\Lambda^{\rm f})=0$, so
\begin{equation}
 \tC_{A'}\Lambda^{\rm f}=\mr\Omega^{-1}\Lambda^{\rm f}\tC_{A'}\mr\Omega.
\end{equation}
Recalling \eqref{CtCtOmega}, we get
\begin{equation}
 \tC_{A'}\tC_{B'}\Lambda^{\rm f}=0. \label{tCtCLambdaf}
\end{equation}

\section{Some details of calculations}\label{App:DetailsCalculations}

\subsection{The conformal HH equation \eqref{conformalHHeq}}\label{App:HH}

The commutator term in \eqref{EqWithCommutator3} is
\begin{equation}
[\mathcal{C}_{A'}, \tilde{\mathcal{C}}_{B'}\tilde{\mathcal{C}}_{C'}\tilde{\mathcal{C}}_{D'}]\tilde{\mathcal{C}}^{A'}\Phi =
-\tilde{\mathcal{C}}_{B'}\tilde{\mathcal{C}}_{C'}K^{1}_{D'}-\tilde{\mathcal{C}}_{B'}K^{2}_{C'D'}-K^{3}_{B'C'D'},
\end{equation}
where
\begin{align}
 K^{1}_{D'} ={}& [\tilde{\mathcal{C}}_{D'},\mathcal{C}_{A'}]\tilde{\mathcal{C}}^{A'}\Phi, \\
 K^{2}_{C'D'} ={}& [\tilde{\mathcal{C}}_{C'},\mathcal{C}_{A'}]\tilde{\mathcal{C}}_{D'}\tilde{\mathcal{C}}^{A'}\Phi, \\
 K^{3}_{B'C'D'} ={}& [\tilde{\mathcal{C}}_{B'},\mathcal{C}_{A'}]\tilde{\mathcal{C}}_{C'}\tilde{\mathcal{C}}_{D'}\tilde{\mathcal{C}}^{A'}\Phi.
\end{align}
Using identity \eqref{ContractedCommutatorPS} and taking into account that $p(\Phi)=-4$ and $p(\tilde{\mathcal{C}}_{A'})=1$, we get:
\begin{align}
 K^{1}_{D'} ={}& 2\Phi^{\rm f}_{01E'D'}\tC^{E'}\Phi+6\Lambda^{\rm f}\tC_{D'}\Phi, \\
 K^{2}_{C'D'} ={}& \tilde{\Psi}_{E'F'C'D'}\tC^{E'}\tC^{F'}\Phi + 2\Lambda^{\rm f}\tC_{C'}\tC_{D'}\Phi 
  + \epsilon_{C'D'}\Phi^{\rm f}_{01E'F'}\tC^{E'}\tC^{F'}\Phi, \\
 K^{3}_{B'C'D'} ={}&  2\tilde{\Psi}_{E'F'B'(C'}\tC_{D')}\tC^{E'}\tC^{F'}\Phi - 2\Lambda^{\rm f}\tC_{B'}\tC_{C'}\tC_{D'}\Phi 
  +2\Phi^{\rm f}_{01(C'}{}^{E'}\tC_{D')}\tC_{B'}\tC_{E'}\Phi
\end{align}
Replacing, we have:
\begin{equation}
\C_{A'}\tC_{B'}\tC_{C'}\tC_{D'}\tC^{A'}\Phi = \tC_{B'}\tC_{C'}\tC_{D'} \left( \C_{A'}\tC^{A'}\Phi-6\Lambda^{\rm f}\Phi\right) + N_{B'C'D'},
 \label{AuxConformalHHeq}
\end{equation}
where for convenience we defined
\begin{align}
\nonumber N_{B'C'D'} \equiv &{} -\tC_{(B'} \left( \tilde{\Psi}_{C'D')E'F'}\tC^{E'}\tC^{F'}\Phi \right) 
  -2\tilde{\Psi}_{E'F'(B'C'}\tC_{D')}\tC^{E'}\tC^{F'}\Phi \\
\nonumber & +6\tC_{(B'}\tC_{C'} \left( \Phi\tC_{D')}\Lambda^{\rm f} \right) -2(\tC_{(B'}\Lambda^{\rm f})(\tC_{C'}\tC_{D')}\Phi) \\
 &-2\tC_{(B'}\tC_{C'}\left( \Phi^{\rm f}_{D')E'01}\tC^{E'}\Phi \right) +2\Phi^{\rm f}_{E'(B' 01}\tC_{C'}\tC_{D')}\tC^{E'}\Phi.
\end{align}
The idea now is to express everything in terms of $\tilde\Psi_{A'B'C'D'}$ and its derivatives, by means of the Bianchi identities.
Using \eqref{tCtCLambdaf}, the terms with $\Lambda^{\rm f}$ simplify to $+4(\tC_{(B'}\Lambda^{\rm f})(\tC_{C'}\tC_{D')}\Phi)$, 
and upon expanding the two derivatives in the first term of the last line, the terms with $\Phi^{\rm f}_{A'B'01}$ are
$$-2(\tC_{(B'}\tC_{C'}\Phi^{\rm f}_{D')E'01})(\tC^{E'}\Phi) -4 (\tC_{(B'}\Phi^{\rm f}_{C'|E'|01})(\tC_{D')}\tC^{E'}\Phi).$$
The Bianchi identity \eqref{BianchiXwto} can be written as
\begin{equation}
 \tC^{A'}\tilde{\Psi}_{A'B'C'D'} = -\tC_{B'}\Phi^{\rm f}_{C'D'01}-2\epsilon_{B'(C'}\tC_{D')}\Lambda^{\rm f}. 
\end{equation}
Taking another derivative and symmetrizing: $\tC_{(A'}\tC_{B'}\Phi^{\rm f}_{C')D'01} = -\tC_{(A'}\tC^{E'}\tilde{\Psi}_{B'C')D'E'}$. 
Therefore, after replacing the resulting expressions for $\tC_{(B'}\tC_{C'}\Phi^{\rm f}_{D')E'01}$ and $\tC_{B'}\Phi^{\rm f}_{C'E'01}$,
we arrive at
\begin{align*}
 N_{B'C'D'} = &{} -\tC_{(B'} \left( \tilde{\Psi}_{C'D')E'F'}\tC^{E'}\tC^{F'}\Phi \right) -2\tilde{\Psi}_{E'F'(B'C'}\tC_{D')}\tC^{E'}\tC^{F'}\Phi \\
  & + 2(\tC_{(B'}\tC^{F'}\tilde{\Psi}_{C'D')E'F'})(\tC^{E'}\Phi) + 4(\tC^{F'}\tilde{\Psi}_{E'F'(B'C'})(\tC_{D')}\tC^{E'}\Phi).
\end{align*}
Now, we have:
\begin{align*}
 \tilde{\mathcal{C}}_{B'}\tilde{\mathcal{C}}_{C'}\tilde{\mathcal{C}}_{D'} 
 \left[ \mr\Omega(\tilde{\mathcal{C}}_{E'}\tilde{\mathcal{C}}_{F'}\Phi)(\tilde{\mathcal{C}}^{E'}\tilde{\mathcal{C}}^{F'}\Phi) \right] 
 ={}& 2\tC_{(B'}(\tilde{\Psi}_{C'D')E'F'}\tC^{E'}\tC^{F'}\Phi)+4\tilde{\Psi}_{E'F'(B'C'}\tC_{D')}\tC^{E'}\tC^{F'}\Phi \\
 & +4\tilde{\Psi}_{E'F'(C'D'}(\mr\Omega^{-1}\tC_{D')}\mr\Omega)(\tC^{E'}\tC^{F'}\Phi) \\
 & +6(\tC_{(B'}\mr\Omega)(\tC_{C'}\tC^{E'}\tC^{F'}\Phi)(\tC_{D')}\tC_{E'}\tC_{F'}\Phi),
\end{align*}
and 
\begin{align*}
 \tC_{B'}\tC_{C'}\tC_{D'} \left[ (\tC^{E'}\mr\Omega)(\tC^{F'}\Phi)(\tC_{E'}\tC_{F'}\Phi) \right] ={}& 
 \tilde{\Psi}_{E'F'(B'C'}(\mr\Omega^{-1}\tC_{D')}\mr\Omega)(\tC^{E'}\tC^{F'}\Phi) \\
 & +\tfrac{3}{2}(\tC_{(B'}\mr\Omega)(\tC_{C'}\tC^{E'}\tC^{F'}\Phi)(\tC_{D')}\tC_{E'}\tC_{F'}\Phi) \\
 & +(\tC_{(B'}\tC^{F'}\tilde{\Psi}_{C'D')E'F'})(\tC^{E'}\Phi) \\
 & +2 (\tC^{E'}\tilde{\Psi}_{E'F'(B'C'})(\tC_{D')}\tC^{F'}\Phi).
\end{align*}
Therefore:
\begin{equation}
 N_{B'C'D'} = \tC_{B'}\tC_{C'}\tC_{D'} \left[ -\tfrac{1}{2}\mr\Omega(\tC_{E'}\tC_{F'}\Phi)(\tC^{E'}\tC^{F'}\Phi) 
+2(\tC^{E'}\mr\Omega)(\tC^{F'}\Phi)(\tC_{E'}\tC_{F'}\Phi) \right].
\end{equation}
Thus, the right hand side of \eqref{AuxConformalHHeq} is $\tC_{B'}\tC_{C'}\tC_{D'}(-K)$, where
\begin{equation}
 -K \equiv \C_{A'}\tC^{A'}\Phi-6\Lambda^{\rm f}\Phi-\tfrac{1}{2}\mr\Omega(\tC_{E'}\tC_{F'}\Phi)(\tC^{E'}\tC^{F'}\Phi) 
+2(\tC^{E'}\mr\Omega)(\tC^{F'}\Phi)(\tC_{E'}\tC_{F'}\Phi).
\end{equation}
Replacing the expressions for $\C_{A'}\tC^{A'}$ and $\Lambda^{\rm f}$, the conformal HH equation \eqref{conformalHHeq}
is obtained.

\subsection{Proof of \eqref{CommutatorForSigma1}}\label{App:Sigma1}

We have:
\begin{align}
 \C_{A'}\tC_{B'}\tC_{C'}\theta_{D'}{}^{A'} = \tC_{B'}\tC_{C'} \C_{A'}\theta_{D'}{}^{A'}
 -\tC_{B'}[\tC_{C'},\C_{A'}]\theta_{D'}{}^{A'} - [\tC_{B'}, \C_{A'}]\tC_{C'}\theta_{D'}{}^{A'},
 \label{AuxCommutatorSigma1}
\end{align}
The commutators can be computed
using \eqref{ContractedCommutatorPS}, \eqref{IdCurvC1}, \eqref{IdCurvC3}, and $p(\theta_{D'}{}^{A'})=-2$:
\begin{align*}
 [\tC_{C'},\C_{A'}]\theta_{D'}{}^{A'} ={}& (\tC_{C'}\tilde{\gamma}_{F'D'E'})\theta^{E'F'} + (\tC_{C'}\gamma_{E'01})\theta_{D'}{}^{E'}, \\
 [\tC_{B'},\C_{A'}]\tC_{C'}\theta_{D'}{}^{A'} ={}& (\tC_{B'}\tilde{\gamma}_{E'C'F'})(\tC^{F'}\theta_{D'}{}^{E'})
 +(\tC_{B'}\tilde{\gamma}_{E'D'F'})(\tC_{C'}\theta^{F'E'}).
\end{align*}
Using the identities $\tilde{\gamma}_{F'D'E'}=\epsilon_{F'D'}\tilde{\gamma}_{A'E'}{}^{A'}+\tilde{\gamma}_{D'F'E'}$ 
and $\tC_{C'}\tilde{\gamma}_{A'E'}{}^{A'} = -\tC_{C'}\gamma_{E'01}$, we get
\begin{align*}
 \tC_{B'}[\tC_{C'},\C_{A'}]\theta_{D'}{}^{A'} ={}& \tC_{B'}[(\tC_{C'}\tilde{\gamma}_{D'E'F'})\theta^{E'F'}] \\
 ={}& \tC_{B'}\tC_{C'}(\tilde{\gamma}_{D'E'F'}\theta^{E'F'})-\tC_{B'}(\tilde{\gamma}_{D'E'F'}\tC_{C'}\theta^{E'F'}).
\end{align*}
Then the terms with commutators in \eqref{AuxCommutatorSigma1} are:
\begin{align}
\tC_{B'}[\tC_{C'},\C_{A'}]\theta_{D'}{}^{A'} + [\tC_{B'}, \C_{A'}]\tC_{C'}\theta_{D'}{}^{A'} =
 \tC_{B'}\tC_{C'}(\tilde{\gamma}_{D'E'F'}\theta^{E'F'}) + Z_{B'C'D'}, 
 \label{AuxCommutatorSigma1-2}
\end{align}
where 
\begin{align*}
 Z_{B'C'D'} ={}& -(\tC_{B'}\tilde{\gamma}_{D'E'F'})(\tC_{C'}\theta^{E'F'}) - \tilde{\gamma}_{D'E'F'}\tC_{B'}\tC_{C'}\theta^{E'F'}\\
 & +(\tC_{B'}\tilde{\gamma}_{E'C'F'})(\tC^{F'}\theta_{D'}{}^{E'}) + (\tC_{B'}\tilde{\gamma}_{E'D'F'})(\tC_{C'}\theta^{E'F'}).
\end{align*}
Using now the fact that $\tilde{\gamma}_{D'E'F'}=2\tC_{(E'}\theta_{F')D'}$ and replacing in $Z_{B'C'D'}$, 
a tedious but straightforward calculation leads to $Z_{B'C'D'}\equiv 0$.
Replacing then \eqref{AuxCommutatorSigma1-2} into \eqref{AuxCommutatorSigma1},
equation \eqref{CommutatorForSigma1} follows.


\begin{thebibliography}{99}

\bibitem{Penrose1976} 
  R.~Penrose,
  {\em Nonlinear Gravitons and Curved Twistor Theory},
  Gen.\ Rel.\ Grav.\  {\bf 7}, 31 (1976).

\bibitem{MW} 
  L.~J.~Mason and N.~M.~J.~Woodhouse,
  {\em Integrability, selfduality, and twistor theory},
  Oxford, UK: Clarendon (1996) 364 p. (London Mathematical Society monographs, new series: 15)

\bibitem{Plebanski1975}
 J.~F.~Plebanski,
 {\em Some solutions of complex Einstein equations},
 J. Math. Phys. \textbf{16} (1975), 2395-2402

\bibitem{Dunajski} 
  M.~Dunajski,
  {\em Solitons, instantons, and twistors}
  (Oxford graduate texts in mathematics. 19)

\bibitem{PlebanskiRobinson} 
  J.~F.~Plebanski and I.~Robinson,
  {\em Left-Degenerate Vacuum Metrics},
  Phys.\ Rev.\ Lett.\  {\bf 37}, 493 (1976).

\bibitem{Gaiotto}
 D.~Gaiotto, G.~W.~Moore and A.~Neitzke,
 {\em Four-dimensional wall-crossing via three-dimensional field theory},
 Commun. Math. Phys. \textbf{299} (2010), 163-224
 \href{https://arxiv.org/abs/0807.4723}{[arXiv:0807.4723 [hep-th]]}

\bibitem{Bridgeland}
 T.~Bridgeland, and I.~A.~B.~Strachan, 
 {\em Complex hyperk\"ahler structures defined by Donaldson-Thomas invariants}, 
 Lett. Math. Phys. \textbf{111}, 54 (2021) 
 \href{https://arxiv.org/abs/2006.13059}{[arXiv:2006.13059 [math.AG]]}

\bibitem{Dunajski2020}
 M.~Dunajski,
 {\em Null K\"ahler geometry and isomonodromic deformations},
 \href{https://arxiv.org/abs/2010.11216}{[arXiv:2010.11216 [math.DG]]}.

\bibitem{Adamo}
T.~Adamo, L.~Mason and A.~Sharma,
 {\em Twistor sigma models for quaternionic geometry and graviton scattering},
 \href{https://arxiv.org/abs/2103.16984}{[arXiv:2103.16984 [hep-th]]}.

\bibitem{PlebanskiFinley} 
 J.~D.~Finley, III and J.~F.~Plebanski,
 {\em The Intrinsic Spinorial Structure of Hyperheavens},
 J.\ Math.\ Phys.\  {\bf 17}, 2207 (1976).

\bibitem{Boyer} 
  C.~P.~Boyer, J.~D.~Finley, III and J.~F.~Plebanski,
  {\em Complex General Relativity, H And HH Spaces: A Survey Of One Approach},
  in Held, A., editor, General Relativity and Gravitation, Vol. II, pages 241-281. Plenum, New York

\bibitem{Calderbank}
 D.~Calderbank and H.~Pedersen, 
 {\em Einstein-Weyl geometry}, 
 Surveys in Differential Geometry, 6(1), pp.387-423, 2001

\bibitem{Lawson}
Lawson, H. Blaine, and Marie-Louise Michelsohn,
{\em Spin Geometry} 
(PMS-38), Volume 38. Princeton university press, 2016.

\bibitem{Atiyah}
 M.~F.~Atiyah, N.~J.~Hitchin and I.~M.~Singer,
 {\em Self-duality in Four-Dimensional Riemannian Geometry},
 Proc. Roy. Soc. Lond. A \textbf{362}, 425-461 (1978)

\bibitem{Araneda2021}
B.~Araneda,
{\em Double field theory, twistors, and integrability in 4-manifolds},
\href{https://arxiv.org/abs/2106.01094}{[arXiv:2106.01094 [hep-th]]}.

\bibitem{Flaherty}
 E.~J.~Flaherty Jr., 
 {\em Hermitian and K\"ahlerian Geometry in Relativity}, 
 Springer Lecture Notes in Physics, Vol. 46 (Springer-Verlag, New York, 1976)

\bibitem{Flaherty2} 
 E.~J.~Flaherty Jr., 
 {\em Complex Variables in Relativity}, in
 General Relativity and Gravitation. Vol. 2. One hundred years after the birth of Albert Einstein. 
 Edited by A. Held. New York, NY: Plenum Press, p.207, 1980

\bibitem{PR1}
 R.~Penrose and W.~Rindler,
  {\em Spinors And Space-time. 1. Two Spinor Calculus And Relativistic Fields},
  Cambridge, Uk: Univ. Pr. (1984) 458 P. (Cambridge Monographs On Mathematical Physics)

\bibitem{PR2}
  R.~Penrose and W.~Rindler,
  {\em Spinors And Space-time. Vol. 2: Spinor And Twistor Methods In Space-time Geometry},
  Cambridge, Uk: Univ. Pr. (1986) 501p

\bibitem{TdC}
 G.~F.~Torres del Castillo, 
 {\em Null strings and Bianchi identities}, 
 J. Math. Phys. \textbf{24}, 590 (1983)

\bibitem{Hickman}
 M.~.S~Hickman and C.~B.~G~McIntosh, 
 {\em Complex relativity and real solutions. IV. Perturbations of vacuum Kerr-Schild spaces}, 
 General Relativity and Gravitation, Volume 18, Issue 12, pp.1275-1290 (1986)

\bibitem{Rozga}
 K.~R\'ozga, 
 {\em Double Kerr-Schild equivalence and hyperheavens}, 
 J. Math. Phys. \textbf{28}, 1107 (1987)

\bibitem{Jeffryes}
 B.~P.~Jeffryes, 
 {\em Half-algebraically special space-times, field equations and potentials},
 Max Planck Institut f\"ur Astrophysik Report, pp 12-101, vol 263.

\bibitem{Bailey1}
  T.~N.~Bailey, 
  {\em Complexified conformal almost-Hermitian structures and the conformally invariant eth and thorn operators},
  Class.\ Quant.\ Grav.\  {\bf 8}, no. 1, (1991)

\bibitem{Bailey2}
 T.~N.~Bailey, 
 {\em The space of leaves of a shear--free congruence, multipole expansions, and Robinson's theorem},
 J. Math. Phys. {\bf 32}, 1465 (1991)

\bibitem{Penrose2DTa}
 R.~Penrose,
 {\em Five things you can do with a surface in $\mathbb{PT}$}, 
 Twistor Newsletter \textbf{42}, pp. 1--5, 1997

\bibitem{Penrose2DTb}
 R.~Penrose,
 {\em  Twistor geometry of light rays},
 Class. Quant. Grav. \textbf{14} (1997), A299-A323

\bibitem{Adamo17} 
  T.~Adamo,
  {\em Lectures on twistor theory},
  PoS Modave {\bf 2017}, 003 (2018)
  \href{https://arxiv.org/abs/1712.02196}{[arXiv:1712.02196 [hep-th]]}.

\bibitem{Woodhouse}
 N.~M.~J.~Woodhouse, 
 {\em Real methods in twistor theory}, 
 Class. Quant. Grav. \textbf{2} (1985), 257-291.

\bibitem{Nurowski}
 P.~Nurowski and A.~Trautman, 
 {\em Robinson manifolds as the Lorentzian analogs of Hermite manifolds}, 
 Differential Geometry and its Applications, \textbf{17} (2002), 175-195 
 \href{https://arxiv.org/abs/math/0201266}{[arXiv:math/0201266 [math.DG]]}

\bibitem{TCh1}
 A.~Taghavi-Chabert,
{\em Optical structures, algebraically special spacetimes, and the Goldberg-Sachs theorem in five dimensions},
 Class. Quant. Grav. \textbf{28} (2011), 145010
 \href{https://arxiv.org/abs/1011.6168}{[arXiv:1011.6168 [gr-qc]]}.

\bibitem{TCh2}
 A.~Taghavi-Chabert,
 {\em Pure spinors, intrinsic torsion and curvature in even dimensions}, 
 Differ. Geom. Appl. \textbf{46} (2016), 164-203
 \href{https://arxiv.org/abs/1212.3595}{[arXiv:1212.3595 [math.DG]]}.

\bibitem{HCLee}
 H.~C.~Lee, 
 {\em A kind of even-dimensional differential geometry and its application to exterior calculus}, 
 Amer. J. Math. 65, 433-438 (1943)

\bibitem{Araneda18}
B.~Araneda,
 {\em Conformal invariance, complex structures and the Teukolsky connection},
 Class. Quant. Grav. \textbf{35}, no.17, 175001 (2018)
 \href{https://arxiv.org/abs/1805.11600}{[arXiv:1805.11600 [gr-qc]]}.

\bibitem{Araneda20}
 B.~Araneda, 
 {\em Two-dimensional twistor manifolds and Teukolsky operators}, 
 Lett. Math. Phys. \textbf{110}, no. 10, 2603-2638 (2020), 
 \href{https://arxiv.org/abs/1907.02507}{[arXiv:1907.02507 [gr-qc]]}.

\bibitem{Gualtieri}
 M.~Gualtieri, 
 {\em Generalized complex geometry}, PhD thesis, Oxford U., 2003. 
 \href{https://arxiv.org/abs/math/0401221}{arXiv:math/0401221 [math-dg]}.

\bibitem{Hollands2020}
S.~Hollands and V.~Toomani,
{\em On the radiation gauge for spin-1 perturbations in Kerr-Newman spacetime},
Class. Quant. Grav. \textbf{38} (2021) no.2, 025013
\href{https://arxiv.org/abs/2008.08550}{[arXiv:2008.08550 [gr-qc]]}.

\bibitem{WW}
 R. S. Ward and R. O. Wells,
 {\em Twistor Geometry and Field Theory}, 
 Cambridge University Press (1991)

\bibitem{DunajskiMason}
 M.~Dunajski and L.~J.~Mason,
 {\em HyperKahler hierarchies and their twistor theory},
 Commun. Math. Phys. \textbf{213} (2000), 641-672
 \href{https://arxiv.org/abs/math/0001008}{[arXiv:math/0001008 [math.DG]]}.

\bibitem{Tod}
 P.~Tod,
 {\em One-sided type-D Ricci-flat metrics},
 \href{https://arxiv.org/abs/2003.03234}{[arXiv:2003.03234 [gr-qc]]}.

\bibitem{DP}
 M.~Dunajski and M.~Przanowski,
 {\em Null K\"ahler structures, Symmetries and Integrability},
 in: Conference on Topics in Mathematical Physics, General Relativity, and Cosmology on 
 the Occasion of the 75th Birthday of Jerzy F. Plebanski, 147-155,
 \href{https://arxiv.org/abs/gr-qc/0310005}{[arXiv:gr-qc/0310005 [gr-qc]]}
 

\end{thebibliography}
\end{document}